\documentclass[sigconf,nonacm]{acmart}
\AtBeginDocument{%
  }

\setcopyright{acmlicensed}
\copyrightyear{2026}
\acmYear{2026}
\acmDOI{XXXXXXX.XXXXXXX}

\acmConference[CCS '26]{the 2026 ACM SIGSAC Conference on Computer and Communications Security}{November 15--19,
  2026}{The Hague, The Netherlands}

\acmISBN{978-1-4503-XXXX-X/2026/11}

\usepackage{amsmath}
\usepackage{amsthm}
\usepackage{mathtools}
\usepackage{bbm}
\usepackage{enumitem}
\usepackage{hyperref}
\usepackage{physics}
\usepackage[mathscr]{euscript}
\usepackage{eurosym}
\usepackage{color}
\usepackage[ruled]{algorithm}
\usepackage[noend]{algpseudocode}
\usepackage{caption}
\usepackage{cryptocode}
\usepackage{marginnote}
\usepackage[colorinlistoftodos]{todonotes}
\usepackage{seqsplit}
\usepackage{footmisc}

\newcommand{\No}{\mathcal{N}}

\renewcommand{\L}{\mathcal{A}}

\renewcommand{\P}{\mathcal{P}}

\newcommand{\tv}{t_u}

\newcommand{\operator}{O}

\newcommand{\ark}[1]{\Pi_{\mathsf{ark}}(#1)}
\newcommand{\arkff}[1]{\Pi^{\ff}_{\mathsf{ark}}(#1)}
\newcommand{\btc}{\Pi_{\mathsf{BTC}}}
\newcommand{\C}[2]{\mathcal{L}_{#1}^{#2}}
\newcommand{\tx}[2]{\texttt{tx}_{#1}^{\texttt{#2}}}
\newcommand{\val}{\texttt{value}}
\newcommand{\lockScript}{\texttt{lckScr}}
\newcommand{\vtxoLockScript}{\texttt{vtxoLckScr}}

\newcommand{\out}[2]{\texttt{out}_{#1}^{\texttt{#2}}}

\newcommand{\outputs}{\texttt{outs}}
\newcommand{\inputs}{\texttt{ins}}
\newcommand{\witnesses}{\texttt{wits}}

\newcommand{\True}{\texttt{True}}
\newcommand{\False}{\texttt{False}}
\newcommand{\checkSig}[1]{\texttt{chkSig}_{#1}}
\newcommand{\checkMultiSig}[1]{\texttt{chkMulSig}_{#1}}
\newcommand{\relTimelock}[1]{\texttt{relTlk}\qty(#1)}
\newcommand{\absTimelock}[1]{\texttt{absTlk}\qty(#1)}

\newcommand{\vtxo}[1]{\texttt{vtxo}_{#1}}

\newcommand{\vtxos}{\texttt{vtxos}}
\newcommand{\utxos}{\texttt{utxos}}
\newcommand{\batch}{\texttt{batch}}
\newcommand{\vtxt}{\texttt{vtxt}}
\newcommand{\roott}{\texttt{root}}
\newcommand{\patht}{\texttt{path}}
\newcommand{\leaves}{\texttt{leaves}}
\newcommand{\commitmentTx}{\texttt{commitmentTx}}
\newcommand{\verifyCommitmentTx}{\texttt{verifyCommitTx}}
\newcommand{\verifyConnector}{\texttt{verifyConnector}}

\newcommand{\forfeitTx}{\texttt{forfeitTx}}
\newcommand{\batchTemplate}{\texttt{batchTmp}}
\newcommand{\signerTemplate}{\texttt{signerTmp}}
\newcommand{\verifyPath}{\texttt{verifyPath}}
\newcommand{\st}{\texttt{st}}
\newcommand{\connectorTemplate}{\texttt{connectorTmp}}

\newcommand{\nott}{\textbf{not }}
\newcommand{\musig}{\texttt{musig}}

\newcommand{\new}{\texttt{new}}
\newcommand{\old}{\texttt{old}}
\newcommand{\boardings}{\texttt{boardings}}
\newcommand{\batchSwaps}{\texttt{batchSwaps}}
\newcommand{\exits}{\texttt{exits}}
\newcommand{\toBoard}{\texttt{toBoard}}
\newcommand{\toBatchSwap}{\texttt{toBatchSwap}}
\newcommand{\toExit}{\texttt{toExit}}
\newcommand{\preSpent}{\texttt{preSpent}}
\newcommand{\preConfirmed}{\texttt{preConfirmed}}

\newcommand{\confirmedVTXO}{\texttt{confirmedVTXO}}
\newcommand{\confirmedBatches}{\texttt{confirmedBatches}}
\newcommand{\sweep}{\texttt{sweep}}
\newcommand{\unconfirmed}{\texttt{unconfirmed}}
\newcommand{\unconfirmedSpent}{\texttt{unconfirmedSpent}}
\newcommand{\submitCommitTx}{\texttt{submitCommitTx}}

\newcommand{\unconfirmedToBoard}{\texttt{unconfirmedBoardings}}
\newcommand{\unconfirmedToBatchSwap}{\texttt{unconfirmedBatchSwaps}}
\newcommand{\unconfirmedToExit}{\texttt{unconfirmedExits}}
\newcommand{\confirmedToBoard}{\texttt{confirmedBoardings}}
\newcommand{\confirmedToBatchSwap}{\texttt{confirmedBatchSwaps}}
\newcommand{\confirmedToExit}{\texttt{confirmedExits}}
\newcommand{\expired}{\texttt{expired}}
\newcommand{\spent}{\texttt{spent}}
\newcommand{\replaced}{\texttt{replaced}}

\newcommand{\connectorScript}{\texttt{connectorScript}}

\newcommand{\sk}{sk}
\newcommand{\pk}{pk}

\newcommand{\Sign}{\mathsf{Sign}}

\newcommand{\arkstate}{\Sigma}

\newcommand{\Nff}{N^{\texttt{ff}}}
\newcommand{\ff}{\texttt{ff}}

\def\bitcoin{%
  \leavevmode
  \vtop{\offinterlineskip 
    \setbox0=\hbox{B}%
    \setbox2=\hbox to\wd0{\hfil\hskip-.03em
    \vrule height .3ex width .15ex\hskip .08em
    \vrule height .3ex width .15ex\hfil}
    \vbox{\copy2\box0}\box2}}

\newtheorem{theorem}{Theorem}[section]
\newtheorem{remark}[theorem]{Remark}
\newtheorem{definition}[theorem]{Definition}
\newtheorem{lemma}[theorem]{Lemma}
\newtheorem{corollary}[theorem]{Corollary}

\newcommand{\pim}[1]{\todo{Pim: #1}}
\newcommand{\za}[1]{}
\newcommand{\mm}[1]{}

\begin{document}

\title{Ark: Offchain Transaction Batching in Bitcoin}


\author{Pim Keer}
\affiliation{%
  \institution{TU Wien}
  \city{Vienna}
  \country{Austria}
}
\email{pim.keer@tuwien.ac.at}

\author{Ioannis Alexopoulos}
\affiliation{%
  \institution{TU Wien}
  \city{Vienna}
  \country{Austria}
}
\email{ioannis.alexopoulos@tuwien.ac.at}

\author{Matteo Maffei}
\affiliation{%
  \institution{TU Wien}
  \city{Vienna}
  \country{Austria}
}
\email{matteo.maffei@tuwien.ac.at}

\author{Marco Argentieri}
\affiliation{%
  \institution{Ark Labs}
  \city{Tallinn}
  \country{Estonia}
}
\email{marco@arklabs.to}

\author{Andrew Camilleri}
\affiliation{%
  \institution{Ark Labs}
  \city{Tallinn}
  \country{Estonia}
}
\email{andrew@arklabs.to}

\author{Zeta Avarikioti}
\affiliation{%
  \institution{TU Wien, Common Prefix}
  \city{Vienna}
  \country{Austria}
}
\email{georgia.avarikioti@tuwien.ac.at}


%


\begin{abstract}
Bitcoin is the cryptocurrency with the largest market capitalisation, but its widespread adoption is fundamentally limited by the scalability constraints of its consensus algorithm, which requires every transaction to be confirmed onchain. To address this, several Layer-2 scalability solutions have been proposed to move payments offchain---most notably, the Lightning Network. However, their deployment remains hindered by cumbersome setup requirements: users must lock funds onchain to participate and engage in complex auxiliary protocols (e.g., for channel rebalancing, top-ups, and routing). Other solutions, like payment pools, sidechains and rollups, cannot be implemented in a non-custodial way on Bitcoin due to its limited scripting capabilities, or require all protocol participants to update the offchain state. 

In this work, we present Ark, the first Bitcoin-compatible commit-chain. Ark enables offchain transactions of virtual UTXOs (VTXOs), through an untrusted operator who aggregates them into succinct onchain commitments. A distinctive feature of Ark is its ease of deployment: users can receive offchain payments without locking any funds beforehand and Ark state updates can be performed only requiring the users involved in that update.

We formally define the Ark protocol and prove its security. During this process, we identified two attacks affecting the testnet implementation, which we responsibly disclosed and proposed fixes for, which have been now integrated into the mainnet implementation. Our experimental evaluation demonstrates that Ark can commit onchain to batches of arbitrarily many VTXOs with a constant-sized footprint of approximately 200 vB. Cooperative exits add one output per user, while unilateral exits require $\mathcal{O}(\log n)$ transactions of roughly 150 vB per VTXO for a batch of $n$ VTXOs.
\end{abstract}


%





\maketitle

\section{Introduction} 
The security and decentralisation guarantees of the Bitcoin consensus protocol \cite{nakamoto2008bitcoin} come at the expense of severely limited throughput of approximately 10 transactions per second, which is about three orders of magnitude lower than conventional credit card systems that can process around 10k transactions per second. To overcome this limitation, a variety of Layer-2 protocols have been proposed, most notably payment channel networks \cite{poon2016bitcoin,decker2015fast,decker2018eltoo,avarikioti2019brick,hearn2023contract,avarikioti2020cerberus,aumayr2024bitcoin} and their extensions (such as payment channel hubs \cite{heilman2017tumblebit,dziembowski2019perun,tairi20212,qin2023blindhub}, virtual channels \cite{dziembowski2019multi,dziembowski2019perun,aumayr2021bitcoin,avarikioti2025thunderdome}, and channel factories \cite{burchert2018scalable,pedrosa2019scalable}).
However, these approaches share a fundamental drawback: they require a \emph{cumbersome and}, most importantly, \emph{costly setup}. Specifically,   users must lock collateral onchain to initially fund their channels as well as to engage later on in complex auxiliary protocols (e.g., for channel rebalancing~\cite{tiwari2022wiser,avarikioti2022hide,avarikioti2024musketeer}, channel top-ups, and payment routing~\cite{roos2018settling,avarikioti2024route,sivaraman2020high}). These usability and liquidity challenges have hindered the widespread deployment of such protocols.

In this work, we present Ark, the first Bitcoin-compatible commit-chain. A distinctive strength of Ark lies in its \emph{open architecture}. In contrast to payment channel solutions, users can receive and hold funds without issuing any onchain transaction, thereby lowering entry barriers. At the same time, funds deposited in Ark remain fully fungible: they can be transferred to anyone, Ark user or not, who may subsequently decide to use them within Ark or move them outside---for example, in order to make onchain Bitcoin payments or even to setup a Layer-2 protocol on top of Ark. This design eases onboarding, enhances interoperability, and mitigates liquidity fragmentation.

 At its core, Ark introduces the concept of \textbf{virtual UTXOs} (VTXOs), which serve as offchain counterparts of standard UTXOs. Ark allows users to transact offchain by drawing on shared liquidity managed by an operator. Each transaction updates the allocation of VTXOs, and the operator periodically commits state updates to the Bitcoin blockchain. Furthermore, Ark features \emph{unilateral exit}: each user may decide to bring their VTXOs onchain without the operator's collaboration, and without forcing any other Ark participant to exit. Additionally, Ark imposes \emph{minimal online requirements}: only the operator must remain continuously online, while users need to connect periodically to refresh their VTXOs; a process that occurs entirely offchain. Finally, in the presence of a rational operator and an existentially honest signer committee, Ark enables \emph{fast finality}: users willing to make these extra assumptions can opt in to spend their VTXOs immediately, without the risk of the transaction being reverted.


\subsection{Related work}
\label{subsec:related-work}
The comparison between Ark and other Layer-2 protocols is summarised in \autoref{tab:comparison} and discussed below. We omit custodial approaches like eCash \cite{chaum1983blind,cashu2025cashu,fedi2025fedi}.

\textbf{Payment channel networks (PCN).} PCNs, such as the Lightning Network~\cite{poon2016bitcoin}, are networks of 2-party channels, i.e., 2-of-2 multisignature UTXOs, where parties can cooperate to instantly update their channel balances, and where any party can at any time close the channel unilaterally and claim its funds onchain after a dispute period, in which the counterparty should come online to prevent a malicious user from exiting with an old state. PCNs therefore achieve \emph{fast finality}, i.e., transactions are instant and non-revertible, and \emph{unilateral exit}, i.e., users can exit their funds without the collaboration of another party. Closing, as well as opening a channel, has an onchain cost and requires to lock funds, which can only be used \emph{within} the network. Ark, on the other hand, is \emph{externally fungible}: it allows users who are not onboarded to receive funds with no onchain footprint. Moreover, PCNs require routing algorithms to find a path with sufficient capacity in the network, rebalancing algorithms to reallocate funds among channels, and possibly onchain top-up transactions. Payment channel hubs \cite{heilman2017tumblebit,dziembowski2019perun,tairi20212,qin2023blindhub}, virtual channels \cite{dziembowski2019multi,dziembowski2019perun,aumayr2021bitcoin,avarikioti2025thunderdome} and state channels \cite{coleman2015statechannels,miller2019sprites,aumayr2021generalized}, which can track arbitrary state, suffer from similar limitations. Channel factories~\cite{burchert2018scalable,pedrosa2019scalable} alleviate these issues by creating multiple channels at once, amongst which liquidity can be reallocated. However, one user leaving leads to other users being forced to leave as well. 

\textbf{Statechains.} Somsen~\cite{somsen2018statechains} proposed statechains to instantly transfer ownership of a UTXO to arbitrary parties with no onchain activity, and thus achieve fast finality and external fungibility. The owner can claim funds unilaterally after a dispute period. Spark \cite{spark2025spark} allows splitting the UTXO amongst multiple owners. However, statechains are only secure assuming a trusted operator that cosigns each transfer and deletes its key after each state update, preventing collusion with previous owners. Ark similarly enables operator-assisted transfers, but the operator can be malicious: we just need to assume their rationality for fast finality. 

\textbf{Payment pools.} The most developed example is CoinPool \cite{naumenkocoinpool}. It allows multiple users to share a single UTXO, make instant offchain transfers within this UTXO, and allows unilateral exit at any time by any user. Contrary to Ark, it assumes the introduction of several new Bitcoin opcodes (as of the current, post-Taproot, state of Bitcoin: it is thus not \emph{Bitcoin-compatible}), and has a high interactivity requirement since each state update requires each user's approval. 

\textbf{Sidechains.} Sidechains \cite{back2014enabling,kiayias2019proof,nick2020liquid} are independent blockchains, able to process transactions to arbitrary users with faster block times (but not necessarily with fast finality), or richer smart contract functionality, running in parallel to the main blockchain, enabling the transfer of assets between them via a two-way peg mechanism. In Bitcoin, this approach has to be custodial, with users sending funds to a trusted committee, requiring its collaboration to exit again. BitVM \cite{aumayr2024bitvm,linus2025bitvm2,woll2026bitvm3} could partially remove this trust. Regardless, sidechains reduce their security from that of Bitcoin to that of the sidechain's consensus mechanism. 

\textbf{Commit-chains.} Here, an untrusted operator coordinates transactions between its participants (not externally fungible) and periodically commits the state to the blockchain. Most commit-chains \cite{poon2017plasma,khalil2018commit} assume an expressive blockchain to handle user exits and disputes and are thus not Bitcoin-compatible. Commit-chains achieve delayed finality, i.e., the payment is only finalised once the corresponding commitment is confirmed onchain, and survived the dispute period, during which participants should come online to contest any misbehaviour. Assuming a malicious operator, Ark also achieves delayed finality (whereas a rational operator enables fast finality). Clique \cite{riahi2024clique} is a commit-chain handling simple payments on Bitcoin. Its original design requires a covenant functionality Bitcoin currently does not have, but could be modified to instead use signatures from all users to update the state. In contrast, Ark only requires signatures of users involved in a transaction and the operator to perform a state update and is fully Bitcoin-compatible. Contrary to other commit-chains, both Clique and Ark operators fund the onchain commitment with their own liquidity.  


\textbf{Rollups.} Rollups \cite{kalodner2018arbitrum,optimism2021optimism} offload computation offchain, but post all transaction data onchain. Just like commit-chains, rollups achieve unilateral exit and delayed finality, but are not externally fungible, and they are in general not Bitcoin-compatible. BitVM \cite{aumayr2024bitvm,linus2025bitvm2,woll2026bitvm3} can be leveraged to design Bitcoin-compatible rollups. However, participation relies on a committee that is honest and live at setup for liveness and assumes existential honesty throughout the protocol execution for safety Ark instead focuses on reducing its onchain footprint, having only a small transaction onchain committing to a state update. Ark only requires such a committee for its opt-in fast finality mechanism. Moreover, rollups require a data availability layer (e.g., the blockchain) to store all transaction data, whereas with Ark it is only necessary that each user stores its own local state to guarantee its capability to exit unilaterally.

\begin{table}[h]
\caption{Protocol comparison. Cols.: TA: trust assumptions; UE: unilateral exit; EF: external fungibility (no need for onchain interaction to receive funds); FF: fast finality (safely spend funds immediately, D = delayed until onchain commitment, $^*$ under rational operator and $1$-of-$n$ honest committee, otherwise D for Ark); SU: parties required per state update (except unilateral exit) (op. = operator); OL: operator liquidity; BC: Bitcoin-compatible; OR: users should come online each: (DP = Dispute period, OC = Onchain commitment).}
\centering
\label{tab:comparison}
\setlength{\tabcolsep}{1.9pt}

\begin{tabular}{|c|c|c|c|c|c|c|c|c|}
\hline
                                                          & TA                                                                    & UE           & EF           & FF             & SU                                                                & OL                                                     & BC           & OR                                                       \\ \hline
PCN                                                       &                                                                       & $\checkmark$ &              & $\checkmark$   & \begin{tabular}[c]{@{}c@{}}channel \\ parties\end{tabular}        &                                                        & $\checkmark$ & DP                                                       \\ \hline
\begin{tabular}[c]{@{}c@{}}State-\\ chain\end{tabular}    & \begin{tabular}[c]{@{}c@{}}Honest\\ op.\end{tabular}                  & $\checkmark$ & $\checkmark$ & $\checkmark$   & \begin{tabular}[c]{@{}c@{}}owner \\ + op.\end{tabular}            & 0                                                      & $\checkmark$ & DP                                                       \\ \hline
\begin{tabular}[c]{@{}c@{}}Payment\\ Pool\end{tabular}    &                                                                       & $\checkmark$ &              & $\checkmark$   & all                                                               &                                                        &              & DP                                                       \\ \hline
\begin{tabular}[c]{@{}c@{}}Side-\\ chain\end{tabular}     & \begin{tabular}[c]{@{}c@{}}Honest\\ vali-\\ dators\end{tabular}       &              & $\checkmark$ &                & \begin{tabular}[c]{@{}c@{}}vali- \\ dators\end{tabular}           &                                                        & $\checkmark$ & Never                                                    \\ \hline
\begin{tabular}[c]{@{}c@{}}Commit-\\ chain\end{tabular}   &                                                                       & $\checkmark$ &              & D              & op.                                                               & 0                                                      &              & OC                                                       \\ \hline
\begin{tabular}[c]{@{}c@{}}Bitcoin \\ Clique\end{tabular} &                                                                       & $\checkmark$ &              & D              & all                                                               & \begin{tabular}[c]{@{}c@{}}Clique\\ value\end{tabular} & $\checkmark$ & OC                                                       \\ \hline
Rollup                                                    &                                                                       & $\checkmark$ &              & D              & op.                                                               & 0                                                      &              & OC                                                       \\ \hline
Ark                                                       & \begin{tabular}[c]{@{}c@{}}Rational\\ op. only \\ for FF\end{tabular} & $\checkmark$ & $\checkmark$ & $\checkmark^*$ & \begin{tabular}[c]{@{}c@{}}op. +\\ involved \\ users\end{tabular} & \begin{tabular}[c]{@{}c@{}}Batch\\ value\end{tabular}  & $\checkmark$ & \begin{tabular}[c]{@{}c@{}}VTXO \\ lifetime\end{tabular} \\ \hline
\end{tabular}
\end{table}

\subsection{Our contributions}

The contributions of this work can be summarised as follows:

\begin{itemize}[leftmargin=*]
\item We present Ark, the first Bitcoin-compatible commit-chain, featuring unilateral exit, opt-in fast finality, and an open architecture, allowing receivers to join without performing any onchain transaction, and funds to be spent within as well as outside Ark.
\item We formalise and prove the security of Ark (VTXO security, user balance security, atomicity, operator balance security, and fast finality balance security). In doing so, we identified two attacks on the original Ark protocol, as implemented on testnet: the hostage attack and spam attack. We solve these attacks by introducing reset transactions, which have been integrated into the mainnet implementation.
\item  We conduct an experimental evaluation to demonstrate the performance and scalability gains. Specifically, we show that the onchain footprint for a commitment with arbitrarily many VTXOs is constant ($197$~vB) and for unilateral exit is logarithmic in the number $n$ of VTXOs per commitment ($\lceil\log n\rceil \cdot 150 + 107$~vB), which is a reasonable price to pay to obtain the functionality of a commit-chain in Bitcoin. Finally, we show the time required to produce a commitment is linear in the number of users involved in the commitment. For example, it requires a mere 2.7 seconds for 200 users. 
\end{itemize}


\section{Model and Protocol Overview}
\label{sec:overview}

\subsection{Model and Notation}
\label{subsec:system-model}
\smallskip\noindent\textbf{System Model.}
Ark is a Bitcoin-native transaction batching protocol that lets an arbitrary set of users $\mathcal{P}=\{P_1,P_2,\ldots\}$ transact offchain with a minimal onchain footprint, using an \emph{untrusted Ark operator $\operator$}. Users hold virtual balances (represented by \emph{VTXOs}) and issue Ark requests, such as boarding (onramp), transactions, batch swaps, and exits (offramp). Users can moreover claim their funds onchain at any time. $\operator$ aggregates requests and \emph{batches} them into onchain \emph{commitment transactions}, which are approved by the involved users in signing sessions coordinated by $\operator$, and are predominantly funded by $\operator$.



\smallskip\noindent\textbf{Blockchain and UTXO Model.}
\label{subsec:ledger-model}
We assume a robust public transaction ledger $\L$ (Bitcoin) that operates in the UTXO model. The ledger $\L$ consists of a sequence of transactions. 

\smallskip\noindent\textit{Transactions and Script.} 
In the UTXO model, each transaction $\tx{}{}$ maps a non-empty list of existing unspent transaction outputs (UTXOs) to a list of new UTXOs. A transaction is a tuple $\tx{}{} = (\inputs,\witnesses,\outputs)$, 
where $\inputs=[\out{1}{},\ldots,\out{n}{}]$ references unspent outputs, $\witnesses=[w_1,\ldots,w_n]$ provides unlocking data (witnesses), and $\outputs=[\out{1}{*},\ldots,\out{m}{*}]$ are new outputs. Each output $\out{}{}=(\val,\lockScript)$ holds an amount $\out{}{}.\val$, locked by a locking script $\out{}{}.\lockScript$. A transaction is \emph{valid} if every input is unspent, each script accepts its respective witness, i.e., $\out{i}{}.\lockScript(w_i)=\True$, and if no new coins are created, i.e., $\sum_{i=1}^{n}\out{i}{}.\val \;\ge\; \sum_{j=1}^{m}\out{j}{*}.\val$.

The locking scripts are expressed in the stack-based Bitcoin scripting language.  We use:
(i) signature checks $\checkSig{\pk}$ that require a valid signature under $\pk$,
(ii) $p$-of-$q$ multisignature checks $\checkMultiSig{p,q;\pk_1,\ldots,\pk_n}$,
(iii) absolute timelocks $\absTimelock{T}$ and relative timelocks $\relTimelock{t}$.
Taproot allows key-path spends with Schnorr signatures or script-path spends by revealing a single committed script together with a Merkle proof. We denote a Taproot output by $\texttt{Taproot}(\pk_I;\texttt{scriptPath}_1,\ldots,\texttt{scriptPath}_n)$, where $\pk_I$ is the internal key used together with the root of a Merkle tree built over the script paths to determine the ``tweaked'' public key, which is shown in the actual output script. This output can then be spent by only providing a signature under the tweaked public key, or by providing the witness for a script path together with its Merkle inclusion proof. We set $\pk_I$ to $\False$ when the key path is disabled.

\smallskip\noindent\textit{Ledger Model.}
We adopt the blockchain (ledger) model of~\cite{garay2024bitcoin} to which parties can submit transactions. In this model, each participant of the blockchain protocol (miner) $M$ maintains a local view of the ledger $\C{}{M}$; write $\C{-k}{M}$ for $\C{}{M}$ without the last $k$ blocks. We assume our ledger is robust with depth and wait parameters $k,u$, i.e., it satisfies the following properties (informally): with overwhelming probability, (i) if $\tx{}{}\in\C{-k}{M}$ for some honest $M$, it will be reported by the other honest miners at the same position in their ledger views (persistence), and (ii) if $\tx{}{}$ is given to all honest miners for $u$ consecutive communication rounds, each honest $M$ will report $\tx{}{}\in\C{-k}{M}$ (liveness).
We use this interface to reason about confirmation and reorgs. We denote for any Ark party (user or $\operator$) $Q$ by $\C{}{Q}$ the local view of $Q$, where $Q$ shares the view of some miner $M$.


\smallskip\noindent\textbf{Communication Model.}
\label{subsec:comm-model}
Parties communicate over authenticated channels. The network is synchronous, i.e.,  all messages are delivered within a known time bound $\Delta$ (which is also the network model of Bitcoin itself). We assume $\operator$ is online, while users may be intermittently online. Users performing fast finality payments need to remain online until their funds are spent.

\smallskip\noindent\textbf{Cryptographic and Script Assumptions.}
\label{subsec:crypto}
We assume standard cryptographic primitives, i.e., collision-resistant hash functions and EUF-CMA secure signature schemes. Taproot Schnorr signatures are used for key/script paths and for efficient $n$-of-$n$ multisignature via MuSig2 \cite{nick2020musig2} (aggregate key denoted $\bigoplus_i pk_i$, checked by $\checkSig{\bigoplus_i pk_i}$). Finally, we remark that as script-level covenants are unavailable, spend restrictions can be emulated with $n$-of-$n$ co-signing. Native support for stronger covenants in Bitcoin would reduce interactivity, when compared to emulation with multisignatures, but are not required for correctness.

\smallskip\noindent\textbf{Threat Model.}
\label{subsec:threat}
A PPT adversary may corrupt any subset of parties (including $\operator$), possibly adaptively and with full collusion. The adversary controls message scheduling (subject to the synchronous bound) and the mempool, but cannot break cryptography or the ledger’s security properties (but with negligible probability). We consider  the adversary to be Byzantine, i.e., it may deviate from the protocol arbitrarily. We further provide an extension enabling fast finality, where we assume rationality of the operator, i.e., it can misbehave only if this leads to a financial gain, and existential honesty for (i.e., at least one being honest among) the signing committee. 

\subsection{Protocol Overview}
\label{subsec:protocol-overview}
We use three recurring objects:
(i) \emph{VTXOs}: offchain outputs with at least one collaborative spend path (requiring $\operator$'s signature) and at least one unilateral spend path (user only, after a relative timelock);
(ii) \emph{batches}: onchain outputs that compactly commit to a set of VTXOs by only being spendable before a certain time through presigned, offchain transactions creating these VTXOs; and
(iii) \emph{connectors}: onchain outputs that can only be spent to create multiple \emph{anchors}, outputs that atomically bind offchain state updates (changes in the VTXO set) to a future onchain batch (cf.\ §\ref{subsec:settle-vtxo}).
Figure~\ref{fig:flow} illustrates the end-to-end flow; we summarise the components below and cross-reference the detailed sections.

\begin{itemize}[leftmargin=*]
  \item \textbf{Onchain commitments
  ,} §\ref{subsec:tx-batching},§\ref{subsec:commitment-tx}. 
  The operator $\operator$ broadcasts \emph{commitment transactions}, typically funded by its own liquidity. These commitments contain \emph{batches} that compactly commit to the newly created VTXOs. 
  \item \textbf{Offchain transaction execution
  ,} §\ref{subsec:ark-tx}. 
  Users collaborate with $\operator$ to build offchain (virtual) transactions, spending existing VTXOs and creating new VTXOs.
  \item \textbf{Batch swapping
  ,} §\ref{subsec:settle-vtxo}. 
  Users exchange (offchain) old VTXOs for new ones in the next batch by \emph{forfeiting} the old VTXOs (allowing $\operator$ to claim them if they go onchain). \emph{Connectors} atomically bind the swap to the next commitment.
  \item \textbf{Boarding
  ,} §\ref{subsec:commitment-tx}. A user posts a boarding transaction whose output is spendable either jointly with $\operator$ or unilaterally after a timelock. Once confirmed, $\operator$ spends it in the next commitment, giving VTXOs to the user for offchain use. 
  \item \textbf{Exiting
  ,} §\ref{subsec:commitment-tx}. To exit, the user forfeits their VTXOs and receives onchain outputs in the next commitment. If $\operator$ is unresponsive or adversarial, the user exits a VTXO unilaterally via the script path of the corresponding batch.
\end{itemize}

\begin{figure*}[ht]
  \centering
  \includegraphics[width=0.8\linewidth]{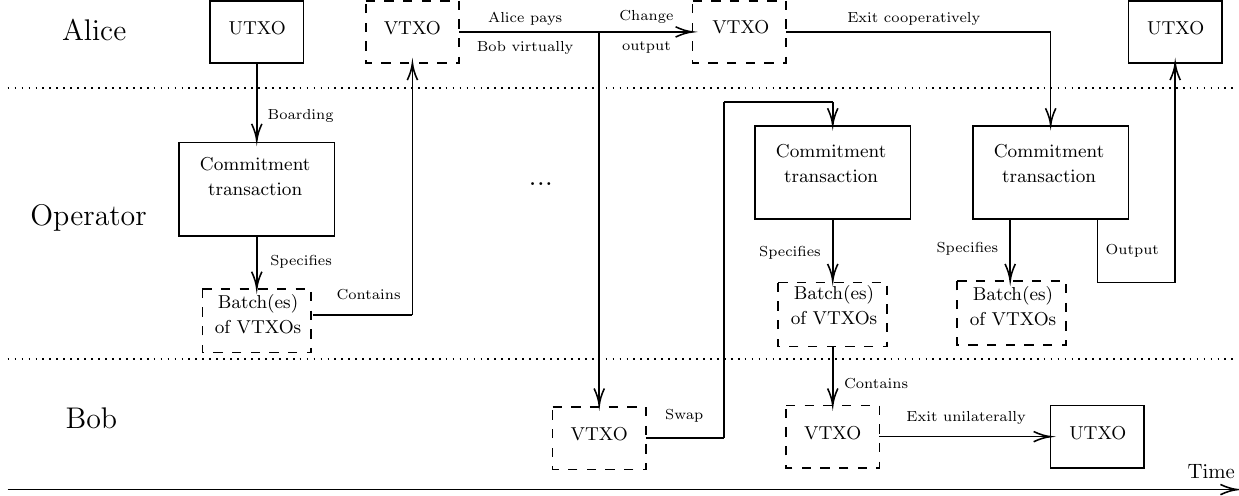}
  \caption{An example Ark protocol flow (solid boxes appear onchain, dashed boxes optimistically never appear onchain). Alice joins the Ark by submitting her onchain UTXO through a boarding procedure (§\ref{subsec:commitment-tx}). She receives a VTXO contained in a batch (§\ref{subsec:tx-batching}), which is an output of a commitment transaction (§\ref{subsec:commitment-tx}) posted onchain by $\operator$. Alice can use this VTXO to transact offchain, e.g., by paying Bob. This happens via an ordinary (but offchain) Bitcoin transaction (§\ref{subsec:ark-tx}), creating a payment VTXO for Bob and a change VTXO for Alice. Alice can exit the Ark cooperatively with $\operator$ (§\ref{subsec:commitment-tx}), turning her change VTXO into an onchain output in a new commitment transaction. Bob secures his new funds via a batch swap (§\ref{subsec:settle-vtxo}) with $\operator$, exchanging the received VTXO for a new VTXO, which is again part of a batch output of an onchain commitment transaction. This mechanism ensures that if $\operator$ would no longer respond, Bob can claim his funds onchain unilaterally (§\ref{subsec:commitment-tx}).}
  \label{fig:flow}
\end{figure*}

\subsection{Security and Scalability Properties}
\label{subsec:properties}
We introduce the security and performance guarantees below, which we informally prove in §\ref{sec:security}. Formal statements and proofs are deferred to Appendix~\ref{app:proofs}. The Ark protocol maintains a set of unspent VTXOs, which we can see as offchain state. Intuitively, VTXO Security guarantees that any Ark state can be mapped to a corresponding Bitcoin state (i.e., a set of UTXOs). Ark Atomicity guarantees the transitions between Ark states either take place in full or not at all. 
\begin{itemize}[leftmargin=*]
    \item \textbf{VTXO Security.} An unexpired\footnote{Expired batches of VTXO can be claimed at once by $\operator$.} 
    VTXO committed in an onchain batch can only be spent with a valid witness (safety). Any VTXO held by an honest user can be unilaterally redeemed onchain before batch expiry (liveness). 
    \item \textbf{Ark Atomicity.} For any Ark action (boarding, transaction, batch swap, exit), either all inputs are consumed and all outputs are created, or none  is. 
\end{itemize}
From these two properties, we can derive \emph{balance security}, for both the users and the operator. Under the additional assumptions of a rational operator and an existentially honest signer committee, users can moreover opt in to securely transact with fast finality, ensuring balance security even for VTXOs not committed in a batch.

Finally, with respect to scalability, Ark’s onchain footprint is: (i) \emph{execution constant}: the onchain commitment size is constant in the number of  VTXOs; 
(ii) \emph{optimistic exit constant}: with a responsive operator, a user exits with $\mathcal{O}(1)$ onchain transactions; and (iii) \emph{pessimistic exit logarithmic}: in the worst case, a user exits a VTXO with $\mathcal{O}(\log n)$ transactions, where $n$ is the batch’s VTXO population. With fast finality, all previous offchain transactions leading to the to-be-exited VTXO need to be posted as well.

\section{The Ark Protocol}
\label{sec:construction}

At the core of Ark is an operator $\operator$ with key pair $(\sk_\operator,\pk_\operator)$ who facilitates transactions of \emph{virtual UTXOs} called VTXOs. VTXOs are offchain outputs that can be spent without leaving an onchain footprint (§\ref{subsec:ark-tx}). Each VTXO is locked by a Taproot script with two paths: a \emph{collaborative} spend with $\operator$, and a \emph{unilateral} spend that requires no interaction with $\operator$.

\begin{definition}[Virtual UTXO / VTXO]
\label{def:vtxo}
A \emph{virtual UTXO} or \emph{VTXO}, with operator $\operator$, is a 
transaction output of the form
$\emph{\texttt{vtxo}}:=\qty(\emph{\texttt{value}},\emph{\texttt{vtxoLockScript}})$ 
where \emph{\vtxoLockScript} is a Taproot locking script with:
(i) an unspendable key path; 
(ii) at least one \emph{collaborative} script path, requiring a signature of $\operator$; and 
(iii) at least one \emph{unilateral} script path, not requiring $\operator$'s signature and delayed by a relative timelock $\tv$ (determined by $\operator$), ensuring that a VTXO appearing onchain can be spent collaboratively at least $\tv$ blocks before it becomes unilaterally spendable.
A VTXO is called \emph{collaboratively} or \emph{unilaterally} spent if the corresponding path is used.
\end{definition}

The simplest VTXO is a \emph{single-signature} output: an amount an Ark user Alice, with key pair $\qty(\sk_A,\pk_A)$, can spend via a signature under $\pk_A$. Its locking script $\texttt{Taproot}(\False;\checkSig{\pk_\operator\oplus \pk_A},$ $\checkSig{\pk_A}\wedge$ $\relTimelock{\tv})$ requires a collaborative witness $w_{collab}=\sigma_{\operator\oplus A}$ or a unilateral witness $w_{unilat}=\sigma_A$. For simplicity, we assume in this section that VTXOs have a clear ``owner'' (e.g., Alice or Bob), who can spend them both collaboratively and unilaterally. However, VTXOs may include multiple complex script paths, generalising ownership to anyone providing a valid witness. This is assumed in the rest of the paper.

Throughout the rest of this paper, we adopt two conventions. First, when $\vtxo{}=(\val,\vtxoLockScript)$ appears as an output in a (virtual) transaction, it represents a value $\val$ and a Taproot output script only\footnote{To be precise, the script reads $\texttt{OP\_1 OP\_PUSHBYTES\_32 } \pk_T$.} showing the 32-byte tweaked public key committing to the key and script paths in $\vtxoLockScript$. A blockchain observer only sees $\pk_T$ and cannot infer the spending paths. In other contexts (e.g., an Ark user sharing $\vtxo{}$ with $\operator$), both $\val$ and all spending paths are shared, allowing derivation of $\pk_T$. Second, when requests create new VTXOs, these are initially just value–script pairs, not yet tied to a specific transaction. Context clarifies whether we refer to this pair or to a fully defined VTXO within a batch.

We now describe Ark’s core functionality: transaction batching, which consolidates multiple VTXOs into a single onchain output called a \emph{batch}. Users can execute transactions offchain by spending and creating VTXOs, and batch swapping to atomically exchange old VTXOs for fresh ones in a new batch. This process motivates the design of commitment transactions, the only onchain footprint. We then extend this structure to support Ark user entry and exit.

\subsection{Transaction Batching}
\label{subsec:tx-batching}
As the name suggests, VTXOs are meant to remain virtual. While they can be unilaterally converted into UTXOs via an exit path, this is ideally avoided: VTXO holders are instead expected to collaborate with $\operator$ to claim funds onchain (see §\ref{subsec:commitment-tx}). This allows us to design a structure that minimises onchain footprint while still permitting unilateral exits if needed. This structure, called a \emph{batch}, is a transaction output predominantly funded by $\operator$ that can be spent either via a \emph{sweep path} (returning funds to $\operator$) or an \emph{unroll path}, which splits the onchain UTXO into offchain VTXOs via a \emph{virtual transaction tree (VTXT)}. This tree is made up of virtual transactions: presigned Bitcoin transactions that optimistically never go onchain, specifying how the batch should be distributed amongst VTXOs.


\begin{definition}[Virtual transaction tree]
\label{def:vtxt}
    A \emph{virtual transaction tree (VTXT)} is a directed rooted tree, given by the ordered pair $G=(V,A)$, where $V$ is a set of virtual transactions (the \emph{nodes}) with exactly one input, and $A$ is a set of ordered pairs of virtual transactions (the \emph{edges}) such that, for every $u_1,u_2\in V$, we have $(u_1,u_2)\in A$ if and only if $u_2.\inputs\supseteq u_1.\outputs$. There is also exactly one virtual transaction $r\in V$, called the \emph{root}, such that there are no edges $a\in A$ of the form $(u,r)$, where $u\in V$. For all other $u_2\in V$, there is exactly one $u_1\in V$ such that $(u_1,u_2)\in A$. Any $s\in V$ for which there are no edges of the form $(s,u)$ in $A$ (for $u\in V$) is called a \emph{leaf}. 
    For any $u\in V$, we define $\patht(u)$ as the sequence $(u_1,\ldots,u_\ell)$, where $u_1$ is the root, $u_\ell=u$, and $(u_i,u_{i+1})\in A$ for each $i\in\qty{1,\ldots,\ell-1}$. For any $\out{}{}\in u.\outputs$, we interpret $\patht(\out{}{})$ as $\patht(u)$.
\end{definition}
\begin{definition}[Batch]
\label{def:batch}
    A \emph{batch} is a transaction output locked by a Taproot script with an unspendable key path and exactly two script paths: (i) a \emph{sweep path} that allows $\operator$ to claim the entire output after a certain block height, called the \emph{batch expiry}, and (ii) an \emph{unroll path} that specifies spending according to a VTXT with root spending the full batch, each leaf a VTXO as its only output, and the remaining nodes virtual transactions that have batches as their only outputs.    
\end{definition}
\begin{remark}
\label{rem:signing}
The VTXT structure needs to be enforced by a covenant, ensuring the batch can be spent \emph{only} according to the VTXT before expiry. To preserve Bitcoin compatibility, this covenant can be emulated with an $n$-of-$n$ multisignature (e.g., Musig2), where $\operator$ coordinates signing sessions with all VTXO holders. Either all holders sign every virtual transaction, or, as in Figure~\ref{fig:tx}, each holder signs only along the path to their VTXO. The latter reduces the number of interactions without compromising safety, since every honest holder must approve the transactions needed to redeem their VTXO, ensuring their batch funds cannot be spent otherwise. 
\end{remark}
The purpose of batch expiry is explained in §\ref{subsec:commitment-tx}. For now, note that any VTXO holder who knows the virtual transaction path in the VTXT to their VTXO can broadcast it onchain and exit unilaterally. Figure~\ref{fig:tx} illustrates this: to redeem $v_1$, $P_1$ must publish the transactions with outputs $(v_1+v_2, v_3+v_4)$, $(v_1, v_2)$, and $v_1$. Only $P_1$’s VTXO appears onchain, but $P_1$ must pay fees, which may make small VTXOs uneconomical to exit.

\subsection{Ark transactions}
\label{subsec:ark-tx}
As discussed earlier, VTXOs are intended to remain offchain, enabling the off\-chain execution of transactions via \emph{Ark transactions}. These spend one or more VTXOs from a batch through the collaborative path, creating new VTXOs.
To illustrate the Ark transaction mechanism, consider an Ark user Alice who wants to send amount $p$ to Bob. Alice holds a VTXO $\vtxo{A}$ with value $a > p$ in a batch. She constructs an Ark transaction $\tx{A}{ark}$, a virtual transaction with inputs $[\vtxo{in,1},\ldots,\vtxo{in,n}]$, collaborative witnesses $[w_{collab}^1,\ldots,w_{collab}^n]$, and outputs $[\vtxo{out,1},$ $\ldots,\vtxo{out,m}]$. In this example, $n=1$, $\vtxo{in,1}=\vtxo{A}$, $w_{collab}^1 = \sigma_{\operator\oplus A}$, and the outputs may include: a VTXO $\vtxo{B}$ with amount $p$ spendable by Bob, a change VTXO for Alice, and possibly a fee VTXO to $\operator$.

Alice signs $\tx{A}{ark}$ and asks $\operator$ to co-sign. Once signed, she sends the transaction to Bob along with the VTXT path to $\vtxo{A}$, enabling Bob to exit $\vtxo{B}$ unilaterally. In a way, Bob now holds his own VTXO, without needing prior funds in Ark or on Bitcoin. 

However, there is a critical issue: Bob must trust Alice not to double-spend $\vtxo{A}$. She could unilaterally exit with it (forcing Bob to constantly monitor the chain and post $\tx{A}{ark}$ if Alice ever exits with $\vtxo{A}$) or collude with $\operator$ to create conflicting Ark transactions spending the same VTXO. Since Ark transactions are virtual, Alice and a (malicious) $\operator$ could theoretically spend the same VTXO arbitrarily many times. If Alice broadcasts $\vtxo{A}$ onchain, only one such transaction could be confirmed. Next, we explain how Bob can secure his funds by batch swapping $\vtxo{B}$ for a new VTXO. 

\subsection{Batch Swaps}
\label{subsec:settle-vtxo}
To finalise Alice’s payment, Bob \emph{atomically} swaps $\vtxo{B}$ for a new VTXO, a process we call \emph{batch swapping}. Here, $\operator$ takes $\vtxo{B}$ and issues Bob a new VTXO $\vtxo{B'}$ in the next batch, which serves as onchain confirmation that Bob owns value $b$ (which may be smaller than $p$ to account for a fee $\operator$ may charge). This guarantees that either both transfers succeed or both fail, removing the need for Bob to monitor the chain or trust $\operator$ not to collude with Alice. For clarity we describe one-to-one VTXO swaps, though the protocol supports swapping multiple VTXOs into arbitrary new allocations.

A batch swap starts with Bob requesting $\operator$ to swap $\vtxo{B}$. $\operator$ builds a transaction $\tx{\operator}{}$ spending its own funds and producing outputs $\tx{\operator}{}.\outputs = [(b,\checkMultiSig{2,2;pk_\operator,pk_B}),(\varepsilon,\checkSig{pk_\operator})]$. This is a basic commitment transaction (formally defined in §\ref{subsec:commitment-tx}). The first output is a batch containing one VTXO; the second is an \emph{anchor output} with dust value $\varepsilon$, essential for atomicity.


$\operator$ also creates a virtual transaction $\tx{}{virtual}$ spending the first output of $\tx{\operator}{}$ and producing $\vtxo{B'}$. $\operator$ signs $\tx{}{virtual}$ and shares it with Bob along with the anchor output. Bob then constructs and signs a \emph{forfeit transaction} $\tx{}{forfeit}$ spending both $\vtxo{B}$ and the anchor, sending all funds to $\operator$. The signature uses the \texttt{SIGHASH\_ALL} flag\footnote{Bitcoin’s \texttt{SIGHASH} flags specify which parts of a spending transaction are covered by a signature. In particular, \texttt{SIGHASH\_ALL} fixes all inputs and outputs.} and is only valid for this specific transaction spending the anchor. Thus, $\tx{}{forfeit}$ is valid only if both $\tx{\operator}{}$ and $\vtxo{B}$ are onchain. $\operator$ can now safely broadcast $\tx{\operator}{}$, as the anchor ensures atomicity. If Bob turns $\vtxo{B}$ into a UTXO, $\operator$ can claim $b$ via $\tx{}{forfeit}$. On the other hand, $\operator$ cannot claim $\vtxo{B}$ using $\tx{}{forfeit}$ as long as $\tx{\operator}{}$ is not onchain.

Bob has successfully swapped $\vtxo{B}$ for $\vtxo{B'}$, which he can now unilaterally exit by posting $\tx{}{virtual}$, regardless of whether Alice tries to unilaterally exit with or double-spend $\vtxo{A}$. He no longer needs to monitor the blockchain for Alice turning $\vtxo{A}$ into a UTXO. If she does, it is now in $\operator$'s interest to broadcast $\tx{A}{ark}$ and $\tx{}{forfeit}$ to claim what are now its funds.


This construction understandably seems cumbersome. However, its power becomes clear when realising we can process more than one VTXO in this way. Indeed, the first output of $\tx{\operator}{}$ is just a batch containing one VTXO (where the VTXT consists only of $\tx{}{virtual}$). We can increase the size of this batch, including other VTXOs that may also have been created through batch swaps. 

Finally, note that the current construction also locks $\operator$’s funds, since Bob’s signature is required.  We show next how these funds are later released to $\operator$.

\subsection{Commitment transactions}
\label{subsec:commitment-tx}
The previously defined $\tx{\operator}{}$ offers a natural way to contain batches as its outputs. Moreover, the second output of $\tx{\operator}{}$ enables to swap VTXOs atomically. This second output is a so-called \emph{connector}.
\begin{definition}[Connector]
\label{def:connector}
    A \emph{connector} is a transaction output which is locked by a Taproot script $\emph{\connectorScript}$ with an unspendable key path and a script path that specifies spending according to a VTXT where the root spends the full connector, and where each leaf of the VTXT has an anchor output as its only output. The remaining nodes of the VTXT are virtual transactions with a connector as their only output. Each output can be spent by a signature from $\operator$.
\end{definition}
A connector contains (multiple) anchor outputs serving as inputs to forfeit transactions, that can only be included onchain if the commitment containing that connector is included onchain. Unlike batches, the virtual transactions of a connector are signed solely by $\operator$. We now simply define a commitment transaction as follows.
\begin{definition}[Commitment transaction]
\label{def:commit-tx}
    A \emph{commitment transaction} (or commitment) is a transaction with at least one batch or one connector output.
\end{definition}
A commitment may output multiple batches and connectors. Additionally, it may contain inputs and outputs related to users joining (\emph{boarding}) and leaving (\emph{exiting}) the Ark. 


\smallskip\noindent\textit{Boarding.} Consider a user Alice with some onchain funds locked in a UTXO $\out{A}{}$. Alice can join, or \emph{board} the Ark via a two-step process. First, Alice constructs, signs and broadcasts the \emph{boarding transaction} $\tx{}{board}$ with input $\out{A}{}$ and output $\out{A}{*}$, where the locking script of the output $\out{A}{*}$ can be written as $\texttt{Taproot}(\texttt{False};\checkSig{pk_\operator\oplus pk_A},$ $\checkSig{pk_A}\wedge$ $\relTimelock{t_b})$.
This is a Taproot script with an unspendable key path and two script paths. The first script path is a $\emph{cooperative}$ path, in which Alice and $\operator$ spend the funds together, and the second is an \emph{exit} path, which lets Alice spend her funds after a timeout $t_b$. Alice can now send a boarding request to $\operator$, who will first verify that $\out{A}{*}$ cannot be spent without $\operator$ before the timeout (preventing a double-spend of the next commitment $\tx{}{commit}$), and then create one or more VTXOs for Alice in a batch of $\tx{}{commit}$. In return, $\out{A}{*}$ is added as an input to $\tx{}{commit}$ via its cooperative path. In case Alice would not want to board the Ark after all, she can either spend $\out{A}{*}$ via the exit path after the timeout, or cooperate with $\operator$ to spend $\out{A}{*}$ for a new UTXO that can be spent by Alice only.

\smallskip\noindent\textit{Exiting.} Alice holds one or more VTXOs in confirmed batches and can exit either \emph{unilaterally} or \emph{collaboratively}.  
In the unilateral case, she simply broadcasts $\patht(\vtxo{})$ for each owned $\vtxo{}$, without interacting with $\operator$.  
In the collaborative case, Alice batch swaps her VTXOs, but instead of receiving new ones, $\operator$ includes a UTXO for her in the next commitment. As in §\ref{subsec:settle-vtxo}, $\operator$ also obtains a signed forfeit transaction to prevent Alice from reusing old VTXOs.  

\begin{figure*}
    \centering
    \includegraphics[width=\linewidth]{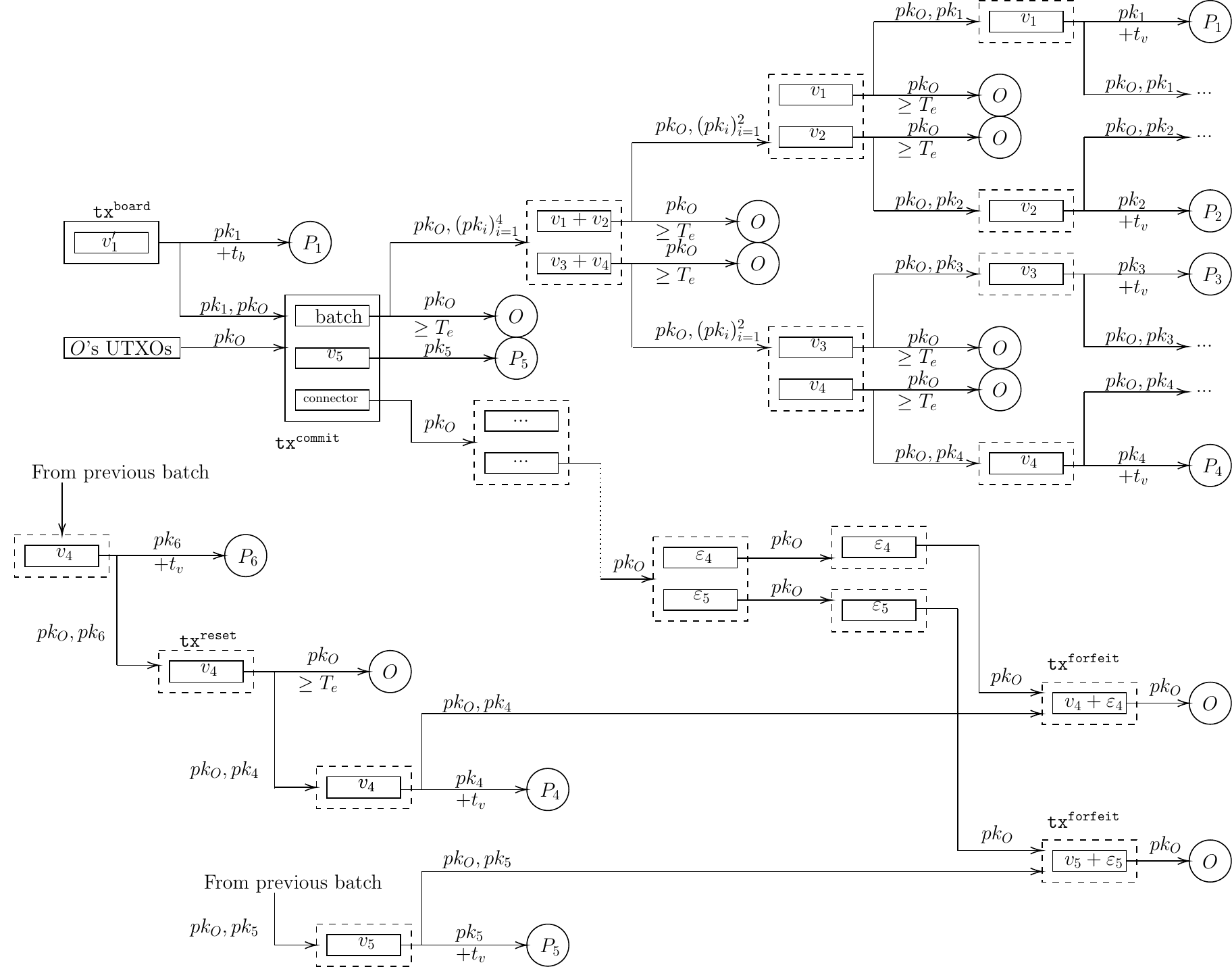}
    \caption{Transaction dependencies within the Ark protocol. UTXOs are represented by small rectangles enclosing a label, representing either the value or their function. (Multiple) UTXOs can be enclosed by a larger rectangle, which is a transaction. Multiple arrows may leave from a UTXO. Each arrow is labelled by a spending condition and may point to a transaction that spends the respective UTXO. Dashed transactions are virtual transactions and optimistically never appear onchain. We set $T_e=h+2k+t_e$, where $h$ is the block height $\tx{}{commit}$ got submitted by $\operator$ and $t_e$ the time in which a batch should expire. Moreover, recall that $\tv$ is the minimum delay for a unilateral VTXO exit compared to any collaborative exit, and $t_b$ the boarding transaction timeout period. 
    }
    \label{fig:tx}
\end{figure*}

Constructing a fully signed commitment requires user coordination (see Figure~\ref{op:op} and Appendix~\ref{app:protocol}). This process assumes responsive participants\footnote{If a user is offline, the transaction aborts but can be retried without that request.}. Figure~\ref{fig:tx} illustrates: $P_1$ boards via a boarding transaction; $\operator$ combines $P_1$’s and its own funds into a commitment; $P_4$ swaps a VTXO; and $P_5$ exits. The batch then includes a VTXO for $P_1$, a new VTXO for $P_4$, a UTXO for $P_5$, and a connector allowing $\operator$ to reclaim $P_4$’s and $P_5$’s old VTXOs if posted onchain.


\subsection{Ark transactions revisited}
\label{subsec:ark-revisited}
The earlier testnet version of the Ark protocol is fully described by §\ref{subsec:tx-batching}-§\ref{subsec:commitment-tx}. However, we found out two critical vulnerabilities in this protocol, which we call the \emph{hostage attack} and \emph{spam attack}.

\smallskip\noindent\textbf{Hostage attack.} Recall from §\ref{subsec:settle-vtxo} that $\operator$ is responsible for broadcasting Ark and forfeit transactions when already spent or batch swapped VTXOs appear onchain. A problem arises when an Ark transaction includes inputs from multiple batches: at some point, some batches will have expired whereas others have not. $\operator$ wants to sweep expired batches, but doing so exposes it to users who try to unilaterally claim unexpired VTXOs. Normally, $\operator$ would broadcast the Ark transactions up to the batch swap in response, then the forfeit transaction to claim batch swapped output VTXOs. However, once a batch containing an input VTXO has been swept, the Ark transaction becomes invalid and cannot be included onchain. $\operator$ has thus no way to claim this spent VTXO, resulting in a loss of funds. A malicious user could hold $\operator$'s funds hostage for prolonged periods of time by having Ark transactions with inputs from many different batches. For example, if a user has an Ark transaction spending from batches expiring at $T_1$ and $T_2>T_1$, and an Ark transaction spending from batches expiring at $T_2$ and $T_3>T_2$, then $\operator$ cannot sweep any batch until $T_3$. This attack can be prolonged as long as the operator processes Ark transactions spending from multiple batches. If the operator would not do so this would severely limit the Ark's functionality.

\smallskip\noindent\textbf{Spam attack.} Recall from §\ref{subsec:settle-vtxo} that a forfeit transaction is only signed for a VTXO that is batch swapped. Hence, a malicious user Mallory holding a VTXO $\vtxo{M}$ could, in theory, make a chain of Ark transactions to herself, each time spending the new VTXO in the next transaction, and batch swap or cooperatively exit the final VTXO $\vtxo{M'}$. When the new commitment is confirmed, she could then exit $\vtxo{M}$ unilaterally. The operator now has to post all Ark transactions onchain, incurring onchain fees to finally broadcast the forfeit transaction that spends $\vtxo{M'}$. As the potential Ark transaction fees will be significantly lower than the corresponding onchain fees, this spam attack could make it unprofitable for $\operator$ to broadcast the transaction chain up to the forfeit transaction, allowing Mallory to illegitimately claim both $\vtxo{M}$ and $\vtxo{M'}$, essentially stealing $\vtxo{M'}.\val$ from the operator. 

\smallskip\noindent\textbf{Reset transactions.} To defend against these attacks, we introduce an additional virtual transaction, the \emph{reset transaction}, which allows $\operator$ to claim spent VTXOs even if the corresponding Ark transaction became invalid after an earlier sweep. This transaction shifts the responsibility of posting the subsequent Ark transaction to users.

Consider the setting of §\ref{subsec:ark-tx}, where Alice sends an amount $p$ to Bob. In addition to constructing an Ark transaction as before, she also creates a reset transaction spending $\vtxo{A}$ collaboratively, producing one output with value $p$ and locking script $\texttt{Taproot}(\False;$ $\checkSig{pk_\operator\oplus pk_A},$ $\checkSig{pk_\operator}\wedge\absTimelock{T_e}))$. Her Ark transaction then spends this output, giving $p$ to Bob. The reset output can thus be spent either jointly by Alice and $\operator$ or swept by $\operator$ once the batch expires. Alice first signs the Ark transaction and sends both the Ark and reset transactions to $\operator$, who verifies the scripts, signs both, and returns them. Alice then signs the reset transaction and passes it to $\operator$\footnote{More precisely, both cooperatively produce the aggregated signatures for the Ark and reset transactions, in that order.}. The overall back-and-forth is shown in Figure~\ref{op:ark}. The resulting transaction flow is also illustrated in Figure~\ref{fig:tx}, where $P_3$ sends $P_4$ the amount $v_4$, which $P_4$ then batch swaps.

If Alice wants to perform an Ark transaction with multiple inputs, she first constructs a reset transaction for each input VTXO, and builds an Ark transaction spending the outputs of all these reset transactions. This protects $\operator$ for the hostage attack: if VTXOs from different batches are spent by an Ark transaction, expired batches can safely be swept. Indeed, if an input VTXO appears onchain, $\operator$ can simply broadcast the corresponding reset transaction, claiming the funds at that batch's expiry. The Ark transaction, which requires all input VTXOs and their corresponding reset transactions to be onchain, became invalid once the expired batches were swept. A receiver wishing to exit unilaterally with the outputs of the Ark transaction should thus do so before any of the batches holding an input to the Ark transaction expires.

In case of a spam attack, where Mallory exited unilaterally with $\vtxo{M}$, the operator will now simply broadcast the appropriate reset transaction, forcing Mallory to post the subsequent Ark transaction. The operator will always force Mallory to post the next Ark transaction, until eventually the operator can post the forfeit transaction. Mallory having to post all the Ark transactions to try to claim $\vtxo{M}$ should deter Mallory from trying to claim both $\vtxo{M}$ and $\vtxo{M'}$. Note that this deters any rational attacker from mounting this attack. However, a Byzantine user could mount this attack regardless. To protect against such attackers, $\operator$ could additionally limit the length of a chain of Ark transactions, forcing a batch swap at some point.

\begin{figure}[h]
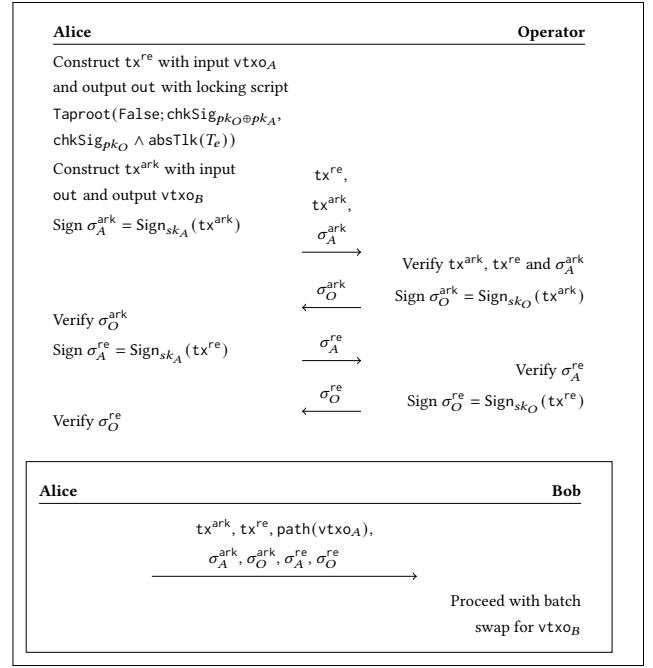

    \centering
    \resizebox{\linewidth}{!}{
        \begin{bbrenv}{A}
            \begin{bbrbox}
                \pseudocodeblock[colsep=0em]{
                    \textbf{Alice} \< \> \textbf{Operator}\\[][\hline]
                    \< \> \\[-1em]
                    \text{Construct $\tx{}{re}$ with input $\vtxo{A}$} \< \> \\
                    \text{and output $\out{}{}$ with locking script} \< \> \\ 
                    \text{$\texttt{Taproot}(\False;\checkSig{pk_\operator\oplus pk_A},$} \< \> \\ 
                    \text{$\checkSig{pk_\operator}\wedge\absTimelock{T_e})$} \< \> \\
                    \text{Construct $\tx{}{ark}$ with input} \< \> \\
                    \text{$\out{}{}$ and output $\vtxo{B}$} \< \> \\
                    \text{Sign $\sigma_A^{\texttt{ark}}=\Sign_{sk_A}(\tx{}{ark})$} \< \> \\[-5em]
                    \< \sendmessageright*[1cm]{\text{$\tx{}{re}$,\ } \\ \text{$\tx{}{ark}$,} \\ \text{$\sigma_A^{\texttt{ark}}$\ }} \> \\[-1em]
                    \< \> \text{Verify $\tx{}{ark}$, $\tx{}{re}$ and $\sigma_A^{\texttt{ark}}$} \\[-0.6em]
                    \< \sendmessageleft*[1cm]{\text{$\sigma_\operator^{\texttt{ark}}$}} \> \text{Sign $\sigma_\operator^{\texttt{ark}}=\Sign_{sk_\operator}(\tx{}{ark})$} \\[-1em]
                    \text{Verify $\sigma_\operator^{\texttt{ark}}$} \< \> \\[-0.6em]
                    \text{Sign $\sigma_A^{\texttt{re}}=\Sign_{sk_A}(\tx{}{re})$} \< \sendmessageright*[1cm]{\text{$\sigma_A^{\texttt{re}}$}} \> \\[-1em]
                    \< \> \text{Verify $\sigma_A^{\texttt{re}}$} \\[-0.6em]
                    \< \sendmessageleft*[1cm]{\text{$\sigma_\operator^{\texttt{re}}$}} \> \text{Sign $\sigma_\operator^{\texttt{re}}=\Sign_{sk_\operator}(\tx{}{re})$} \\[-1em]
                    \text{Verify $\sigma_\operator^{\texttt{re}}$} \< \> \\
                }
                \begin{bbrenv}{D}
                    \begin{bbrbox}
                        \pseudocodeblock{
                            \textbf{Alice} \< \> \textbf{Bob}\\[][\hline]
                            \< \> \\[-1em]
                            \quad\quad\quad\quad\quad\quad \< \sendmessage*{->}{top = {\text{$\tx{}{ark}$, $\tx{}{re}$, $\patht(\vtxo{A})$,} \\ \text{$\sigma_A^{\texttt{ark}}$, $\sigma_\operator^{\texttt{ark}}$, $\sigma_A^{\texttt{re}}$, $\sigma_\operator^{\texttt{re}}$\quad\quad}}, length=4.5cm} \> \\
                            \< \> \text{Proceed with batch} \\
                            \< \> \text{swap for $\vtxo{B}$}
                        }
                    \end{bbrbox}
                \end{bbrenv}
            \end{bbrbox}
        \end{bbrenv}
    }
    \caption{Alice sends funds to Bob via an Ark transaction. Upon receiving the signed Ark transaction $\tx{}{ark}$ and reset transaction $\tx{}{re}$ (described in §\ref{subsec:ark-tx} and §\ref{subsec:ark-revisited} respectively), Bob will request a batch swap as introduced in §\ref{subsec:settle-vtxo} and only consider the transaction finalised once the corresponding commitment transaction is confirmed onchain.}
    \label{op:ark}
\end{figure}

\subsection{Offchain handover}
\label{subsec:handover}
As described in §\ref{subsec:commitment-tx}, unilateral exits involve broadcasting the whole virtual transaction path leading up to the to-be-exited VTXO. Onchain activity may thus spike if many Ark users unilaterally exit their VTXOs, for example when an operator stops responding and users must claim funds to prevent loss. Even if the operator remains active, facilitating collaborative exits can still create numerous onchain UTXOs when many users are involved. If some users wish to stay in an Ark, such exits require a subsequent onchain boarding transaction into a new Ark. To avoid this footprint, we introduce an offchain handover procedure, allowing a VTXO in an Ark with operator $\operator_1$ to move into an Ark with operator $\operator_2$. If $\operator_1$ plans to shut down but stays online to process handovers, its users can migrate to $\operator_2$’s Ark entirely offchain, with the only onchain footprint being $\operator_1$ reimbursing $\operator_2$ for the transferred VTXOs.

Consider a user Alice holding a VTXO $\vtxo{A}$ of value $v$ in a confirmed batch funded by $\operator_1$. Suppose $\operator_1$ plans to shut down its Ark but remains cooperative to transfer Alice's VTXO to an Ark run by $\operator_2$. This transfer is effectively a batch swap with connector spendable by $\operator_1$, and the new VTXO funded by $\operator_2$.

Alice requests a batch swap to $\operator_2$, as in §\ref{subsec:settle-vtxo}. $\operator_2$ includes a new VTXO for Alice and an anchor output in its commitment transaction, but this anchor is now spendable by $\operator_1$. Instead of Alice and $\operator_2$ signing a forfeit transaction spending the anchor and giving $\vtxo{A}$ to $\operator_2$, Alice and $\operator_1$ sign one giving it to $\operator_1$. This two-operator batch swap ensures $\operator_1$ can claim $\vtxo{A}$ if Alice exits unilaterally. Since $\operator_2$ funds Alice’s new VTXO without being able to sweep $\vtxo{A}$, $\operator_1$ must reimburse $\operator_2$ by amount $v$. This is conditional on $\operator_2$’s commitment being confirmed onchain. To enforce this atomicity, $\operator_2$’s commitment includes an anchor output $\out{1\to 2}{}$ used as input to $\tx{1\to 2}{}$, where $\operator_1$ pays $\operator_2$ the amount $v$. Before signing this modified commitment, $\operator_2$ requests $\tx{1\to 2}{}$, signed by $\operator_1$. If valid, $\operator_2$ can broadcast both transactions and safely claim $v$. Conversely, $\operator_1$ only pays if $\operator_2$’s commitment is confirmed onchain. This scheme scales to multiple VTXOs, where $\tx{1\to 2}{}$ now pays the total amount $\operator_2$ needed to fund all the transferred VTXOs. For the users, the online requirement is the same as for a normal batch swap (some tasks are simply in collaboration with $\operator_1$, and others with $\operator_2$). Additionally, we need $\operator_1$ to remain responsive and handle the extra overhead of signing off on $\tx{1\to 2}{}$ and verifying $\operator_2$'s commitment.

\section{Fast Finality}
\label{sec:extensions}
In §\ref{subsec:settle-vtxo} Bob must batch swap the VTXO received from Alice before finalising the payment. Ownership is only guaranteed once the corresponding commitment confirms onchain, as Alice may still collude with $\operator$. This limits how quickly VTXOs can be spent again.

If $\operator$ is honest, it will never cosign a conflicting spend of an already-spent VTXO, and Bob can therefore safely accept and re-spend new VTXOs without waiting for the next onchain commitment.
Building on this observation, we introduce a mechanism that enforces this behaviour economically: a rational operator that double-signs stands to lose more than it can gain.
Concretely, consider a fixed set $\Nff$ of Ark users who opt in to transact with fast finality, and assume the following additions to the model of §\ref{sec:overview}:
\begin{enumerate}[leftmargin=*, label=(A\arabic*)]
    \item \label{asm:1} $\operator$ is \emph{rational}, i.e., maximising profit,
    \item \label{asm:3} users in $\Nff$ form a broadcast network where messages are delivered within the known time bound $\Delta$. 
    \item \label{asm:4} it is known to the users in $\Nff$ that they collectively hold a value $v$ in VTXOs. 
\end{enumerate}

We stress that participation in the fast finality scheme is entirely opt-in. Users who are unwilling to accept assumptions \ref{asm:1}--\ref{asm:4} can continue using the Ark protocol as presented before, where transaction finality reduces to that of the underlying blockchain.

\subsection{Protocol Overview} 
The fast finality protocol requires $\operator$ to lock a collateral of value $c_\operator>v$ in a BitVM instance \cite{linus2025bitvm2,woll2026bitvm3}, set up with all users in $\Nff$. This enables the victim of a double-spend to burn $\operator$'s collateral by presenting the conflicting transactions. We describe this BitVM instance in §\ref{subsec:bitvm}. Once this instance is set up, the fast finality protocol for a fixed $\Nff$ proceeds as follows. Each user in $\Nff$ maintains a local record of all observed Ark transactions, which we call an \emph{Ark ledger}. By assumption~\ref{asm:3}, every Ark transaction reaches all users in $\Nff$ within $\Delta$. Upon receiving a new VTXO, a user broadcasts the corresponding Ark transaction to $\Nff$, waits $2\Delta$ to ensure no conflicting transaction appears, and then accepts the payment. If a conflict is detected, either via the broadcast network or by monitoring the chain, the user initiates the BitVM dispute and burns the operator collateral. The full protocol is given in Protocol~\ref{pol:ff}.

\begin{definition}[Ark ledger]
\label{def:ark-ledger}
    An \emph{Ark ledger} $\L$ is a set of Ark transactions that (i) carry valid witnesses for all their VTXO inputs, and (ii) spend only VTXOs that are either created by other transactions in $\L$ or committed to in a confirmed, unswept batch output.
\end{definition}
Note that $\L$ is an unordered set: the broadcast network disseminates transactions but does not impose a total order on them.

\begin{remark}
\label{rem:ff-path}
    Fast finality Ark transactions may have multiple VTXO inputs and outputs. The notion of $\patht(\vtxo{})$ from Definition~\ref{def:vtxt} therefore generalises: unilaterally exiting $\vtxo{}$ may require posting not only the virtual transactions from the batch, but also every fast finality Ark transaction in the chain leading to $\vtxo{}$, potentially spanning multiple batches. Each user must therefore store the full transaction history relevant to its VTXOs.
\end{remark}

\begin{algorithm}
    \captionsetup{name=Protocol}
    \caption{Fast finality protocol for an Ark user $P_1$.}
    \label{pol:ff}
    \begin{algorithmic}[1]
    \State Check that $P_1$ is listed in $\Nff$.
    \While{$P_1$ holds $\vtxo{}$ not committed to in a confirmed batch}
    \State Listen for transactions onchain in conflict with $\patht(\vtxo{})$
    \If{$\tx{}{*}$ appears, conflicting with $\tx{}{}\in\patht(\vtxo{})$}
    \State Initiate a BitVM dispute with $(\operator,(\tx{}{},\tx{}{*}))$
    \EndIf
    \EndWhile
    \State \textbf{To spend $\vtxo{}$ in $\tx{}{ark}$ to $\boldsymbol{P_2}$:}
    \State Proceed as in Figure~\ref{op:ark}, with $P_1$ Alice and $P_2$ Bob, except that $P_2$ does not request a batch swap
    \State \textbf{When receiving a message $p=(\tx{}{ark},\tx{}{re},\patht(\vtxo{}))$:}
    \If{$P_1$ is the payment recipient}
    \For{all $\tx{}{}\in\patht(\vtxo{})$}
    \State Reject the payment if the sender of $\tx{}{}$ is not in $\Nff$.
    \EndFor
    \State Broadcast $p$ to $\Nff$
    \State Wait $2\Delta$
    \EndIf
    \If{$p$ does not conflict with $P_1$'s Ark ledger $\L$ or the chain}
    \State Add $\tx{}{ark}$ and $\tx{}{re}$ to $\L$
    \If{$P_1$ is the payment recipient}
    \State Accept payment
    \EndIf
    \Else
    \State Initiate a BitVM dispute with the conflicting transactions
    \If{$P_1$ is the payment recipient}
    \State Reject payment
    \EndIf
    \EndIf
    \end{algorithmic}
\end{algorithm}

\smallskip\noindent\textbf{Security intuition.} We claim that any user in $\Nff$ who follows Protocol~\ref{pol:ff} cannot be defrauded, as long as $c_\operator > v$. The $2\Delta$ waiting period ensures that no two honest users in $\Nff$ accept conflicting payments: if a conflict exists, at least one recipient will observe it before finalising. A malicious user who attempts to exit onchain with a conflicting transaction will be detected by the honest recipient monitoring the chain, triggering collateral burning. Since the burned amount $c$ exceeds the maximum gain $v$ from any double-sign, a rational operator will refrain from misbehaving. The same logic applies if the operator attempts to collude with a subset of batch cosigners to restructure the batches from which a fast finality VTXO originates. Finally, because collaborative spends take priority over unilateral spends (which are delayed by $t_u$), an honest user monitoring the chain can always preempt a previous VTXO owner attempting to reclaim already-spent funds.

\begin{remark}
    Our presentation assumes a fixed set $\Nff$. In practice, $\operator$ can periodically publish $\Nff$, either onchain or at a publicly accessible location with an onchain commitment, specifying how to reach each participating user. The only power this grants $\operator$ is determining which users can safely accept fast finality payments; each honest user independently verifies its own membership. The design extends naturally to a dynamic $\Nff$. However, this requires setting up a new BitVM instance with the new set of users. This is a one-time setup, after which the users can transact with fast finality arbitrarily many times (up to VTXO expiry). To ensure continued fast finality operations, the new BitVM instance can be set up before the old instance is wound down, having a brief period in time where the operator needs to lock two amounts of collateral simultaneously. In case a dispute is wrongfully initiated by a user (this can be seen as a denial-of-service attack), the BitVM instance needs to be set up again. This can be done right away (excluding the malicious user from $\Nff$), or waiting until the next scheduled $\Nff$ publication. To remain secure, fast finality transactions should halt until the new BitVM instance is set up. Note that users are discouraged from performing such denial-of-service attacks as they would lose their BitVM collateral.
\end{remark}

\begin{remark}
    An alternative approach would be to burn the operator collateral through private key extraction. A pre-signed burn transaction missing only the operator's signature could then be made valid with the extracted private key. Implementing this in Bitcoin today remains, however, an open problem. DAPS-based constructions~\cite{dong2024remote} are not Bitcoin-compatible, and approaches using Schnorr signature nonce reuse remain infeasible as long as opcodes like \texttt{OP\_CAT} or \texttt{OP\_AND}\cite{perez2019double} remain disabled, and opcodes performing arithmetic operations remain limited to 4-byte operands. Finally, accountable assertions~\cite{ruffing2015liar} can only prevent double-spending on two honest users, as they need to communicate the assertions they received to each other. A malicious sender can always send funds to an honest user with an assertion, to then exit onchain with an Ark transaction to itself, as Bitcoin Script cannot enforce including an assertion for that transaction too.
\end{remark}

\subsection{Punishing double-signing}
\label{subsec:bitvm}
The operator will be discouraged from double-signing if any such misbehaviour is guaranteed to lead to a financial punishment. To this end, the fast finality protocol requires $\operator$ and the users in $\Nff$ to set up a BitVM instance. BitVM \cite{linus2025bitvm2, woll2026bitvm3} is a recently introduced paradigm that allows for optimistic verification of arbitrary programs on Bitcoin. In particular, it enables us to make double-spending natively punishable. Using the notation of \cite{linus2025bitvm2}, we set up a \texttt{BitVM2-CORE} proof system between the users in $\Nff$ as provers and $\operator$ as verifier, based on the relation $\mathcal{R}$, defined for a statement $x$ and a witness $w=(\tx{1}{},\tx{2}{})$, with $\tx{i}{}=(\inputs_{i},\witnesses_{i},\outputs_{i})$ for $i=1,2$ as:
\begin{align*}
    &\mathcal{R}(x,(\tx{1}{},\tx{2}{}))=1\iff\\&\exists\out{}{}\in\inputs_1\cap\inputs_2, w_{i}\in\witnesses_{i}\ (i=1,2):\out{}{}.\lockScript(w_i)=\True\\&\wedge x=\operator\wedge\tx{1}{}\neq\tx{2}{}\wedge w_1\neq w_2\wedge\text{$x$ provided signatures for $w_1$ and $w_2$}
\end{align*}
Each party needs to put up a collateral. For $\operator$, this collateral $c_\operator$ should be more than $v+f$, where $f$ is the amount of onchain fees needed to pay for the onchain dispute resolution, for each user in $\Nff$, we set the collateral to $c_U>f$, and large enough to discourage users from initiating a dispute when no double-signing occurred. The result of this setup is a set of presigned transactions. We do not specify the exact transactions, but describe the possible payout scenarios. A dispute can be initiated by any $U^*\in\Nff$ by providing a statement-witness pair $(x,w)$. If $\mathcal{R}(x,w)=1$, i.e., $\operator$ indeed double-signed, all users in $\Nff$ will be able to reclaim $c_U$. Moreover, from the operator collateral $c_\operator$, $f$ is used to pay the onchain fees, and $v$ is burned. If instead $\mathcal{R}(x,w)=0$, i.e., $U^*$ falsely accused $\operator$ of double-signing, everyone except $U^*$ will be able to reclaim their collateral. From the collateral of $U^*$, $f$ is used to pay the onchain fees, the rest is burned.

The BitVM instance enforces that if anyone in $\Nff$ observes a double-sign by $\operator$, the latter can be punished by having its collateral burned (e.g., via \texttt{OP\_RETURN}). In particular, any user in $\Nff$ who is the victim of a double-spend, can punish the malicious operator.

\section{Security Analysis}
\label{sec:security}

This section argues informally why Ark satisfies the properties of §\ref{subsec:properties}. Formal definitions and proofs of the security properties are deferred to Appendix~\ref{app:proofs}. 

\smallskip\noindent\textbf{VTXO Security, Theorem~\ref{thm:vtxo}.} 
(i) \emph{Safety.} Any $\vtxo{}$ in a confirmed, unexpired batch cannot be spent without a valid witness, as long as one cosigner on $\patht(\vtxo{})$ is honest. An honest cosigner never signs conflicting transactions, emulating a covenant: before expiry, a VTXO is either exited unilaterally or spent collaboratively with $\operator$, in which case it is only considered spent when there is a valid Ark or forfeit transaction, requiring a valid witness.
(ii) \emph{Liveness.} For any $\vtxo{}$ in a confirmed batch, if at least one cosigner on $\patht(\vtxo{})$ is honest so that $\vtxo{}$ remains unspent, 
then anyone knowing $\patht(\vtxo{})$ can broadcast it to exit unilaterally. This is safe up to $2k$ (blockchain depth parameter) blocks before batch expiry to ensure all transactions in the path confirm onchain.

\smallskip\noindent\textbf{Ark Atomicity, Theorem~\ref{thm:atomicity}.} 
For any action involving an honest party (user or operator), either the state changes exactly as intended or not at all. We show this by going through each honest Ark action and making sure
honest participants can abort at any time upon detecting misbehaviour and only provide witnesses for transactions they approve, ruling out partial execution.

\smallskip\noindent\textbf{Balance Security, Theorems~\ref{thm:user} and \ref{thm:operator}.}  
(For a user) A user’s balance consists of funds in boarding outputs and VTXOs that can be claimed without $\operator$’s help (formally defined in Appendix~\ref{app:proofs}). Applying VTXO Security ensures that an honest user, holding the necessary paths, can always reclaim these funds via unilateral exit.
(For an operator) $\operator$ funds each commitment but eventually reclaims its inputs once the commitment expires. By atomicity, every coin spent by an honest operator corresponds to at least as much recoverable value in the sweep. Tracking all flows shows that after a phase of serving exits, all operator funds return.

\smallskip\noindent\textbf{Fast Finality, Theorem~\ref{thm:ff}.}
Because of the broadcast network between honest users in $\Nff$, each double-spend to two honest users will be detected and rejected. The only way in which a double-spend can occur after an honest user finalised a payment is if the malicious sender tries to exit unilaterally with either a conflicting VTXO, or an already spent VTXO. In both cases, the honest user will detect this, and respectively burn the operator collateral, or spend the spent VTXO with the corresponding (collaboratively signed) Ark transaction. Because the operator has to fund the collateral, which exceeds the potential gains from double-signing, and since the operator is assumed rational, it will refrain from double-signing.

\smallskip\noindent\textbf{Scalability.} 
(i) \emph{Constant Updates.} A batch can in principle commit to an arbitrarily large VTXT. Only the root of the tree appears onchain as part of the commitment transaction. The connector output can also contain arbitrarily many connectors. Note, however, that to confirm each successive VTXO spend, the new VTXO must appear in its own batch, requiring an additional onchain commitment. 
(ii) \emph{Constant optimistic exit.} With a cooperating operator, a user can exit multiple VTXOs collaboratively as an additional output in a commitment transaction.  
(iii) \emph{Logarithmic pessimistic exit.} A unilateral exit requires posting $\patht(\vtxo{})$, which has length $\mathcal{O}(\log n)$ in a VTXT with $n$ VTXOs as its leaves. If $\vtxo{}$ is an output in a (fast finality) Ark transaction, the unilateral exit costs increase. Indeed, all offchain Ark transactions need to be posted onchain as well in order for $\vtxo{}$ the transaction containing $\vtxo{}$ to be valid.

\section{Implementation and Evaluation}
\label{sec:implementation}

We implemented the Ark protocol as an open-source mainnet system\footnote{\url{https://github.com/ark4fish/ARK}}, incorporating the reset transaction scheme introduced in §\ref{subsec:ark-revisited} to defend against the hostage and spam attacks. Using this implementation, we experimentally evaluate two aspects: the time required to construct a commitment transaction as a function of the number of users, and the onchain footprint of commitments and unilateral exits.

\smallskip\noindent\textbf{Commitment construction time.}
We measure the time to finalize a commitment transaction when $\operator$ receives exactly one batch swap request from each of $n$ users, for $n \in \{2^1, \ldots, 2^7, 200\}$. Table~\ref{tab:times} reports the duration of three phases: (BC)~$\operator$ builds the commitment transaction (one batch, one connector) and distributes it to all users; (SS)~all MuSig2 signing sessions for the virtual transactions in the batch's VTXT are completed; and (FF)~users construct and fully sign forfeit transactions with~$\operator$. Signing the commitment transaction itself took $\operator$ on average $0.023$~seconds.

\begin{table}[h]
\centering
\caption{Execution times $t$ in seconds per value of $n$ for each phase (with $R^2$ for linear regression $t=a_0+a_1 n$, and $p$-value for $H_0\colon a_1=0$ against $H_1\colon a_1\neq 0$).}
\label{tab:times}
\setlength{\tabcolsep}{1.9pt}
\begin{tabular}{c|cccccccc|cc}
\hline $t(n)$  & 2     & 4     & 8     & 16    & 32    & 64    & 128   & 200   & $R^2$ & $p$  \\ \hline
BC & 0.030 & 0.049 & 0.053 & 0.080 & 0.084 & 0.220 & 0.427 & 0.806 & 0.981 & 2.25e-6  \\
SS & 0.190 & 0.194 & 0.200 & 0.219 & 0.333 & 0.388 & 0.707 & 1.075 & 0.991 & 2.01e-7  \\
FF & 0.114 & 0.123 & 0.132 & 0.169 & 0.224 & 0.370 & 0.556 & 0.831 & 0.998 & 4.71e-9  \\ \hline
\end{tabular}
\end{table}

All three phases grow linearly in $n$, consistent with the $\mathcal{O}(n)$ complexity of MuSig2~\cite{nick2020musig2}. At $n = 200$, users must be online simultaneously for approximately $2.7$~seconds. In practice, users submit requests over time and may go offline between rounds; $\operator$ collects pending requests and initiates the commitment protocol once its batching policy triggers, prompting the relevant users to reconnect. A key advantage of Ark is that only users involved in a given update need to participate, allowing $\operator$ to produce smaller, more frequent batches. We note that the experiment operates with minimal network delay; real-world latency would increase these times.

\smallskip\noindent\textbf{Onchain footprint.}
The experiment also produces the commitment transaction and the batch's VTXT, from which we extract transaction sizes via PSBT\footnote{Partially Signed Bitcoin Transaction format.} data (see Appendix~\ref{app:tx}). Examining $\patht(\vtxo{})$ for any VTXO in the batch, we observe consistent sizes: $107$~vB for the leaf transaction (\texttt{P2TR}\footnote{Pay To Taproot output.} VTXO output and one anchor\footnote{\label{fot:anchor}For onchain fee management.}), and $150$~vB for each intermediate virtual transaction (two \texttt{P2TR} outputs and one anchor\footref{fot:anchor}). With fast finality, each additional Ark transaction along the path incurs a comparable cost. The commitment transaction itself (comprising a change output, a batch output, and a connector output, all \texttt{P2TR}) is $197$~vB. As a concrete example, at a fee rate of $6$~sat/vB ($1\bitcoin = 10^8$~sat), unilaterally exiting a VTXO from a batch of $n = 128$ costs $6 \cdot (\lceil \log 128 \rceil \cdot 150 + 107) = 6{,}942$~sat (which is approx.~\$4.58 at the time of writing).

\smallskip\noindent\textbf{Comparison with other protocols.}
Table~\ref{tab:cost-comparison} summarizes how Ark's onchain costs compare to other Bitcoin offchain protocols across three dimensions: commitment size (where applicable), cooperative exit cost, and unilateral exit cost.

\begin{table}[h]
\centering
\caption{Onchain cost comparison across Bitcoin offchain protocols. Here $n$ denotes the batch/clique size, $h$ the number of in-flight HTLCs (for a single channel), and $d$ the depth of the user's UTXO in Spark.}
\label{tab:cost-comparison}
\setlength{\tabcolsep}{4pt}
\begin{tabular}{l|c|c|c}
\hline
Protocol & Commit. (vB) & Coop.\ exit & Unilat. exit (vB) \\ \hline
Ark & $197$ & $+1$ output & $\lceil \log n \rceil \cdot 150 + 107$ \\
Lightning~\cite{bolt3} & --- & $\mathcal{O}(1)$ & $h \cdot 43 + 181$ \\
Spark~\cite{spark2025spark} & --- & $\mathcal{O}(1)$ & $(d+3) \cdot 315$ \\
CoinPool & $\mathcal{O}(1)$ & $\mathcal{O}(1)$ & $\mathcal{O}(1)$ \\
Bitcoin Clique & $\mathcal{O}(1)$ & $+1$ output & $\mathcal{O}(\log n)$ \\
\hline
\end{tabular}
\end{table}

For commitment transactions, only CoinPool and Bitcoin Clique have comparable constructions with a constant number of outputs. Lightning and Spark do not have commitments in the sense that their onchain footprints do not batch an arbitrary number of UTXOs. CoinPool achieves constant cooperative/unilateral exit cost, as the user spends a shared UTXO and creates only two outputs. Bitcoin Clique's unilateral exit cost grows logarithmically in the clique size, analogous to Ark's scaling in the batch size. Both Clique and Ark require one output in a commitment per cooperative exit. Lightning's unilateral exit cost depends on in-flight HTLCs, while Spark's cost scales linearly in the UTXO depth $d$. This is analogous to Ark with fast finality, although Ark transactions can be smaller in size. Moreover, Ark's batch swap mechanism allows users to reset the path length, reducing exit to only the VTXT transactions.

\section{Discussion, Limitations, and Future Work}
\label{sec:discussion}

\smallskip\textbf{Operator Centralisation.}
Ark relies on a single operator to coordinate offchain transactions, manage liquidity, and produce commitments. This introduces a centralisation risk at odds with Bitcoin’s decentralised ethos. While unilateral exit ensures users can always reclaim funds, the operator remains a single point of failure: if unavailable or adversarial, performance degrades and onchain costs rise. Future work could explore multi-operator designs via federation or cryptographic coordination, and examine how the operator’s roles (funding, signing, batching) might be distributed across multiple parties. 

\smallskip\noindent\textbf{Liquidity Requirements.}
The operator must commit upfront capital to fund each batch, recoverable only at batch expiry together with any fees collected. Future work should quantify liquidity needs under realistic market conditions and develop trust-minimised mechanisms for external liquidity providers to fund commitments. This, together with appropriately designed fee structures, could position Ark liquidity provision as a Bitcoin-native, near risk-free return on capital, attracting sufficient supply.

\smallskip\noindent\textbf{Collateral Requirements for Fast Finality.}
Beyond operational liquidity, the fast finality mechanism requires the operator to lock collateral exceeding the total value held by opt-in users, ensuring that double-signing is always unprofitable. The opportunity cost of this lock-up must be covered by fees, and its sustainability under varying adoption levels remains an open question. Future work should analyse fee mechanisms that make the collateral economically viable, and explore alternative deterrence approaches that reduce the required amount, for example by accounting for the operator's loss of future revenue. Honest operator behaviour could also be enforced via TEEs, which may present an acceptable trust trade-off for fast finality users.


\smallskip\noindent\textbf{Fast Finality Assumptions.}
The fast finality mechanism requires users to remain online to monitor for potential double-spends, either offchain or onchain, similar to the online requirement in the Lightning Network \cite{poon2016bitcoin}. Moreover, it requires the operator and participating users to set up a BitVM instance. If new users wish to join with the same security guarantees, a new instance must be created involving all active users; only then can the previous instance be safely retired. Native covenant support and opcodes such as \texttt{OP\_CAT} or \texttt{OP\_AND} could eliminate the need for BitVM by making operator collateral only spendable by a burn transaction and enabling private key extraction by forcing signature nonce reuse through Bitcoin Script.

\smallskip\noindent\textbf{Onchain Cost of Unilateral Exit.}
Unilateral exits require broadcasting $\mathcal{O}(\log n)$ virtual transactions, and fast finality increases this cost further since all preceding Ark transactions must also appear onchain. Small VTXO holders may be priced out. Compared to rollups, which post all transaction data onchain by default, this is an improvement; compared to Lightning, whose exit cost per channel grows only with in-flight HTLCs, it is worse in absolute terms. However, the Ark operator is designed to remain continuously online (similar to a payment channel hub), making unilateral exits less frequent in practice. Future work should explore VTXT designs that reduce exit costs and evaluate policies limiting successive Ark spends before a mandatory batch swap.

\smallskip\noindent\textbf{Miner Extractable Value (MEV).}
The operator controls request scheduling and VTXO grouping, granting sequencing power analogous to a rollup sequencer. However, Ark does not introduce new scripting capabilities: the operator cannot reorder or insert transactions with programmable side effects. Its influence is limited to delaying or censoring requests, which is already possible in Bitcoin itself, and is mitigated by unilateral exit. Richer application logic built on top of Ark (e.g., via BitVM~\cite{aumayr2024bitvm,linus2025bitvm2,woll2026bitvm3}) could introduce application-level MEV, but this remains architecturally sandboxed: Ark handles VTXO transfers, while application semantics are enforced offchain via fraud proofs, and BitVM does not grant the operator additional control over execution order.

\smallskip\noindent\textbf{Bank Run Scenario.}
As with other Layer-2 protocols, a worst-case scenario for Ark is a bank run in which many users exit simultaneously, producing a transaction surge proportional to the number of VTXOs in the affected batches. Given Bitcoin's limited throughput, this may delay exits; users should therefore initiate unilateral exits well before batch expiry to avoid race conditions with operator sweeps. Future work should quantify how quickly a bank run can be resolved onchain as a function of the number of exiting VTXOs. We note, however, that this scenario can be avoided entirely when the operator remains cooperative: as shown in §\ref{subsec:handover}, users can migrate their VTXOs to a new Ark offchain, allowing a controlled wind-down with no unilateral exit footprint.

\smallskip\noindent\textbf{Rational Security Analysis.}
Under a rational miner model~\cite{aumayr2024bitcoin}, Ark becomes vulnerable to timelock bribing~\cite{nadahalli2021timelock}: a colluding operator could pay miners to censor a user's unilateral exit until batch expiry, then sweep the funds. Integrating bribe-resistant constructions~\cite{tsabary2021mad,wadhwa2022he,chung2022ponyta} or deriving tighter time bounds under fine-grained rational models~\cite{avarikioti2025composable} are important future directions. 
More broadly, our analysis assumed the operator is simply present and showed that it cannot violate economic safety (unless miners are rational) and can only censor users into exiting onchain. This censoring power is particularly relevant for batch swaps, which must occur at least once per VTXO per expiry period. In practice, however, operators will be profit-driven entities whose operational and capital costs, including the opportunity cost of locked liquidity and the fees needed to attract external providers, must be covered by user fees. Future work should devise fee mechanisms that make operating an Ark profitable while incentivising responsive behaviour. 

\section{Conclusion}
\label{sec:conclusion}

We presented Ark, the first Bitcoin-compatible commit-chain. Ark enables offchain transactions of virtual UTXOs through an untrusted operator who batches them into succinct onchain commitments. Unlike payment channel networks, Ark allows receivers to onboard entirely offchain and requires only the involved users, rather than all participants, to coordinate each state update. We formally defined the protocol, proved VTXO security, atomicity, and balance security for both users and the operator, and showed how an opt-in fast finality mechanism can eliminate the need to wait for onchain confirmation under a rational operator. Our security analysis uncovered two vulnerabilities in the original protocol, the hostage and spam attacks, which we resolved via the so-called reset transactions that are now integrated into an open-source mainnet implementation. Experimentally, we showed that Ark commits to arbitrarily many VTXOs with a constant $197$~vB onchain footprint, achieves cooperative exits at constant cost, and bounds unilateral exit cost to $\lceil\log n\rceil \cdot 150 + 107$~vB for a batch of $n$ VTXOs, with commitment construction completing in 2.7 seconds for 200 users.

\section*{Acknowledgements}
The original version of the Ark protocol was conceptualised by Burak \cite{burak2023ark}. Our work builds on his idea, which required covenants, making it Bitcoin-compatible. 
This research was partially funded by Ark Labs, by the European Research Council (ERC) under the European Union’s Horizon 2020 research (grant agreement 101141432-BlockSec), by the Austrian Science Fund (FWF) through the SFB SpyCode project F8510-N and F8512-N, and by the WWTF through the projects 10.47379/ICT22045 and 10.47379/ICT25056.
Finally, we thank Christos Stefo and Yuheng Wang for insightful discussions.

\bibliographystyle{ACM-Reference-Format}
\bibliography{references}

\appendix


\renewcommand{\thefigure}{A.\arabic{figure}}
\renewcommand{\thetable}{A.\arabic{table}}
\setcounter{figure}{0}
\setcounter{table}{0}

\section{Global Parameters, Notation and Commitment Transaction Construction}
\begin{table}[h]
\caption{Global parameters and Notation}
\label{tab:notation}
\centering
\begin{tabular}{c|l}
\hline
$P,\operator$          & Ark user and Ark operator         \\\hline
$\tx{}{},\out{}{}$     & Transaction, transaction output   \\\hline
$\vtxo{}$              & Virtual transaction output        \\\hline
$\C{}{Q}$, $\C{-k}{Q}$ &
\begin{tabular}[c]{@{}l@{}}Local view of party $Q$, with and without latest \\ $k$ blocks, denote inclusion by $\tx{}{}\in\C{}{P}$ \\ (or $\out{}{}\in\C{}{P}$ for \emph{unspent} outputs)\end{tabular}

       \\\hline
$t_b$         & \begin{tabular}[c]{@{}l@{}}Boarding timeout: unilateral recovery delay \\ for boarding funds\end{tabular} \\\hline

$t_e$         & \begin{tabular}[c]{@{}l@{}}Batch expiry time: minimum number of blocks \\ before batch expires\end{tabular}   \\\hline

$\tv$         & \begin{tabular}[c]{@{}l@{}}VTXO unilateral delay: minimum delay of \\ a VTXO unilateral script path compared \\ to any collaborative script path\end{tabular}     \\\hline
$\varepsilon$ & Dust-sized value of anchor outputs             \\\hline
$k,u$ & Ledger depth/wait parameters as defined in~\cite{garay2024bitcoin} \\\hline
\begin{tabular}[c]{@{}l@{}}$\C{}{Q}(r)$,\\ $\C{-k}{Q}(r)$\end{tabular} & \begin{tabular}[c]{@{}l@{}}Local view of party $Q$ at round $r$, with and \\ without latest $k$ blocks, as in \cite{garay2024bitcoin}, used only \\ in Appendix~\ref{app:proofs}, denote inclusion just as before \end{tabular}

       \\\hline
$\C{}{1}\preceq\C{}{2}$ & $\C{}{1}$ is a prefix of $\C{}{2}$ \\\hline

\end{tabular}
\end{table}

\begin{figure*}[b]
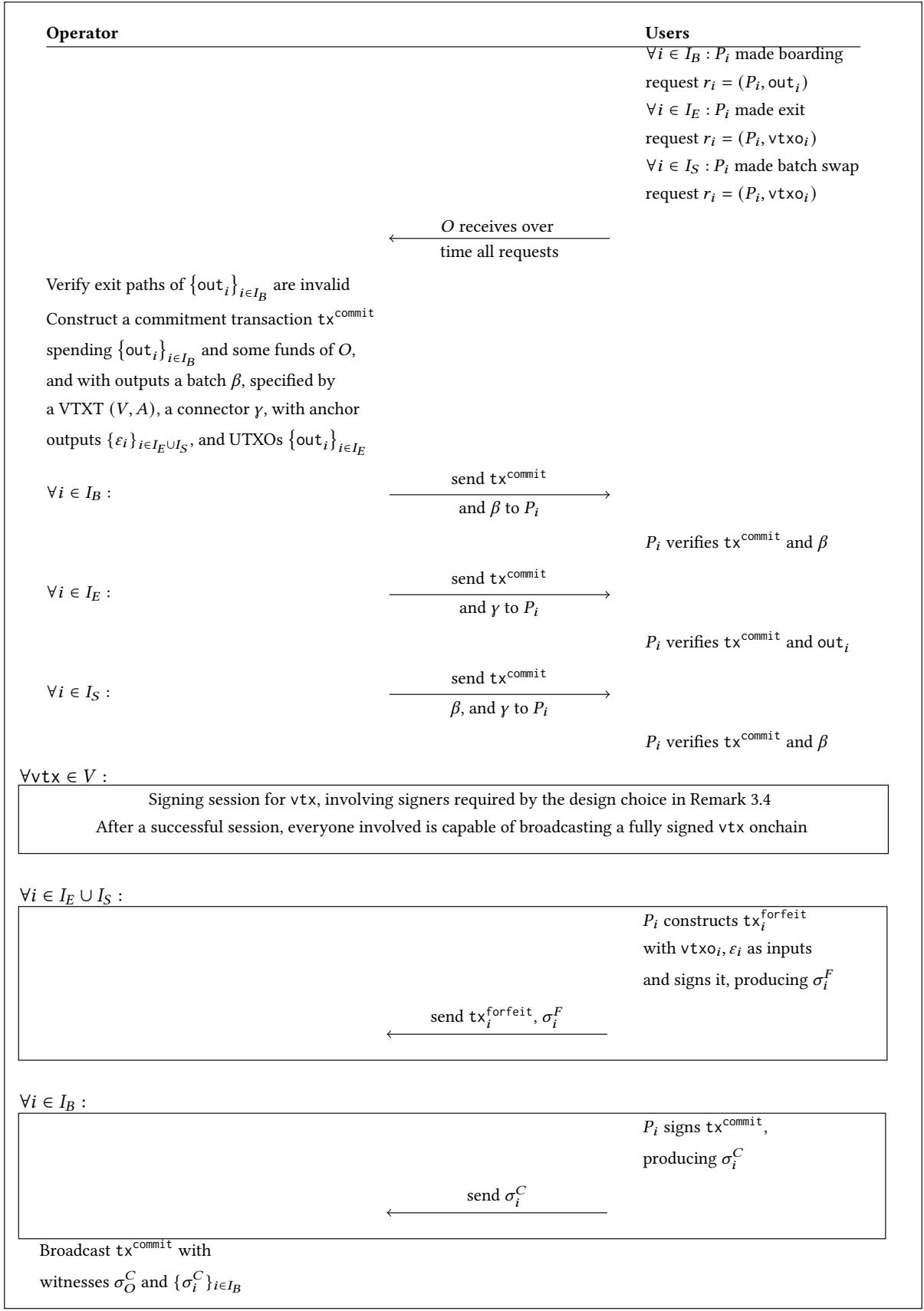

    \centering
    \resizebox{0.9\textwidth}{!}{
    \begin{bbrenv}{A}
    \begin{bbrbox}
        \pseudocodeblock{
            \textbf{Operator} \< \< \textbf{Users} \\[][\hline]
            \< \< \text{$\forall i\in I_B:$ $P_i$ made boarding} \\
            \< \< \text{request $r_i=(P_i,\out{i}{})$} \\
            \< \< \text{$\forall i\in I_E:$ $P_i$ made exit} \\
            \< \< \text{request $r_i=(P_i,\vtxo{i})$} \\
            \< \< \text{$\forall i\in I_S:$ $P_i$ made batch swap} \\
            \< \< \text{request $r_i=(P_i,\vtxo{i})$} \\
            \< \sendmessageleft{top = \text{$\operator$ receives over }, bottom = \text{time all requests}} \< \\
            \text{Verify exit paths of $\qty{\out{i}{}}_{i\in I_B}$ are invalid} \< \< \\
            \text{Construct a commitment transaction $\tx{}{commit}$} \< \< \\
            \text{spending $\qty{\out{i}{}}_{i\in I_B}$ and some funds of $\operator$,} \< \< \\
            \text{and with outputs a batch $\beta$, specified by} \< \< \\
            \text{a VTXT $(V,A)$, a connector $\gamma$, with anchor} \< \< \\
            \text{outputs $\qty{\varepsilon_i}_{i\in I_E\cup I_S}$, and UTXOs $\qty{\out{i}{}}_{i\in I_E}$} \< \< \\
            \forall i\in I_B: \< \sendmessageright{top = \text{send $\tx{}{commit}$}, bottom = \text{and $\beta$ to $P_i$}} \< \\
            \< \< \text{$P_i$ verifies $\tx{}{commit}$ and $\beta$}\\
            \forall i\in I_E: \< \sendmessageright{top = \text{send $\tx{}{commit}$}, bottom = \text{and $\gamma$ to $P_i$}} \< \\
            \< \< \text{$P_i$ verifies $\tx{}{commit}$ and $\out{i}{}$}\\
            \forall i\in I_S: \< \sendmessageright{top = \text{send $\tx{}{commit}$}, bottom = \text{$\beta$, and $\gamma$ to $P_i$}} \< \\ 
            \< \< \text{$P_i$ verifies $\tx{}{commit}$ and $\beta$}\\
        }
        \begin{bbrenv}{B}
            \begin{bbrbox}[name = \text{$\forall \texttt{vtx}\in V:$}, namepos = top left]
                \pseudocodeblock{
                \<\quad\quad\quad \text{Signing session for \texttt{vtx}, involving signers required by the design choice in Remark~\ref{rem:signing}} \< \\ \quad\quad\quad\enspace
                \< \text{After a successful session, everyone involved is capable of broadcasting a fully signed \texttt{vtx} onchain} \< \quad\quad\quad\enspace
                }
            \end{bbrbox}
        \end{bbrenv}
        \pseudocodeblock{\< \<}
                \begin{bbrenv}{C}
            \begin{bbrbox}[name = \text{$\forall i\in I_E\cup I_S:$}, namepos = top left]
                \pseudocodeblock{
                \quad\quad\quad\quad\quad\quad\quad\quad\quad\quad\quad\quad\quad\quad\quad\quad\quad\quad\quad\ \< \< \text{$P_i$ constructs $\tx{i}{forfeit}$}\ \ \\
            \< \< \text{with $\vtxo{i},\varepsilon_i$ as inputs}\\
            \< \< \text{and signs it, producing $\sigma^F_i$} \hspace{2em}\\
            \< \sendmessageleft{top = \text{send $\tx{i}{forfeit}$, $\sigma^F_i$}} \<
            }
            \end{bbrbox}
        \end{bbrenv}
        \pseudocodeblock{\< \<}
        \begin{bbrenv}{D}
            \begin{bbrbox}[name = \text{$\forall i\in I_B:$}, namepos = top left]
                \pseudocodeblock{
                \quad\quad\quad\quad\quad\quad\quad\quad\quad\quad\quad\quad\quad\quad\quad\quad\quad\quad\quad\ \< \< \text{$P_i$ signs $\tx{}{commit}$,}\\
                \< \< \text{producing $\sigma^C_i$} \quad\quad\quad\quad\quad\quad\quad\\
            \< \sendmessageleft{top = \text{send $\sigma^C_i$}} \<
                }
            \end{bbrbox}
        \end{bbrenv}
        \pseudocodeblock{
            \text{Broadcast $\tx{}{commit}$ with} \< \< \\
            \text{witnesses $\sigma_\operator^C$ and $\{\sigma^C_i\}_{i\in I_B}$} \< \< \quad\quad\quad\quad\quad\quad\quad\quad\quad\quad\quad\quad\quad\quad\quad\quad\quad\quad\quad\quad\quad\quad\quad\quad\quad\quad\quad\quad\quad\quad\quad\quad\quad\quad\quad
        }
    \end{bbrbox}
    \end{bbrenv}
    }
    \caption{High-level construction of a commitment transaction from requests indexed by $I_B$, $I_E$, $I_S$ for boarding, exit, and batch swap request, respectively. 
    }
    \label{op:op}
\end{figure*}





\section{Commitment transactions and VTXTs}
\label{app:tx}
As an example, let us look at the commitment transaction created by our experiment in §\ref{sec:implementation} for $n=2$. The commitment transaction can be found on Mutinynet at: \url{https://mutinynet.com/} by searching for the transaction ID

\noindent\text{\texttt{3fd6ddc344170c6ba90abb5a6ae7e4aaa1ca94c2}}\\\text{\texttt{78de1501aba34fb7738740d3}.} 

One can see one large input, as well as one large change output, together with a batch output of $10000$~sat and a connector output of $660$~sat. The connector output consists of two dust value ($330$~sat) outputs. The batch commits to two VTXOs, one of $8000$~sat and one of $2000$~sat. Below, we share the transactions in PSBT format for conciseness. These can easily be decoded into JSON via a website such as \url{https://chainquery.com/bitcoin-cli/decodepsbt}:
\begin{itemize}[leftmargin=*]
    \item VTXT Root, spending the batch output and having 2 P2TR outputs for the leaf transactions, and one anchor output: {\ttfamily\seqsplit{ cHNidP8BAJYDAAAAAdNAh3O3T6OrARXeeMKUyqGq5OdqWrsKqWsMF0TD3dY/AAAAAAD/////A0AfAAAAAAAAIlEg1uuTP5t4SvcKMZtq4x3b6J5dADh8hC4t+79wnvdQBkzQBwAAAAAAACJRIHaEDwnIyDSOyHFf/jLuenc2dJ1lRTh4LnlMw7wVjN5lAAAAAAAAAAAEUQJOcwAAAAAAARNAyBga5bp6IDKv35ReC9Nqn0lDvB6Km2uitFQC88HNBkRvrSBQotq8004h1E9+pNUVM9VO6mg8q74vRJDwvnHZewxjb3NpZ25lcgAAAAAhAlWozfpWDEPyugdJ95dGUrYTUFjZVaqhCezPp8QKtyUSDGNvc2lnbmVyAAAAASEChNDc7qODfpNp/P7KaCotiCExEw6RXcLIgJkm1HtwFO4MY29zaWduZXIAAAACIQKbjuy0+nlO9A5YxcUwsNP7V/Tt56lculCTrhoZXbbWugZleHBpcnkDUztAAAAAAA==}}
    \item Leaf 1, having the VTXO with value $8000$~sat as output, together with an anchor: {\ttfamily\seqsplit{ cHNidP8BAGsDAAAAAQdO2R6SLbqa1+KtulJ84a/0qnMC5t0XNkzacdHcdVSTAAAAAAD/////AkAfAAAAAAAAIlEg7N6TSBPPzHXoOwvg7b5YxseoaDNxayBBQEmnhRvXLroAAAAAAAAAAARRAk5zAAAAAAABE0AYwVd0VW4mvU6RT1pdv+Nz1wAZ8nXmbxmsSooq2dj0VbHvmGe3sC7YzyXLEPhJ9PpLYnflBFMSSog/e/Cqy85pDGNvc2lnbmVyAAAAACECVajN+lYMQ/K6B0n3l0ZSthNQWNlVqqEJ7M+nxAq3JRIMY29zaWduZXIAAAABIQKE0Nzuo4N+k2n8/spoKi2IITETDpFdwsiAmSbUe3AU7gZleHBpcnkDUztAAAAA}}
    \item Leaf 2, having the VTXO with value $2000$~sat as output, together with an anchor : {\ttfamily\seqsplit{ cHNidP8BAGsDAAAAAQdO2R6SLbqa1+KtulJ84a/0qnMC5t0XNkzacdHcdVSTAQAAAAD/////AtAHAAAAAAAAIlEgk4Bo7yqsimoNGhA5XrjZnY0HpJU6dY9coXomK4rSHi8AAAAAAAAAAARRAk5zAAAAAAABE0ASijnyUYyN6+A+ab+qpb7WUmz8ngKWaLpBqBiVcYUTCjuNwkERVqaOIVtsr/iETZ00kpEdae02Pvzwt+Jt7ttiDGNvc2lnbmVyAAAAACECVajN+lYMQ/K6B0n3l0ZSthNQWNlVqqEJ7M+nxAq3JRIMY29zaWduZXIAAAABIQKbjuy0+nlO9A5YxcUwsNP7V/Tt56lculCTrhoZXbbWugZleHBpcnkDUztAAAAA}}
\end{itemize}

\section{Detailed Protocol Specification}
\label{app:protocol}
We start off by listing some conventions and simplifying assumptions:
\begin{itemize}[leftmargin=*]
    \item We denote by $\batch(\vtxo{})$ the batch output containing $\vtxo{}$.
    \item A party $P$ can represents a single user or multiple users. We assume for simplicity that one entity communicates with the operator and can organise the users behind it to provide the necessary witnesses, and we assume that the view $\C{}{P}$ of this entity is well-defined (for example that the multiple users represented by $P$ all agree on the state of the blockchain).  
    \item For any VTXO that has multiple spending conditions where different parties can provide witnesses satisfying one of these conditions, we assume that all these parties are aware of the state of this VTXO, i.e., whether for a given Ark state $\Sigma=(C,F,S)$, this VTXO is in $C$, $F$, or $S$ (see Definition~\ref{def:ark-state}. This is relevant for request routines (Routine~\ref{alg:boarding}, \ref{alg:batch-swap}, \ref{alg:exit}, and \ref{alg:ark-tx}), to ensure that honest users never make a request spending an already spent VTXO. 
\end{itemize}
Furthermore, recall that the lists $\vtxos$ and $\vtxos'$ should be understood here as lists of tuples $(\val,\vtxoLockScript)$, such that $\operator$ can read off all spending paths from $\vtxoLockScript$. Of course, the locking script itself present in any transaction would only contain the tweaked public key committing to these spending paths. 

The protocol setup is minimal and presented in Routine~\ref{alg:setup}. It only requires the Operator $\operator$ to let users know how they can create a boarding address, which boils down to $\operator$ sharing its public key $\pk_\operator$ and some other Ark-specific parameters, for example on a website. Apart from that, a number of empty lists are instantiated for the operator to keep track of user requests and VTXOs. Finally, $\operator$ defines a batching policy: a set of rules that specifies when and how $\operator$ should construct a commitment transaction. This policy may also be published.

Users who want to board the Ark can now do so by sending a boarding request to the operator, as described in Routine~\ref{alg:boarding}. This boarding request specifies a list of VTXOs which the user in question wants to add to the Ark.

The Ark operator can now verify this boarding request, making sure that the UTXO specified in the request is not already spent, checking the locking script is of the correct form (in particular that there are no other spending paths), amongst others. This is specified in Routine~\ref{alg:verify-boarding}.

We can now also specify the other types of request users can make to the operator: batch swap requests (Routine~\ref{alg:batch-swap}), exit requests (Routine~\ref{alg:exit}), Ark transaction requests (Routine~\ref{alg:ark-tx}), and their respective verification procedures (Routines~\ref{alg:verify-swap}, \ref{alg:verify-exit} and \ref{alg:verify-ark-tx-request}). For the Ark transaction, we also write out the procedure to send an Ark transaction to the receiver (Routine~\ref{alg:send-ark-tx}), and what checks the receiver should subsequently perform (Routine~\ref{alg:verify-ark-tx}).

Out of all of the received requests, the operator can select a number of them to construct a commitment, as detailed in Routine~\ref{alg:commitment-tx}. The batch outputs of this transaction are specified by a batch template, which is a list of VTXTs (one for each batch output). How these VTXTs are constructed exactly is specified in the batching policy from Routine~\ref{alg:setup} and is up to the operator. It employs its own rules and optimisations to construct the VTXTs that go beyond this protocol specification. Of course, the VTXTs are constructed such that the obtained batch outputs satisfy Definition~\ref{def:batch}. Similarly, the connector outputs are specified by a connector template: a list of connector outputs, and a connector function $\gamma$, mapping each VTXO that is being batch swapped or exited from to a leaf of one of the VTXTs in the connector template.

The commitment produced by $\operator$ is still missing signatures. $\operator$ has to communicate with all the relevant parties, who will at least each individually check the commitment is valid (Routines~\ref{alg:verify-path}, \ref{alg:verify-connector} and \ref{alg:verify-commitment}). Furthermore, parties who sent a boarding request or batch swap request should make sure their requested VTXOs are included in the batch template, and parties who sent a batch swap or exit request should find the anchor output in the connector template associated with their old VTXOs that are being swapped or exited.

If every party convinced themselves that the commitment is correct, parties who sent a batch swap or exit request will produce forfeit transactions (Routine~\ref{alg:forfeit}) that are only valid once the commitment holding the specific, corresponding anchor output is available onchain, and multiple signing sessions (Routine~\ref{alg:musig}) will happen to produce witnesses for the commitment itself, and for the virtual transactions contained in the different batch outputs.

All the steps to go from user requests to a fully signed commitment are gathered in Routine~\ref{alg:process-requests}. At the end, the commitment can be broadcast, and the operator can update the various lists tracking VTXOs and open requests. This is specified in Routine~\ref{alg:update-lists}.

Finally, we specify the behaviour of honest parties, apart from how they should act when building commitment transactions. For the operator, this is done in Routine~\ref{alg:operator}. For example, the operator has to monitor the chain for exited VTXOs and possibly respond by broadcasting the necessary reset and/or forfeit transactions. Other honest parties also have some tasks to perform, however not continuously. In Routine~\ref{alg:party}, we describe what a party wanting to spend a VTXO should do: first try to spend collaboratively, and if at some point this has still failed, exit unilaterally (Routine~\ref{alg:unilateral-exit}) and spend the created UTXO. In Routine~\ref{alg:recipient}, we give the steps the recipient of an Ark transaction should take.

\clearpage

\begin{algorithm*}
    \captionsetup{name=Routine}
    \caption{Setup by Operator $\operator$ with public key $\pk_\operator$.}
    \label{alg:setup}
    \begin{algorithmic}[1]
    \State $\operator$ publishes its public key $\pk_\operator$, as well as the following Ark parameters:
    \begin{itemize}
        \item $t_b$: Timeout period for boarding transaction output.
        \item $t_e$: Minimum amount of blocks that should pass before a batch expires.
        \item $\tv$: Minimum delay in VTXO unilateral script path. 
        \item $t_r$: Time after which if the commitment transaction is still not included the operator will retry.
    \end{itemize}
    \State $\operator$ also instantiates the variables:
    \begin{itemize}
        \item $\toBoard:=[~]$, the list of all boarding requests that have been verified by $\operator$ and can be added to a commitment transaction.
        \item $\toBatchSwap:=[~]$, the list of all batch swap requests that have been verified by $\operator$ and can be added to a commitment transaction.
        \item $\toExit:=[~]$, the list of all exit requests that have been verified by $\operator$ and can be added to a commitment transaction.
        \item $\unconfirmed:=[~]$, the list of all VTXOs that are part of the batch output of a transaction which has not yet been confirmed onchain.
        \item $\confirmedVTXO:=[~]$, the list of all VTXOs that are part of the batch output of a commitment transaction which has been confirmed onchain.
        \item $\confirmedBatches:=[~]$, the list of all batch outputs that are part of a commitment transaction which has been confirmed onchain.
        \item $\expired:=[~]$, the list all VTXOs that are part of a batch output of which the expiry timelock passed.
        \item $\unconfirmedSpent:=[~]$, the list of all VTXOs that have been spent, either by an Ark transaction, or by a batch swap or exit in an unconfirmed commitment transaction.
        \item $\spent:=[~]$, the list of all pairs $(\vtxo{},\tx{}{})$ where $\vtxo{}$ could be spent onchain by a valid reset or forfeit transaction $\tx{}{}$.
        \item $\replaced:=[~]$, the list of all VTXOs  that are part of the batch output of a transaction which has been double-spent or invalidated.   
        \item $\preSpent:=[~]$, the list of all VTXOs and boarding outputs spent in a yet unconfirmed boarding, batch swap, or exit request.
        \item $\preConfirmed:=[~]$, the list of all VTXOs that are output from Ark transactions, but not yet batch swapped.
        \item $\unconfirmedToBoard:=[~]$, $\unconfirmedToBatchSwap:=[~]$, and $\unconfirmedToExit:=[~]$, the lists of all tuples of boarding, batch swap, and exit requests, respectively, with the fully signed, unconfirmed commitment transaction they have been processed in.
        \item $\confirmedToBoard:=[~]$, $\confirmedToBatchSwap:=[~]$, $\confirmedToExit:=[~]$, the lists of all tuples of boarding, batch swap, and exit requests, respectively, with the fully signed, confirmed commitment transaction they have been processed in.
    \end{itemize}
    \State Finally, $\operator$ specifies a \emph{batching policy}. We abstract this policy away into a \emph{batching black box} $\beta$, which deterministically maps a given state of the lists $\toBoard$, $\toBatchSwap$ and $\toExit$ of available, verified requests to lists $\boardings$, $\batchSwaps$ and $\exits$ of requests included in the commitment transaction, as well as templates $\batchTemplate$, $\connectorTemplate$ and $\gamma$ that specify how the VTXOs and anchor outputs will be arranged in batch and connector outputs of the commitment transaction. This policy may also include limits such as the maximum number of outputs per virtual transaction, or the maximum tree depth.
    
    Note that $\beta$ can also be a function of other external factors, such as the current time or onchain fees. For simplicity, we assume $\operator$ shares its lists $\toBoard$, $\toBatchSwap$ and $\toExit$ with $\beta$, listens to $\beta$ and starts building a commitment transaction as soon as $\beta$ outputs a tuple of the form $(\boardings,\batchSwaps,\exits,\batchTemplate,\connectorTemplate,\gamma)$. 
    \end{algorithmic}
\end{algorithm*}

\begin{algorithm*}
    \captionsetup{name=Routine}
    \caption{Boarding request to $\operator$ for VTXO list $\vtxos'$ with cosigner set $N$ by a party $P$.}
    \label{alg:boarding}
    \begin{algorithmic}[1]
    \Require A list of UTXOs $[\out{j}{}]_{j\in J}\subseteq\C{-k}{P}$ for which $\sum_{j\in J}\out{j}{}.\val\geq v:=\sum_{\vtxo{}'\in\vtxos'}\vtxo{}'.\val$, and such that for each $j\in J$, $P$ can provide a witness $w_j$ that allows $\out{j}{}$ to be spent for the boarding transaction. $P$ should also know the operator public key $\pk_\operator$. Finally, $P$ should have sets $S_C$ and $S_U$ of collaborative and unilateral script paths, as well as an internal public key $\pk_I$ that is a NUMS (Nothing Up My Sleeve) point (such as the one recommended in BIP 341 \cite{wuille2020bip341}), in order for the key path to be unspendable. $P$ should be able to provide witnesses for these script paths (collaborating with $\operator$ when needed).
    \State $P$ constructs the boarding transaction $\tx{P}{board}$, with $\tx{P}{board}.\inputs = [\out{j}{}]_{j\in J}$, $\tx{P}{board}.\witnesses = \qty[w_j]_{j\in J}$, and $\tx{P}{board}.\outputs=\qty[\out{P}{board}]$, where $\out{P}{board}.\val=v$ and \label{rtn:boarding-out}
    \begin{equation*}
         \out{P}{board}.\lockScript = \texttt{Taproot}\big(\texttt{False};S_C,S_U\big).
    \end{equation*}
    \State $P$ submits $\tx{P}{board}$ to $\btc$.
    \State Once $\tx{P}{board}\in\C{-k}{P}$, $P$ sends to $\operator$ the boarding request
    \begin{equation*}
        r = \qty(\text{``boarding: ''},P,N,\out{P}{board},\pk_I,S_C,S_U,\vtxos').
    \end{equation*}
    \end{algorithmic}
\end{algorithm*}

\begin{algorithm*}
    \captionsetup{name=Routine}
    \caption{\textbf{\texttt{verifyBoardingRequest($r$)}:} Verify boarding request $r=\qty(\text{``boarding: ''},P,N,\out{}{},pk_I,S_C,S_U,\vtxos')$.}
    \label{alg:verify-boarding}
    \begin{algorithmic}[1]
    \State $\operator$ checks that $\out{}{}\in\C{-k}{\operator}$ and that $\out{}{}\notin\preSpent$.\label{rtn:verify-1}
    \State $\operator$ checks that each script path in $S_C$ requires a signature from $\operator$, and that each script path in $S_U$ has a relative timelock of at least $t_b$.
    \State Given $\pk_I$, $S_C$, $S_U$, $\operator$ makes sure that the locking script of $\out{}{}$ indeed commits to the correct script paths and unspendable key path. \label{rtn:verify-3}
    \State $\operator$ checks that $\out{}{}.\val\geq\sum_{\vtxo{}'\in\vtxos'}\vtxo{}'.\val$. \label{rtn:verify-boarding-funds}
    \For{$\vtxo{}'\in\vtxos'$}
    \State $\operator$ checks that $\vtxo{}'$ satisfies Definition~\ref{def:vtxo}.
    \EndFor \label{rtn:verify-6}
    \If{all of the checks in Lines~\ref{rtn:verify-1} - \ref{rtn:verify-6} are successful}
    \State Add $(P,N,\out{}{},\vtxos')$ to $\toBoard$.
    \State Add $\out{}{}$ to $\preSpent$.
    \EndIf
    \end{algorithmic}
\end{algorithm*}

\begin{algorithm*}
    \captionsetup{name=Routine}
    \caption{Batch swap request to $\operator$ of a VTXO list $\vtxos$ into a VTXO list $\vtxos'$ with cosigner set $N$ by $P$.}
    \label{alg:batch-swap}
    \begin{algorithmic}[1]
    \Require VTXO lists $\vtxos\subseteq C\cup F$ and $\vtxos'$ such that for each $\vtxo{}\in\vtxos$, $\batch(\vtxo{})$ exists and is in $\C{-k}{P}$, $\vtxo{}$ can be spent collaboratively by $P$ and the operator, and such that $\sum_{\vtxo{}\in\vtxos}\vtxo{}.\val\geq \sum_{\vtxo{}'\in\vtxos'}\vtxo{}'.\val$. 
    \State $P$ sends to $\operator$ the batch swap request
    \begin{equation*}
        r=\qty(\text{``batch swap: ''}, P, N, \vtxos, \vtxos').
    \end{equation*}
    \end{algorithmic}
\end{algorithm*}

\begin{algorithm*}
    \captionsetup{name=Routine}
    \caption{Exit request to $\operator$ of a VTXO list $\vtxos$ into a UTXO list $\utxos$ by $P$.}
    \label{alg:exit}
    \begin{algorithmic}[1]
    \Require A VTXO list $\vtxos\subseteq C\cup F$ such that for each $\vtxo{}\in\vtxos$, $\batch(\vtxo{})$ exists and is in $\C{-k}{P}$, and $P$ is able to provide a witness to spend $\vtxo{}$ via a collaborative VTXO script path. Also, a UTXO list $\utxos$ such that $\sum_{\vtxo{}\in\vtxos}\vtxo{}.\val\geq \sum_{\out{}{}\in\utxos}\out{}{}.\val$. 
    \State $P$ sends to $\operator$ the exit request
    \begin{equation*}
        r=\qty(\text{``exit: ''}, P,\vtxos,\utxos).
    \end{equation*}
    \end{algorithmic}
\end{algorithm*}

\begin{algorithm*}
    \captionsetup{name=Routine}
    \caption{\textbf{\texttt{musig($\tx{}{}$,$N$)}:} Perform a \texttt{MuSig2} signing session with public key set $N$ for transaction $\tx{}{}$.}
    \label{alg:musig}
    \begin{algorithmic}[1]
    \State The signers in $N$ perform a \texttt{MuSig2} signing session, such that after successful completion every signer has the signature of $\tx{}{}$ under the aggregate public key $\bigoplus_{\pk\in N}\pk$.
    \end{algorithmic}
\end{algorithm*}

\begin{algorithm*}
    \captionsetup{name=Routine}
    \caption{Ark transaction request to $\operator$ of a VTXO list $\vtxos\subseteq C$ into a VTXO list $\vtxos'$ by a party $P$.}
    \label{alg:ark-tx}
    \begin{algorithmic}[1]
    \Require A VTXO list $\vtxos$ such that for each $\vtxo{}\in\vtxos$, $\batch(\vtxo{})$ exists and is in $\C{-k}{P}$, and $P$ is able to cooperate with the operator to provide a witness $w_{\vtxo{}}$ to spend $\vtxo{}$ via a collaborative VTXO script path $s_{C,\vtxo{}}$. Also, a VTXO list $\vtxos'$, such that $\sum_{\vtxo{}\in\vtxos}\vtxo{}.\val\geq\sum_{\vtxo{}'\in\vtxos'}\vtxo{}'.\val$.
    \State $P$ constructs for each $\vtxo{}\in\vtxos$ a reset transaction $\tx{P,\vtxo{}}{re}$, with $\tx{P,\vtxo{}}{re}.\inputs = \vtxos$, $\tx{P,\vtxo{}}{re}.\witnesses = \qty[*]_{\vtxo{}\in\vtxos}$, and $\tx{P,\vtxo{}}{re}.\outputs=\out{\vtxo{}}{re}$, with $\out{\vtxo{}}{re}.\val=\vtxo{}.\val$ and $\out{\vtxo{}}{re}.\lockScript=\texttt{Taproot}(\False;$ $s_{C,\vtxo{}},$ $\checkSig{pk_\operator}\wedge\absTimelock{T_{e,\vtxo{}}}))$, where $T_{e,\vtxo{}}$ is the block height at which the batch containing $\vtxo{}$ expires. 
    \State $P$ constructs an Ark transaction $\tx{P}{ark}$, with $\tx{P}{ark}.\inputs = \qty{\out{\vtxo{}}{re}}_{\vtxo{}\in\vtxos}$, $\tx{P}{ark}.\witnesses = \qty[*]_{\vtxo{}\in\vtxos}$, and $\tx{P}{ark}.\outputs=\vtxos'$.
    \State $P$ sends to $\operator$ the Ark transaction request
    \begin{equation*}
        r=\qty(\text{``Ark: ''}, P, \qty{\tx{P,\vtxo{}}{re}}_{\vtxo{}\in\vtxos}, \tx{P}{ark}, \vtxos, \vtxos').
    \end{equation*}
    \end{algorithmic}
\end{algorithm*}

\begin{algorithm*}
    \captionsetup{name=Routine}
    \caption{\textbf{\texttt{verifyBatchSwapRequest($r$)}:} Verify batch swap request $r=\qty(\text{``batch swap: ''},P,N,\vtxos,\vtxos')$.}
    \label{alg:verify-swap}
    \begin{algorithmic}[1]
    \For{$\vtxo{}\in\vtxos$} \label{rtn:verify-swap-1}
    \State $\operator$ checks that $\vtxo{}\in\confirmedVTXO\cup\preConfirmed$.
    \State $\operator$ checks that $\vtxo{}\notin\preSpent$.
    \EndFor \label{rtn:verify-swap-4}
    \For{$\vtxo{}'\in\vtxos'$}
    \State $\operator$ makes sure that $\vtxo{}'$ satisfies Definition~\ref{def:vtxo}.
    \EndFor
    \State $\operator$ checks that $\sum_{\vtxo{}\in\vtxos}\vtxo{}.\val\geq\sum_{\vtxo{}'\in\vtxos'}\vtxo{}'.\val$. \label{rtn:verify-swap-7}
    \If{all of the checks in Lines~\ref{rtn:verify-swap-1} - \ref{rtn:verify-swap-7} are successful}
    \State Add $(P,N,\vtxos,\vtxos')$ to $\toBatchSwap$.
    \State Add each $\vtxo{}\in\vtxos$ to $\preSpent$.
    \EndIf
    \end{algorithmic}
\end{algorithm*}

\begin{algorithm*}
    \captionsetup{name=Routine}
    \caption{\textbf{\texttt{verifyExitRequest($r$)}:} Verify exit request $r=\qty(\text{``exit: ''},P,\vtxos,\utxos)$.}
    \label{alg:verify-exit}
    \begin{algorithmic}[1]
    \For{$\vtxo{}\in\vtxos$} \label{rtn:verify-exit-1}
    \State $\operator$ checks that $\vtxo{}\in\confirmedVTXO\cup\preConfirmed$ 
    \State $\operator$ checks that $\vtxo{}\notin\preSpent$.
    \EndFor \label{rtn:verify-exit-4}
    \State $\operator$ checks that $\sum_{\vtxo{}\in\vtxos}\vtxo{}.\val\geq\sum_{\out{}{}\in\utxos}\out{}{}.\val$. \label{rtn:verify-exit-5}
    \If{all of the checks in Lines~\ref{rtn:verify-exit-1} - \ref{rtn:verify-exit-5} are successful}
    \State Add $(P,\varnothing,\vtxos,\utxos)$ to $\toExit$.
    \State Add each $\vtxo{}\in\vtxos$ to $\preSpent$.
    \EndIf
    \end{algorithmic}
\end{algorithm*}

\begin{algorithm*}
    \captionsetup{name=Routine}
    \caption{\textbf{\texttt{verifyArkTxRequest($r$)}:} Verify Ark transaction request $r=\qty(\text{``Ark: ''}, P, R, \tx{}{ark}, \vtxos, \vtxos')$.}
    \label{alg:verify-ark-tx-request}
    \begin{algorithmic}[1]
    \For{$\vtxo{}\in\vtxos$} \label{rtn:verify-ark-1}
    \State $\operator$ checks that $\vtxo{}\in\confirmedVTXO$ and that $\vtxo{}\notin\preSpent$.
    \State $\operator$ verifies that there is a $\tx{}{re}\in R$ that can claim $\vtxo{}$ at the corresponding batch's expiry.
    \State $\operator$ cooperates with $P$ to produce a valid witness $w_{\vtxo{}}^\texttt{ark}$ for each input of $\tx{}{ark}$.
    \EndFor \label{rtn:verify-ark-4}
    \For{$\vtxo{}'\in\vtxos'$}
    \State $\operator$ makes sure that $\vtxo{}'$ satisfies Definition~\ref{def:vtxo}.
    \EndFor
    \State $\operator$ checks that $\sum_{\vtxo{}\in\vtxos}\vtxo{}.\val\geq\sum_{\vtxo{}'\in\vtxos'}\vtxo{}'.\val$. \label{rtn:verify-ark-7}
    \If{all of the checks in Lines~\ref{rtn:verify-swap-1} - \ref{rtn:verify-ark-7} are successful}
    \State Add each $\vtxo{}'\in\vtxos'$ (as an output of $\tx{}{}$) to $\preConfirmed$.
    \State Cooperate with $P$ to create a witness $w_{\tx{}{}}^\texttt{re}$ for each $\tx{}{}\in R$.
    \State Add $(\vtxos,\tx{}{})$ to $\spent$, where $\tx{}{}\in R$ is the corresponding fully signed reset transaction.
    \EndIf
    \end{algorithmic}
\end{algorithm*}

\begin{algorithm*}
    \captionsetup{name=Routine}
    \caption{\textbf{\texttt{sendArkTx($Q$,$\tx{}{}$,$\vtxos'$,paths)}:} Send an Ark transaction $\tx{}{}$ to $Q$ by $P$.}
    \label{alg:send-ark-tx}
    \begin{algorithmic}[1]
    \Require We assume $\tx{}{}$ is a fully signed Ark transaction, and that $P$ has knowledge of the set $\vtxos'$ of actual locking scripts of $\tx{}{}$ that $Q$ should verify, and $\patht(\vtxo{})$ for each $\vtxo{}$ spent by $\tx{}{}$ (with the reset transaction $\tx{\vtxo{}}{re}$).
    \State $P$ sends to $Q$ the fully signed Ark transaction $\tx{}{}$, as well as the sets $\vtxos'$ of locking scripts of outputs of $\tx{}{}$ that $Q$ should verify and $\texttt{paths}:=\qty{\patht(\vtxo{}):\vtxo{}\in\tx{}{}.\inputs}$.
    \end{algorithmic}
\end{algorithm*}

\begin{algorithm*}
    \captionsetup{name=Routine}
    \caption{\textbf{\texttt{verifyArkTx($\tx{}{}$,$\vtxos'$,paths)}:} Verify a received Ark transaction $\tx{}{}$ with subset of output locking scripts $\vtxos'$ and VTXT paths $\texttt{paths}$ up to the transaction inputs by $Q$.}
    \label{alg:verify-ark-tx}
    \begin{algorithmic}[1]
    \State $Q$ checks that the locking scripts in $\vtxos'$ are actually committed to by locking scripts of outputs of $\tx{}{}$.
    \State $Q$ checks that if it were to submit all transactions in $\texttt{paths}$ to $\btc$, $\tx{}{}$ would be a valid Bitcoin transaction.
    \end{algorithmic}
\end{algorithm*}


\begin{algorithm*}
    \captionsetup{name=Routine}
    \caption{\textbf{\texttt{\commitmentTx(\boardings,
    \batchSwaps,\exits,\batchTemplate,\\ \connectorTemplate,$\gamma$)}:} Construct commitment transaction with batch outputs according to $\batchTemplate$, connectors according to $\connectorTemplate$ and $\gamma$, processing the boarding requests in \boardings, batch swap requests in \batchSwaps, and exit requests in \exits.}
    \label{alg:commitment-tx}
    \begin{algorithmic}[1]
    \Require A list of verified boarding requests $\boardings\subseteq\toBoard$, a list of verified batch swap requests $\batchSwaps\subseteq\toBatchSwap$, a list of verified exit requests $\exits\subseteq\toExit$, a list of UTXOs $[\out{j}{}]_{j\in J}\subseteq\C{-k}{\operator}$ for which
    \begin{align}
        \quad\quad\quad\sum_{j\in J}&\out{j}{}.\val\geq\sum_{(\cdot,\cdot,\cdot,\vtxos')\in\boardings}\sum_{\vtxo{}'\in\vtxos'}\vtxo{}.\val-\sum_{(\cdot,\cdot,\out{}{},\cdot)\in\boardings}\out{}{}.\val\nonumber\\&+\sum_{(\cdot,\cdot,\cdot,\vtxos')\in\batchSwaps}\sum_{\vtxo{}'\in\vtxos'}\vtxo{}'.\val+\sum_{(\cdot,\cdot,\vtxos,\cdot)\in\batchSwaps}|\vtxos|\cdot\varepsilon\nonumber\\&+\sum_{(\cdot,\cdot,\cdot,\utxos)\in\exits}\sum_{\out{}{}\in\utxos}\out{}{}.\val+\sum_{(\cdot,\cdot,\vtxos,\cdot)\in\exits}|\vtxos|\cdot\varepsilon,
    \end{align}
    and such that for each $j\in J$, $\out{j}{}$ can be spent by $\operator$ with a witness $w_j$ for the commitment transaction. Also, the collection of outputs of the leaves of all VTXTs in $\batchTemplate$ should exactly coincide with $\bigcup_{(\cdot,\cdot,\cdot,\vtxos')\in\boardings}\vtxos'\cup\bigcup_{(\cdot,\cdot,\cdot,\vtxos')\in\batchSwaps}\vtxos'$, and $\connectorTemplate$ should contain an anchor output for each VTXO in $\bigcup_{(\cdot,\cdot,\vtxos,\cdot)\in\batchSwaps}\vtxos\cup\bigcup_{(\cdot,\cdot,\vtxos,\cdot)\in\exits}\vtxos$. For ease of notation, given a VTXT $\vtxt$ in $\batchTemplate$, we define a function $\nu_\vtxt$ which makes explicit for every $\tx{}{}\in\vtxt$ the cosigners that should sign this $\tx{}{}$ (note that this is already present implicitly in $\batchTemplate$).
    \State Define $h_\operator$ as the current block height according to the operator.
    \State $\operator$ constructs the unsigned commitment transaction $\tx{}{commit}$ with $\tx{}{commit}.\inputs=[\out{j}{}]_{j\in J}\cup\bigcup_{(\cdot,\cdot,\out{}{},\cdot)\in\boardings}\out{}{}$, $\tx{}{commit}.\witnesses=\qty[*]_{j\in J}\cup\bigcup_{(\cdot,\cdot,\out{}{},\cdot)\in\boardings}[*]$ and $\tx{}{commit}.\outputs$ given by 
    \begin{equation*}
        \quad\quad\quad\bigcup_{\vtxt\in\batchTemplate}\out{\vtxt}{batch}\cup\bigcup_{\vtxt\in\connectorTemplate}\out{\vtxt}{connector}\cup\bigcup_{(\cdot,\cdot,\cdot,\utxos)\in\exits}\utxos,
    \end{equation*}
    where for each $\vtxt\in\batchTemplate$, $\out{\vtxt}{batch}.\val\geq\sum_{\tx{}{leaf}\in\leaves(\vtxt)}\tx{}{leaf}.\outputs[0].\val$, where 
    \begin{align*}
        \quad\quad\quad\out{\vtxt}{batch}.\lockScript=\texttt{Taproot}\big(&\False;\nonumber\\&\checkSig{\pk_\operator}\wedge\absTimelock{h_\operator+2k+t_e},\nonumber\\&\checkSig{\pk_\operator\oplus\bigoplus_{\pk\in \nu_\vtxt(\roott(\vtxt))}\pk}\big),
    \end{align*}
    and where for each $\vtxt\in\connectorTemplate$, $\out{\vtxt}{connector}.\val\geq|\leaves(\vtxt)|\cdot\varepsilon$ and 
    \begin{equation*}
        \out{\vtxt}{connector}.\lockScript=\texttt{Taproot}(\False;\checkSig{\pk_\operator}).
    \end{equation*}
    \State Let $\signerTemplate:=[~]$ 
    \For{$\vtxt=(V,A)\in\batchTemplate$}
    \State Let $V':=\qty{\nu_\vtxt(\tx{}{}):\tx{}{}\in V}$.
    \State Let $A':=\qty{(\nu_\vtxt(\tx{}{}),\nu_\vtxt(\tx{}{*}))\in V':(\tx{}{},\tx{}{*})\in A}$.
    \State Add $(V',A')$ to $\signerTemplate$. 
    \EndFor
    \State \Return $(\tx{}{commit},\signerTemplate)$. 
    \end{algorithmic}
\end{algorithm*}

\begin{algorithm*}
    \captionsetup{name=Routine}
    \caption{\textbf{\texttt{verifyPath($N$,$t$,$\tx{}{commit}$,$\vtxo{}$,\vtxt,\st)}:} Verify using signer tree $\st=(V',A')$ that VTXT $\vtxt=(V,A)$ correctly includes $\vtxo{}$ and enables unilateral exit controlled by $N$ with expiry time $t$, and is properly included in $\tx{}{commit}$.}
    \label{alg:verify-path}
    \begin{algorithmic}[1]
        \Require We assume there is a unique bijection $\phi_V:V\to V'$ and a unique bijection $\phi_A:A\to A'$ such that $\phi_A((v_1,v_2))=(\phi_V(v_1),\phi_V(v_2))$, otherwise, the routine trivially returns $\False$. The operator public key $\pk_\operator$ should be known.
        \If{$\exists\tx{}{}\in\leaves(\vtxt):\vtxo{}\in\tx{}{}.\outputs$.} \label{rtn:verify-vtx-begin}
        \State Write $(\tx{1}{},\ldots,\tx{\ell}{}):=\patht(\vtxo{})$.
        \Else
        \State \Return \False.
        \EndIf
        \State Write $\tx{0}{}:=\tx{}{commit}$.
        \State Check $\tx{\ell}{}.\outputs=[\vtxo{}]$.
        \For{$i=\ell,\ldots,1$}
        \State Check $\phi_V(\tx{i}{})\supseteq N$.
        \State Check $\exists\out{}{}\in\tx{i-1}{}.\outputs:\tx{i}{}.\inputs=[\out{}{}]$ and call it $\out{}{*}$.
        \State Check $\out{}{*}.\val\geq\sum_{\out{}{i}\in\tx{i}{}.\outputs}\out{}{i}.\val$ and:
        \begin{align*}
            \quad\quad\quad\out{}{*}.\lockScript=\texttt{Taproot}\big(&\False;\nonumber\\&\checkSig{\pk_\operator}\wedge\absTimelock{t},\nonumber\\&\checkSig{\pk_\operator\oplus\bigoplus_{\pk\in\phi_V(\tx{i}{})}\pk}\big).
        \end{align*}
        \EndFor \label{rtn:verify-vtx-end}
        \If{all of the checks in Lines~\ref{rtn:verify-vtx-begin} - \ref{rtn:verify-vtx-end} are successful}
        \State \Return $(\vtxt,\out{}{*})$.
        \EndIf
    \end{algorithmic}
\end{algorithm*}

\begin{algorithm*}
    \captionsetup{name=Routine}
    \caption{\textbf{\texttt{verifyConnector($\vtxo{}$,$\tx{}{commit}$,$\connectorTemplate$,$\gamma$)}:} Verify using signer tree $\st=(V',A')$ that VTXT $\vtxt=(V,A)$ correctly includes $\vtxo{}$ and enables unilateral exit controlled by $N$ with expiry time $t$, and is properly included in $\tx{}{commit}$.}
    \label{alg:verify-connector}
    \begin{algorithmic}[1]
        \If{$\exists\vtxt\in\connectorTemplate:\exists\tx{}{}\in\leaves(\vtxt):\gamma(\vtxo{})=\tx{}{}$} \label{rtn:verify-connector-begin}
        \State Denote this $\vtxt$ by $\vtxt^*$ and $\tx{}{}$ by $\tx{}{*}$.
        \State Write $(\tx{1}{},\ldots,\tx{\ell}{}):=\patht(\tx{}{*})$, i.e. the path of $\tx{}{*}$ in $\vtxt^*$.
        \State Write $\tx{0}{}:=\tx{}{commit}$.
        \For{$i=\ell,\ldots,1$}
        \State Check $\exists\out{}{}\in\tx{i-1}{}.\outputs:\tx{i}{}.\inputs=[\out{}{}]$ and call it $\out{}{*}$.
        \State Check $\out{}{*}.\val\geq\sum_{\out{}{i}\in\tx{i}{}.\outputs}\out{}{i}.\val$.
        \EndFor \label{rtn:verify-connector-end}
        \EndIf
        \If{all of the checks in Lines~\ref{rtn:verify-connector-begin} - \ref{rtn:verify-connector-end} are successful}
        \State \Return \True.
        \EndIf
    \end{algorithmic}
\end{algorithm*}

\begin{algorithm*}
    \captionsetup{name=Routine}
    \caption{\textbf{\texttt{\forfeitTx($\vtxo{}$,$\out{}{}$)}:} Forfeit transaction for $\vtxo{}$ using anchor output $\out{}{}$ by $P$.}
    \label{alg:forfeit}
    \begin{algorithmic}[1]
    \Require It is assumed $P$ is able to cooperate with $\operator$ to produce a witness $w_{\vtxo{}}$ to spend $\vtxo{}$ via a collaborative VTXO script path, and that $P$ knows the operator public key $\pk_\operator$.
    \State $P$ constructs the transaction $\tx{}{forfeit}$, with $\tx{}{forfeit}=[\vtxo{},\out{}{}]$, $\tx{}{forfeit}.\witnesses=[*,*]$, and $\tx{}{forfeit}.\outputs=[(\vtxo{}.\val+\out{}{}.\val,\checkSig{\pk_\operator})]$.  
    \State $P$ sends $\tx{}{forfeit}$ to $\operator$.
    \end{algorithmic} 
\end{algorithm*}

\begin{algorithm*}
    \captionsetup{name=Routine}
    \caption{Tasks a party $P$ performs whenever receiving an Ark transaction.}
    \label{alg:recipient}
    \begin{algorithmic}[1]
    \Require Assume $P$ is the recipient of a $\texttt{sendArkTx}(P,\tx{}{},\vtxos',\texttt{paths})$ operation, and that $P$ is able to provide a witness to a collaborative script path of all VTXOs $\vtxos_P'\subseteq\vtxos'$. 
    \If{$\texttt{verifyArkTx}(\tx{}{},\vtxos',\texttt{paths})$}
    \State Follow Routine~\ref{alg:party} for $\vtxos'_P$, with $P$ batch swapping or exiting collaboratively.
    \EndIf
    \end{algorithmic}
\end{algorithm*}

\begin{algorithm*}
    \captionsetup{name=Routine}
    \caption{\textbf{\texttt{\verifyCommitmentTx($\tx{}{commit}$,\batchTemplate,\signerTemplate,\connectorTemplate,$\gamma$)}:} Verify commitment transaction $\tx{}{commit}$ with batch template $\batchTemplate$, signer template $\signerTemplate$, connector template $\connectorTemplate$ and connector function $\gamma$ by a party $P$.}
    \label{alg:verify-commitment}
    \begin{algorithmic}[1]
        \Require $P$ is assumed to be involved in at least one request that is being processed by $\tx{}{commit}$, either as witness provider or cosigner. We respectively write by $\boardings$, $\batchSwaps$, $\exits$ the sets of boarding, batch swap, and exit requests included in $\tx{}{commit}$, for which $P$ either made the request, or has its public key in the cosigner set. $P$ should know the operator public key $\pk_\operator$ and the cosigner set $N$ of public keys. All checks below are performed by $P$.
        \State Check that $\sum_{\out{}{}\in\tx{}{commit}.\inputs}\out{}{}.\val\geq\sum_{\out{}{}\in\tx{}{commit}.\outputs}\out{}{}.\val$. \label{rtn:verify-commitment-begin}
        \State Let $h_P$ be the current block height according to $P$ and check that $h_\operator+2k+t_e\geq h_P+2k+t_e$.
        \State Let $R_P:=[~]$.
        \For{$(P,N,\out{}{},\vtxos')\in\boardings$} \label{rtn:verify-commitment-boarding1-start}
        \State Check that $\out{}{}\in\tx{}{commit}.\inputs$.
        \For{$\vtxo{}'\in\vtxos'$}
        \State Check $\exists\vtxt\in\batchTemplate:\exists\tx{}{}\in\leaves(\vtxt):\vtxo{}'\in\tx{}{}.\outputs$. 
        \State Denote $\vtxt$ by $\vtxt^*$ and the corresponding tree in $\signerTemplate$ by $\st^*$.
        \State Let $\texttt{res}:=\verifyPath(P,N,h_\operator+2k+t_e,\tx{}{commit},\vtxo{}',\vtxt^*,\st^*)$.
        \If{$\texttt{res}=(\vtxt^*,\out{}{*})$ for some $\out{}{*}\in\tx{}{commit}.\outputs$}
        \State Add $(\vtxt^*,\out{}{*})$ to $R_P$.
        \Else
        \State \Return \False
        \EndIf
        \EndFor
        \EndFor \label{rtn:verify-commitment-boarding1-end}
        \For{$(P,N,\vtxos,\vtxos')\in\batchSwaps$} \label{rtn:verify-commitment-swap1-start}
        \For{$\vtxo{}\in\vtxos$}
        \State $\verifyConnector(\vtxo{},\tx{}{commit},\connectorTemplate,\gamma)$.
        \EndFor
        \For{$\vtxo{}'\in\vtxos'$}
        \State Check $\exists\vtxt\in\batchTemplate:\exists\tx{}{}\in\leaves(\vtxt):\vtxo{}'\in\tx{}{}.\outputs$. 
        \State Let $\vtxt^*:=\vtxt$ and the corresponding tree in $\signerTemplate$ be $\st^*$.
        \State Let $\texttt{res}:=\verifyPath(P,N,h_\operator+2k+t_e,\tx{}{commit},\vtxo{},\vtxt^*,\st^*)$.
        \If{$\texttt{res}=(\vtxt^*,\out{}{*})$ for some $\out{}{*}\in\tx{}{commit}.\outputs$}
        \State Add $(\vtxt^*,\out{}{*})$ to $R_P$.
        \Else
        \State \Return \False
        \EndIf
        \EndFor
        \EndFor \label{rtn:verify-commitment-swap1-end}
        \For{$(P,N,\vtxos,\utxos)\in\exits$} \label{rtn:verify-commitment-exit1-start}
        \For{$\vtxo{}\in\vtxos$}
        \State $\verifyConnector(\vtxo{},\tx{}{commit},\connectorTemplate,\gamma)$.
        \EndFor
        \For{$\out{}{}\in\utxos$}
        \State Check that $\out{}{}\in\tx{}{commit}.\outputs$ and add $(\out{}{},\out{}{})$ to $R_P$.
        \EndFor
        \EndFor \label{rtn:verify-commitment-exit1-end}
        \State Check there are no $(x_1,y_1),(x_2,y_2)\in R_P$ for which $x_1=x_2$ but $y_1\neq y_2$. \label{rtn:verify-commitment-end}
        \State Check any additional requirements announced by $\operator$ (such as maximum number of outputs per virtual transaction, maximum tree depth, or the presence of always spendable anchor outputs $(\varepsilon,\True)$) by inspecting each VTXT in $\batchTemplate$.
        \If{all of the checks in Lines~\ref{rtn:verify-commitment-begin} - \ref{rtn:verify-commitment-end} are successful}
        \State \Return \True.
        \EndIf
    \end{algorithmic}
\end{algorithm*}

\begin{algorithm*}
    \captionsetup{name=Routine}
    \caption{\textbf{\texttt{unilateralExit($\vtxo{}$)}:} Redeem a VTXO onchain via a unilateral exit by a party $P$.}
    \label{alg:unilateral-exit}
    \begin{algorithmic}[1]
        \Require A VTXO $\vtxo{}$ such that $\batch(\vtxo{})$ exists and is in $\C{-k}{P}$. $P$ also needs to possess all, fully signed, virtual transactions in $\patht(\vtxo{})$ as defined in Definition~\ref{def:vtxt}.
        \State $P$ submits all transactions in $\patht(\vtxo{})\setminus\C{-k}{P}$ to $\btc$.
    \end{algorithmic}
\end{algorithm*}

\begin{algorithm*}
    \captionsetup{name=Routine}
    \caption{\textbf{\texttt{sweep($\out{}{}$)}:} Reclaim the remaining funds from an output $\out{}{}$.}
    \label{alg:sweep}
    \begin{algorithmic}[1]
    \Require We assume $\out{}{}$ has a Taproot script path of the form $\checkSig{\pk_\operator}\wedge\absTimelock{t}$ for some block height $t$.
    \If{$\out{}{}$ is spent by some $\tx{}{*}\in\C{}{\operator}$}
    \For{$\out{}{*}\in\tx{}{*}.\outputs$}
    \If{$\out{}{*}$ has a Taproot script path $\checkSig{\pk_\operator}\wedge\absTimelock{t}$ for some $t$}
    \State $\sweep(\out{}{*})$.
    \EndIf 
    \EndFor
    \Else
    \State Construct $\tx{}{}$ spending $\out{}{}$ with output $(v,\checkSig{\pk_\operator})$, where $v\leq \out{}{}.\val$.
    \State Submit $\tx{}{}$ to $\btc$.
    \EndIf
    \end{algorithmic}
\end{algorithm*}

\begin{algorithm*}
    \captionsetup{name=Routine}
    \caption{\textbf{\texttt{\submitCommitTx($\tx{}{commit}$,$\boardings$, $\batchSwaps$, $\exits$, $\old$, $\new$)}:} Update all the required request and VTXO lists after submitting a fully signed commitment transaction $\tx{}{commit}$ constructed from $\boardings$, $\batchSwaps$ and $\exits$, yielding spent and created VTXO sets $\old$ (paired with the corresponding forfeit transaction) and $\new$. All operations are done by $\operator$.}
    \label{alg:update-lists}
    \begin{algorithmic}[1]
    \State Submit $\tx{}{commit}$ to $\btc$, denote by $h$ the block height at which this happens according to $\operator$.
    \State Remove all VTXOs in $\old$ from $\confirmedVTXO$ and $\preConfirmed$ and add $(\tx{}{commit},\old)$ to $\unconfirmedSpent$.
    \State Add $(\tx{}{commit},\new)$ to $\unconfirmed$.
    \State Remove $\boardings$ from $\toBoard$ and add $(\tx{}{commit},\boardings)$ to $\unconfirmedToBoard$.
    \State Remove $\batchSwaps$ from $\toBatchSwap$ and add $(\tx{}{commit},\batchSwaps)$ to $\unconfirmedToBatchSwap$.
    \State Remove $\exits$ from $\toExit$ and add $(\tx{}{commit},\exits)$ to $\unconfirmedToExit$.    
    \If{$\tx{}{commit}\in\C{-k}{\operator}$}
    \State Remove $(\tx{}{commit},\old)$ from $\unconfirmedSpent$.
    \State Add all pairs of VTXOs and forfeit transactions in $\old$ to $\spent$.
    \State Remove $(\tx{}{commit},\new)$ from $\unconfirmed$.
    \State Add all VTXOs in $\new$ to $\confirmedVTXO$.
    \State Add all batch outputs of $\tx{}{commit}$ to $\confirmedBatches$.
    \State Remove $(\tx{}{commit},\boardings)$ from $\unconfirmedToBoard$.
    \State Remove $(\tx{}{commit},\batchSwaps)$ from $\unconfirmedToBatchSwap$.
    \State Remove $(\tx{}{commit},\exits)$ from $\unconfirmedToExit$.
    \State Add $(\tx{}{commit},\boardings)$ to $\confirmedToBoard$.
    \State Add $(\tx{}{commit},\batchSwaps)$ to $\confirmedToBatchSwap$.
    \State Add $(\tx{}{commit},\exits)$ to $\confirmedToExit$.
    \ElsIf{$\tx{}{commit}$ has never been in $\C{-k}{\operator}$ by block height $h+t_r$}
    \State Remove $(\tx{}{commit},\old)$ from $\unconfirmedSpent$.
    \State Remove $(\tx{}{commit},\new)$ from $\unconfirmed$.
    \State Remove $(\tx{}{commit},\boardings)$ from $\unconfirmedToBoard$.
    \State Remove $(\tx{}{commit},\batchSwaps)$ from $\unconfirmedToBatchSwap$.
    \State Remove $(\tx{}{commit},\exits)$ from $\unconfirmedToExit$.
    \State Add $\boardings$ to $\toBoard$.
    \State Add $\batchSwaps$ to $\toBatchSwap$.
    \State Add $\exits$ to $\toExit$.
    \EndIf
    \end{algorithmic}
\end{algorithm*}

\begin{algorithm*}
    \captionsetup{name=Routine}
    \caption{\textbf{\texttt{processRequests(\boardings,\batchSwaps,\exits,\batchTemplate,\\ \connectorTemplate,$\gamma$)}:} Construct the corresponding commitment transaction according to Routine~\ref{alg:commitment-tx} and finalise it by gathering the required signatures.}
    \label{alg:process-requests}
    \begin{algorithmic}[1]
    \Require It is assumed that the $\operator$ and all cosigners represented in the requests $\boardings$, $\batchSwaps$ and $\exits$ can communicate with each other. For ease of notation, we refer to the cosigners by their public keys, and we denote by $\No$ the set of all cosigners involved in all the requests processed here, and we denote by $\P$ the set of all requesters who made requests processed here.
    \State $\operator$ constructs $(\tx{}{commit},\signerTemplate):=\commitmentTx(\boardings,\batchSwaps,\exits,\batchTemplate,\allowbreak\connectorTemplate,\gamma)$.
    \State Let $\new:=\bigcup_{(\cdot,\cdot,\cdot,\vtxos')\in\boardings\cup\batchSwaps}\vtxos'$ be all to be created VTXOs.
    \State Let $\old:=[~]$.
    \For{$S\in\P\cup\No$}
    \State $\operator$ sends to $S$ the tuple $(\tx{}{commit}, \batchTemplate, \signerTemplate, \connectorTemplate, \gamma)$.
    \If{\nott $\verifyCommitmentTx(\tx{}{commit},\batchTemplate,\signerTemplate,\connectorTemplate,\gamma)$}
    \State $S$ aborts.
    \EndIf
    \EndFor
    \State Since $\verifyCommitmentTx(\ldots)$ succeeded for everyone, there is for each $\vtxt=(V,A)\in\batchTemplate$ a corresponding tree $\st\in\signerTemplate$ with bijections $\phi_V^\vtxt:V\to V'$ and $\phi_A^\vtxt:A\to A'$ such that $\phi_A^\vtxt((v_1,v_2))=(\phi_V^\vtxt(v_1),\phi_V^\vtxt(v_2))$.
    \For{$\vtxt:=(V,A)\in\batchTemplate$}
    \For{$\tx{}{}\in V$}
    \State $\musig(\tx{}{},\phi_V^\vtxt(\tx{}{})\cup\operator)$. \label{rtn:process-requests-board-musig}
    \EndFor
    \EndFor
    \For{$(P,N,\out{}{},\vtxos')\in\boardings$}
    \State $\musig(\tx{}{commit},P\cup\qty{\operator})$.
    \EndFor
    \For{$(P,N,\vtxos,\vtxos')\in\batchSwaps$}
    \For{$\vtxo{}\in\vtxos$}
    \State $\forfeitTx(\vtxo{},\gamma(\vtxo{}))$ yielding $\tx{\vtxo{}}{forfeit}$.
    \State $\operator$ checks that $\tx{\vtxo{}}{forfeit}$ has:
    \begin{itemize}[leftmargin=45pt]
        \item inputs spending $\vtxo{}$ and $\gamma(\vtxo{})$.
        \item one output sending all the funds to $\operator$.
    \end{itemize} 
    \State $P$ and $\operator$ cooperate to provide a collaborative witness to spend $\vtxo{}$. \label{rtn:process-requests-swap-musig}
    \State $\operator$ adds $(\vtxo,\tx{\vtxo{}}{forfeit})$ to $\old$, with $\tx{\vtxo{}}{forfeit}$ now fully signed.
    \EndFor
    \EndFor
    \For{$(P,\varnothing,\vtxos,\utxos)\in\exits$}
    \For{$\vtxo{}\in\vtxos$}
    \State $\forfeitTx(\vtxo{},\gamma(\vtxo{}))$ yielding $\tx{\vtxo{}}{forfeit}$.
    \State $\operator$ checks that $\tx{\vtxo{}}{forfeit}$ has:
    \begin{itemize}[leftmargin=45pt]
        \item inputs spending $\vtxo{}$ and $\gamma(\vtxo{})$.
        \item one output sending all the funds to $\operator$.
    \end{itemize}
    \State $P$ and $\operator$ cooperate to provide a collaborative witness to spend $\vtxo{}$. \label{rtn:process-requests-exit-musig}
    \State $\operator$ adds $(\vtxo,\tx{\vtxo{}}{forfeit})$ to $\old$, with $\tx{\vtxo{}}{forfeit}$ now fully signed.
    \State $\operator$ adds $(\vtxo,\tx{\vtxo{}}{forfeit})$ to $\old$.
    \EndFor
    \EndFor
    \State Once $\operator$ received a forfeit transaction for every $\vtxo{}$ spent in a batch swap or exit request, $\operator$ provides the witnesses $[w_j]_{j\in J}$ for the outputs $[\out{j}{}]_{j\in J}$.
    \State $\submitCommitTx(\tx{}{commit},\boardings,\batchSwaps,\exits,\old,\new)$.
    \end{algorithmic}
\end{algorithm*}


\begin{algorithm*}
    \captionsetup{name=Routine}
    \caption{Tasks an operator performs continuously.}
    \label{alg:operator}
    \begin{algorithmic}[1]
    \While{\True}
    \If{$\operator$ receives a request $r=(\text{``boarding: ''},P,N,\out{}{},\pk_I,S_C,S_U,\vtxos')$}
    \State $\texttt{verifyBoardingRequest}(r)$.
    \EndIf
    \If{$\operator$ receives a request $r=(\text{``Ark: ''},P,\tx{}{},\vtxos,\vtxos')$}
    \State $\texttt{verifyArkTxRequest}(r)$.
    \EndIf
    \If{$\operator$ receives a request $r=(\text{``batch swap: ''},P,N,\vtxos,\vtxos')$}
    \State $\texttt{verifyBatchSwapRequest}(r)$.
    \EndIf
    \If{$\operator$ receives a request $r=(\text{``exit: ''},P,N,\vtxos,\utxos)$}
    \State $\texttt{verifyExitRequest}(r)$.
    \EndIf
    \If{$\beta$ outputs $(\boardings,\batchSwaps,\exits,\batchTemplate,\connectorTemplate,\gamma)$}
    \State $\texttt{processRequests}(\boardings,\batchSwaps,\exits,\batchTemplate,\connectorTemplate,\gamma)$.
    \EndIf
    \For{$out{}{}\in\confirmedBatches$}
    \If{Current block height is past $\out{}{}$'s batch expiry height} \label{rtn:operator-expired}
    \State $\sweep(\out{}{})$.
    \EndIf
    \EndFor
    \For{$(\vtxo{},\tx{}{})\in\spent$} \label{rtn:operator-vtxo-spend-start}
    \If{$\vtxo{}\in\C{}{\operator}$}
    \State Submit $\tx{}{}$ to $\btc$. \label{rtn:operator-spent}
    \EndIf
    \EndFor \label{rtn:operator-vtxo-spend-end}
    \EndWhile
    \end{algorithmic}
\end{algorithm*}

\begin{algorithm*}
    \captionsetup{name=Routine}
    \caption{Tasks a party $P$ performs in order to spend a set of VTXO $\vtxos$ $P$ would want to spend collaboratively/unilaterally.}
    \label{alg:party}
    \begin{algorithmic}[1]
    \Require Assume $P$ is able to provide a witness to both a unilateral and collaborative script path of each VTXO in $\vtxos$, and that $P$ has knowledge of the fully signed virtual transactions on $\patht(\vtxo{})$ for all $\vtxo{}\in\vtxos$. Let $t$ be the batch expiry height. In the case of an Ark transaction, we assume that the transaction inputs can be accessed via the set $\texttt{paths}$ of virtual transaction paths, that the transaction outputs are in the set $\vtxos'$, and that $P$ has defined a set of cosigner $N$ and a function $\eta:N\to\P(\vtxos')$ which maps each cosigner to a set of VTXOs that that cosigner should be able to batch swap or exit with.
    \If{$P$'s current block height is smaller than $t-2k-1$}
    \If{$P$ wants to perform a batch swap or exit}
    \State Initiate the desired batch swap or exit request.
    \State Collaborate with $\operator$ in the corresponding $\texttt{processRequests}(\ldots)$.
    \ElsIf{$P$ wants to perform an Ark transaction $\tx{}{}$}
    \State Initiate the desired Ark transaction request $r$ with outputs $\vtxos'$ according to Routine~\ref{alg:ark-tx}.
    \If{$P$ obtains the fully signed Ark and reset transactions after a successful \texttt{verifyArkTxRequest(}$r$\texttt{)}}
    \For{$Q\in N$}
    \State $P$ executes $\texttt{sendArkTx}(Q,\tx{}{},\eta(Q),\texttt{paths})$.
    \EndFor
    \EndIf
    \EndIf
    \State Perform the desired batch swap or exit request, collaborating with $\operator$.
    \ElsIf{$P$'s current block height is $t-2k-1$ and spending $\vtxos$ failed}
    \For{$\vtxo{}\in\vtxos$}
    \State $P$ executes $\texttt{unilateralExit}(\vtxo{})$ and spends the newly created UTXO.
    \EndFor
    \EndIf
    \end{algorithmic}
\end{algorithm*}

\clearpage

\section{Formal Analysis}
\label{app:proofs}
In this section, we formally define and prove the security of the Ark protocol. As mentioned in §\ref{subsec:system-model}, we work with the Bitcoin backbone static model (fixed set of miners, static difficulty) presented in \cite{garay2024bitcoin} to model the Bitcoin blockchain. In this model, we proceed in discrete rounds. By the synchrony assumption, every message sent at round $r$ by an honest party will be received by all other honest party at round $r+1$. At round $r$, each Ark party $Q$ maintains its own view of the ledger $\C{}{Q}(r)$. We write $\C{-k}{Q}(r)$ for this local view with the last $k$ blocks removed. We write $\C{}{1}\preceq\C{}{2}$ if $\C{}{1}$ is a prefix of $\C{}{2}$. To reason about the confirmation of transactions and the possibility of reorgs, we treat Bitcoin as a protocol $\btc$ with depth parameter $k$ and wait parameter $u$, which ensures, with overwhelming probability (w.o.p.), 
\begin{itemize}[leftmargin=*]
    \item Persistence: For any two honest parties $Q$ and $Q'$, and any round $r$, we either have $\C{-k}{Q}(r)\preceq\C{-k}{Q'}(r)$ or $\C{-k}{Q}(r)\succeq\C{-k}{Q'}(r)$. We refer to $\C{-k}{Q}(r)$ as the \emph{stable} part of $Q$'s ledger, as this will no longer change from round $r$ onward.
    \item Liveness: For any transaction $\tx{}{}$ submitted to $\btc$ by an honest party at round $r$, we have $\tx{}{}\in\C{-k}{Q}(s)$ for every honest $Q$ by round $s=r+u$.
\end{itemize}
If we treat Bitcoin as above, one can show the following.
\begin{corollary}[Corollary 4.12 in \cite{garay2024bitcoin}]
\label{cor:k}
    W.o.p., $\btc$ ensures any $k$ consecutive blocks in the chain of an honest party contain at least one block produced by an honest miner, who includes \emph{all} valid transactions it has heard of by the round it produces the block.
\end{corollary}
We translate the liveness property defined above to a statement in terms of rounds. To this end, we state and prove the following Lemma (adapted from Lemma 16 in \cite{aumayr2024blink}).
\begin{lemma}
\label{lem:live}
    If an honest party $Q$ submits a transaction $\tx{}{}$ to $\btc$ at round $r$, at which time the latest stable block in $Q$'s ledger is $B^Q$, a block $B^*$ including $\tx{}{}$ it will be in the stable part of the ledger of all honest parties at most $3k$ blocks away from $B^Q$, w.o.p.
\end{lemma}
\begin{proof}
    We know by Persistence that at round $r$, there are $k$ consecutive blocks on top of $B^Q$ in $\C{}{Q}$. By Liveness, any $\tx{}{}$ submitted by an honest party $Q$ to $\btc$ at round $r$ will be in the stable part of the ledger of any honest party by round $s=r+u$. We now know two things: (i) every honest miner will have heard of $\tx{}{}$ from round $r+1$ onward. Therefore, $\tx{}{}$ will be included in the first honest block mined from round $r+1$ onward. However, by Corollary~\ref{cor:k}, we know that w.o.p., there cannot be more than $k-1$ blocks between round $r$ and round $s$. Indeed, if there would, this would mean there would be a sequence of $k$ consecutive blocks that do not contain an honest block, as an honest block would have included $\tx{}{}$. Hence, there are at most $2k-1$ blocks between $B^Q$ and $B^*$. (ii) by round $s$, $B^*$ is in the stable part of every honest party's ledger, i.e., $B^*$ is at least $k$ blocks deep. Thus, combining (i) and (ii), $B^*$ is in the stable part of every honest party's ledger after at most $3k$ consecutive blocks from $B^Q$.
\end{proof}
This implies, if an honest party $Q$'s local view is currently at block height $h$, and $Q$ submits $\tx{}{}$ to $\btc$ at that height, then w.o.p., $\tx{}{}$ will be included in a block that is a most $3k$ blocks away from $Q$'s latest stable block (at block height $h-k$). That is, $\tx{}{}$ will be in the stable part of any honest party's ledger at latest at block height $h-k+3k=h+2k$. We use this realisation throughout the rest of the proofs and state it as a Corollary.
\begin{corollary}
\label{cor:2k}
    If an honest party submits $\tx{}{}$ to $\btc$ at block height $h$, $\tx{}{}$ will be in the stable part of every honest party's ledger at latest at block height $h+2k$, w.o.p.
\end{corollary}

We can now introduce the notion of an \emph{Ark state}, which represents all VTXOs in an Ark. Next, we define \emph{Ark state transitions} as sets of \emph{Ark actions}. These notions will help us to formally define the security properties stated in §\ref{subsec:properties}.
\begin{definition}[Ark State]
\label{def:ark-state}
    The \emph{Ark state} of an Ark run by $\operator$ is defined as a tuple $\arkstate:=(C,F,S)$, where $C$ is the set containing each VTXO, that is part of a batch confirmed onchain, for which there is no valid transaction spending that VTXO collaboratively in case it would be present onchain, $S$ is the set of VTXOs for which $\operator$ holds a fully signed reset transaction, or a fully signed forfeit transaction spending an anchor output confirmed onchain, and $F$ is the set of VTXOs that are outputs of Ark transactions, but not in $S$ (viz., unspent outputs). 
\end{definition}
\begin{definition}[Ark Action]
    An \emph{Ark action} is a concise way to denote an Ark protocol action. It is a tuple $\alpha=(P,N,\operator,v_{in},u_{in},v_{out},u_{out},a)$, where $P$ is the party providing the required witnesses to spend the VTXOs in $v_{in}$ and UTXOs in $u_{in}$, in order to create the VTXOs in $v_{out}$ and UTXOs in $u_{out}$. $N$ is the cosigner set involved in committing any VTXO to a batch. Finally, $a$ is a binary variable equal to 1 if and only if $\alpha$ is an Ark transaction. 
\end{definition}
\begin{definition}[Ark State Transition]
    An \emph{Ark state transition} $\tau$ is a set of Ark actions that changes the Ark state from $(C,F,S)$ to $(C',F',S')$. We may also write $(C,F,S)\stackrel{\tau}{\to}(C',F',S')$. 
\end{definition}

Ark state transitions include commitment transactions, unilateral exits, Ark transactions, and batch sweeps. For all but Ark transactions, these transitions only occur if the relevant transactions are confirmed onchain. While the latter three involve a single action each, a commitment transaction may group multiple actions. Thus, we may represent a commitment transaction $\tx{}{}$ as a set $\qty{\alpha_i}_{i\in I}$ for some index set $I$. Conversely, a set of actions $\qty{\alpha_i}_{i\in I}$ with $\alpha_i=(P_i,N_i,v_{in,i},u_{in,i},v_{out,i},u_{out,i},0)$ defines a commitment transaction $\tx{}{}$ with inputs $\bigcup_{i\in I}u_{in,i}$, outputs $\bigcup_{i\in I}u_{out,i}$, batch outputs holding $\bigcup_{i\in I}v_{out,i}$, requiring signatures from $\bigcup_{i\in I}N_i$, and connector outputs for $\bigcup_{i\in I}v_{in,i}$. Witnesses come from the corresponding parties in $\bigcup_{i\in I}P_i$. 

As claimed in §\ref{subsec:properties}, the Ark protocol with operator $\operator$, $\ark{\operator}$, gives a number of security guarantees. Informally, VTXO Security gives \emph{static} guarantees that VTXOs cannot be spent by unauthorised parties, and that authorised parties can always spend VTXOs up to a certain time. This is fundamental in showing User Balance Security. Ark Atomicity ensures that for any honest intent to change the state, the state either changes exactly according to the intent, or not at all. In both cases, honest users can rely on the guarantees of VTXO Security. Intuitively, Ark Atomicity ensures that transacting VTXOs works as expected. Of most importance to an honest operator is that it does not lose funds. We phrase this as Operator Balance Security. Finally, under the extra assumptions \ref{asm:1}-\ref{asm:4}, we show that users who opted into the fast finality mechanism enjoy Fast Finality Balance Security; extending balance security guarantees to VTXOs that have not been committed to by a confirmed onchain batch. Let us now precisely state and prove each security property.

\begin{theorem}[VTXO Security]
\label{thm:vtxo}
    The protocol $\ark{\operator}$ is
    \begin{enumerate}[leftmargin=16pt]
        \item [(i)] \emph{VTXO-safe}, i.e., for each Ark state $\arkstate=(C,F,S)$, for each $\vtxo{}\in C$ confirmed in a batch with expiry $T$, and for each Ark state transition $\tau$ before $T$ such that $(C,F,S)\stackrel{\tau}{\to}(C',F',S')$, if the cosigner set $N$ from the action creating $\vtxo{}$ is honest, we have that if $\vtxo{}\in C$ and $\vtxo{}\notin C'$, there must be an action $\alpha=(P,N',\operator,v_{in},u_{in},v_{out},u_{out},a)\in\tau$ such that $\alpha$ is a unilateral exit of $\vtxo{}$, or $P$ can construct with $\operator$ a valid witness for $\vtxo{}$. 
        \item [(ii)] \emph{VTXO-live}, i.e., for each Ark state $\arkstate=(C,F,S)$, for each $\vtxo{}\in C$ in a confirmed batch with expiry $T$, if the cosigner set $N$ from the action creating $\vtxo{}$ is honest, $\vtxo{}$ can, at latest at block height $T-2k-1$, be turned w.o.p. into a UTXO with the same spending conditions by anyone who can submit the fully signed transactions in $\patht(\vtxo{})$.
    \end{enumerate}
\end{theorem}
\begin{proof}
    Consider an arbitrary Ark state $\arkstate=(C,F,S)$, and let $\vtxo{}\in C$ be a VTXO in a confirmed batch with expiry $T$ (we assume that $C\neq\varnothing$). We assume that $\vtxo{}$ has been created by an Ark action with honest cosigner set $N$. Routine~\ref{alg:verify-path} guarantees that $N$ is among the signers required to spend the batch output before expiry, and every output on the path of virtual transactions to $\vtxo{}$. Since $N$ is honest, it will not agree to sign off on another transaction that would render $\vtxo{}$ invalid. This realisation implies the following:
    \begin{enumerate}[leftmargin=*]
        \item [(i)] For each state transition $\tau$ such that $(C,F,S)\stackrel{\tau}{\to}(C',F',S')$ before $T$, assume that $\vtxo{}\in C$ and $\vtxo{}\notin C'$. Since the batch containing $\vtxo{}$ did not yet expire, the only way $\vtxo{}$ would not be a part of $C'$ any more, is if $\vtxo{}$ would be turned into a UTXO via a unilateral exit, or if there exists a valid (virtual) transaction spending $\vtxo{}$ collaboratively. This transaction can only exist if it includes a valid witness for $\vtxo{}$'s collaborative locking scripts. $P$ should thus be able to provide such a witness, cooperating with $\operator$.
        \item [(ii)] Before batch expiry, the only way to spend the batch is with presigned transactions from the VTXT specifying the batch. Anyone in possession of the virtual transactions on $\patht(\vtxo{})$ can thus submit these transactions, and be sure that there can be no double-spend. This as long as all transactions in the path get confirmed before block height $T$. By Corollary~\ref{cor:2k}, this happens w.o.p. within at most $2k$ blocks. In other words, as long as all transactions in $\patht(\vtxo{})$ get submitted to $\btc$ by a party at latest at block height $T-2k-1$ in that party's local view, $\vtxo{}$ will be available onchain as a confirmed UTXO. 
    \end{enumerate}
\end{proof}
\begin{remark}
    Note how our formal model deviates from reality in the previous proof. For one, we assume that all transactions in $\patht(\vtxo{}{})$ are submitted at once, and can all be included at once in a block. While this is true if the party submitting the transactions would talk directly to a miner, it is not true in practice with Bitcoin's current P2P relay policy. Bitcoin Core at the time of writing only supports packages of one parent and one child transaction being relayed through the network. Hence, a chain of transactions like $\patht(\vtxo{})$, which relies on CPFP (child-pays-for-parent) in order to bump the fees of the parent transactions to guarantee inclusion onchain, may not get relayed reliably through the network. One would have to first make sure that parent transaction are successfully broadcast and included in the mempools of miners, and then broadcast the children (one-by-one). At the time of writing, work is ongoing \cite{clustermempool} to introduce \emph{full package relay} into Bitcoin, which would remove this inconvenience as it would allow larger packages of transactions to be relayed together.

    Second, the Bitcoin backbone model effectively assumes arbitrarily large blocks, such that all transactions submitted (and available in an honest miner's mempool) will be included by that honest miner in one single block. This assumption may no longer be realistic in case of a bank run scenario (see §\ref{sec:discussion}), where congestion may occur by the sheer number of transactions submitted.
\end{remark}
\begin{definition}[Ark Balance]
    For a given party $P$, for a given Ark state $\Sigma=(C,F,S)$ and a given chain view $\C{-k}{P}$, we define the \emph{Ark balance $b_P^{\Sigma,\C{}{}}$} of $P$ at $\Sigma$ as the sum $b_P^{\Sigma,\C{}{}}:=\sum_{\vtxo{}\in V_P^\Sigma}\vtxo{}.\val+\sum_{\out{}{}\in R_P^{\C{}{}}}\out{}{}.\val$, where $V_P$ is the set of all $\vtxo{}\in C$ which are more than $2k$ blocks from expiry, for which \emph{only} $P$ can provide a witness for a unilateral script path of $\vtxo{}$ and possesses all fully signed transactions in $\patht(\vtxo{})$, and where $R_P^{\C{}{}}$ is the set of all unspent boarding outputs, i.e., all $\out{}{}\in\C{-k}{P}$ with $\out{}{}.\lockScript=\texttt{Taproot}(\False,S_C,S_U)$ where each script path in $S_C$ requires $\operator$'s signature and where \emph{only} $P$ can provide a witness for the script paths in $S_U$.
\end{definition}
\begin{theorem}[User Balance Security]
\label{thm:user}
    The Ark protocol $\ark{\operator}$ with batch expiry time $t_e$ guarantees that for each honest party $P$, each Ark state $\Sigma=(C,F,S)$, and each chain view $\C{-k}{P}$, $P$ can claim its Ark balance $b_P^{\Sigma,\C{}{}}$ onchain, subtracting mining fees, w.o.p.
\end{theorem}
\begin{proof}
    By part (ii) of Theorem~\ref{thm:vtxo}, $P$ can turn each $\vtxo{}\in V_P$ into a UTXO confirmed onchain w.o.p., as $P$ is honest and knows $\patht(\vtxo{})$. Of course, $P$ has to pay the corresponding mining fees to have all these virtual transactions confirmed onchain. Since $P$ also knows the witness of a unilateral script path of each of these newly created UTXOs, as well as the witness of a unilateral script path of each boarding output, $P$, and only $P$ by part (i) of Theorem~\ref{thm:vtxo}, can spend all these UTXOs (paying mining fees doing so) and thus claim the funds locked by them. Moreover, $P$ can retrieve the funds locked in boarding outputs $\out{}{}\in R_P^{\C{}{}}$ by spending these outputs via one their unilateral script paths (paying mining fees), due to validity of Bitcoin as only $P$ knows corresponding valid witnesses. In conclusion, $P$ can claim its Ark balance $b_P^{\Sigma,\C{}{}}$ onchain, up to mining fees, w.o.p.
\end{proof}
\begin{remark}
    Note that for the Ark balance to be non-trivial, and for an honest party to actually have retrievable funds, the Ark's batch expiry time should be larger than $2k$.
\end{remark}
\begin{theorem}[Ark Atomicity]
\label{thm:atomicity}
For any Ark state $\Sigma=(C,F,S)$, and for any Ark action $\alpha=(P,N,\operator,v_{in},u_{in},v_{out},u_{out},a)$ with $P,\operator\neq\varnothing$, as long as $P$, $N$, or $\operator$ is honest, $\ark{\operator}$ ensures that exactly one of the following will happen:
    \begin{enumerate}[leftmargin=16pt]
        \item [(i)] $\alpha$ does not get included in a state transition, or
        \item [(ii)] $\alpha$ is included in a state transition $\tau$ resulting in a state $(C',F',S')$, where:
        \begin{itemize}[leftmargin=*]
            \item $\forall\vtxo{}\in v_{in}:C'\cup F'\not\ni\vtxo{}\wedge S'\ni\vtxo{}$. 
            \item Each UTXO in $u_{in}$ is spent onchain.
            \item $\forall\vtxo{}\in v_{out}:(a=0\implies\vtxo{}\in C')\wedge (a=1\implies\vtxo{}\in F')$. 
            \item Each UTXO in $u_{out}$ is confirmed onchain.
        \end{itemize}
    \end{enumerate} 
\end{theorem}
\begin{proof}
    For an arbitrary Ark state $\Sigma=(C,F,S)$, let us consider an arbitrary Ark action $\alpha=(P,N,\operator,v_{in},u_{in},v_{out},u_{out},a)$ with $P,\operator\neq\varnothing$. Assume that $P$ or $\operator$ is honest. Then if $\alpha$ is not an Ark action described by the protocol with $P,\operator\neq\varnothing$, i.e., $(a)$ a boarding, $(b)$ a batch swap, $(c)$ a collaborative exit, or $(d)$ an Ark transaction, the honest participant will not participate in this request, meaning that it will not be included in a state transition. This is because either the honest party is $P$, and will not provide witnesses approving the action for any of the required inputs from $v_{in}$ and $u_{in}$, 
    or the honest party is $\operator$, and will ignore any action that is not in the form of an Ark transaction or a request to be part of a commitment transaction. Thus, it remains to check that in cases $(a)$-$(d)$ we end up either with (i) or with (ii):
    \begin{enumerate}[leftmargin=*]
        \item [(a)] A boarding action is of the form $\alpha=(P,N,\operator,\varnothing,\out{}{board},v_{out},\varnothing,0)$. First, assume that $P$ is honest. In that case, $\alpha$ gives rise to an honestly created boarding request (Routine~\ref{alg:boarding}), which does not spend VTXOs that are already spent. This request either gets ignored by $\operator$, or it gets (possibly incorrectly) included in a commitment transaction. $P$ will verify this commitment transaction according to Routine~\ref{alg:verify-commitment} (in particular Lines~\ref{rtn:verify-commitment-boarding1-start}-\ref{rtn:verify-commitment-boarding1-end} and \ref{rtn:verify-commitment-end}), and will only cooperate to create the witness for $\out{}{board}$ if the verification is successful, according to Routine~\ref{alg:process-requests} Line~\ref{rtn:process-requests-board-musig}. In particular, an honest $P$ will only validate a commitment transaction which spends $\out{}{board}$ and which has batch outputs that contain all the VTXOs in $v_{out}$, with the corresponding virtual transaction paths requiring signatures from $N$. Hence, if this commitment transaction is confirmed onchain, $\out{}{board}$ will be spent, and each $\vtxo{}\in v_{out}$ will be included in $C'$ (with $(C',F',S')$ being the state after this commitment transaction). If it is not confirmed, $\alpha$ will not get included there, but may be included in a different commitment transaction, for which the same reasoning holds. In conclusion, we either have (i) or (ii).
    
            
        Now, assume that $\operator$ is honest. $\alpha$ may only be included in a commitment transaction if $\operator$ receives the corresponding boarding request and it passes verification according to Routine~\ref{alg:verify-boarding}. It will thus only include $\alpha$ in an honestly created commitment transaction, which, if confirmed onchain, guarantees to spend $\out{}{board}$ and adds all VTXOs in $v_{out}$ to $C'$ of the new state $(C',F',S')$. Otherwise, if one of the cosigners or requesters is not responding, $\alpha$ will not get included. Note that Line~\ref{rtn:verify-3} makes sure that the boarding output cannot be double-spent, but even if it could be, this will not lead to atomicity being broken.
            
        \item [(b)] A batch swap action is of the form $\alpha=(P,N,\operator,v_{in},\varnothing,v_{out},\varnothing,0)$. Assume that $P$ is honest. $\alpha$ then gives rise to an honestly created batch swap request (Routine~\ref{alg:batch-swap}), which does not spend VTXOs that have already been spent, and which either gets ignored by $\operator$, or incorporated (possibly incorrectly) in a commitment transaction. $P$ will verify this transaction according to Routine~\ref{alg:verify-commitment} Lines~\ref{rtn:verify-commitment-swap1-start}-\ref{rtn:verify-commitment-swap1-end} and Line~\ref{rtn:verify-commitment-end}. Only if verification succeeds, will $P$ cooperate with $\operator$ to create the necessary witnesses for the forfeit transactions spending each $\vtxo{}\in v_{in}$, by Routine~\ref{alg:process-requests} Line~\ref{rtn:process-requests-swap-musig}. In particular, an honest $P$ will only cooperate for a commitment transaction which has connector outputs containing anchors for each $\vtxo{}\in v_{in}$, and batch outputs containing each $\vtxo{}\in v_{out}$, with the corresponding virtual transaction paths requiring signatures from $N$. If this commitment transaction is confirmed onchain, each $\vtxo{}\in v_{in}$ will not be present in $C'$ or $F'$ depending on whether it was the output of an Ark transaction, and will be in $S'$, as the corresponding forfeit transaction can now spend $\vtxo{}$ if ever present onchain, and each $\vtxo{}\in v_{out}$ will be included in the new $C'$. If not confirmed, $\alpha$ does not change the state. We indeed have either (i) or (ii).
    
            
        To conclude, assume that $\operator$ is honest. $\alpha$ is only included in a commitment transaction if $\operator$ receives the corresponding batch swap request. It will then verify the request according to Routine~\ref{alg:verify-swap}. Consequently, a verified $\alpha$ will only be included in an honestly created commitment transaction. This transaction will only be submitted onchain once $\operator$ has the fully signed forfeit transactions corresponding to each VTXO in $v_{in}$. If confirmed onchain, this commitment transaction guarantees that all VTXOs in $v_{in}$ will not be in $C'\cup F'$ and will be in $S'$, and that all VTXOs in $v_{out}$ are in $C'$ of the new state $(C',F',S')$. Otherwise, if one of the cosigners or requesters is not responding, $\alpha$ will not be included. 
            
        \item [(c)] We can write a collaborative exit action as $\alpha=(P,\varnothing,\operator,v_{in},\varnothing,\varnothing,u_{out},0)$. If $P$ is honest, $\alpha$ gives rise to an honestly created exit request (Routine~\ref{alg:exit}), in particular not spending VTXOs that have already been spent. This either gets ignored by $\operator$, or added (possibly incorrectly) to a commitment transaction. $P$ will verify this commitment transaction according to Routine~\ref{alg:verify-commitment} Lines~\ref{rtn:verify-commitment-exit1-start}-\ref{rtn:verify-commitment-exit1-end} and Line~\ref{rtn:verify-commitment-end}. $P$ will only cooperate with $\operator$ to create the necessary witnesses for the forfeit transactions spending each $\vtxo{}\in v_{in}$ if verification succeeds, by Routine~\ref{alg:process-requests} Line~\ref{rtn:process-requests-exit-musig}. An honest $P$ will only cooperate for a commitment transaction which has connector outputs containing anchors for each $\vtxo{}\in v_{in}$, and separate outputs corresponding to each $\out{}{}\in u_{out}$. If this commitment transaction is confirmed onchain, each $\vtxo{}\in v_{in}$ will not be in $C'$ or $F'$ depending on whether it was the output of an Ark transaction, and will simultaneously be added to $S'$, as the corresponding forfeit transaction can now spend $\vtxo{}$ if ever present onchain, and each $\out{}{}\in u_{out}$ will be available onchain. If not confirmed, $\alpha$ does not change the state. Once again, we have either (i) or (ii).
    
        If we now assume that $\operator$ is honest, $\alpha$ will only be included in a commitment transaction if $\operator$ receives the corresponding exit request. It will then verify the request according to Routine~\ref{alg:verify-exit}. Hence, a verified $\alpha$ will only be included in an honestly created commitment transaction. This transaction will be submitted onchain only if $\operator$ has the fully signed forfeit transactions corresponding to each VTXO in $v_{in}$. This commitment transaction guarantees that all VTXOs in $v_{in}$ will not be in $C'\cup F'$ but in $S'$, and that all UTXOs in $u_{out}$ are confirmed onchain, once the commitment transaction is confirmed onchain. Otherwise, if one of the cosigners or requesters is not cooperating, $\alpha$ will not be included. 
    
        \item [(d)] An Ark transaction action can be denoted as $\alpha=(P,\varnothing,\operator,v_{in},\varnothing,v_{out},u_{out},1)$. If $P$ is honest, $\alpha$ gives rise to an honestly created Ark transaction request (Routine~\ref{alg:ark-tx}), not spending VTXOs that have already been spent. $\operator$ either ignores this request, or cooperates with $P$ to provide the required witnesses (possibly after verifying it via Routine~\ref{alg:verify-ark-tx-request}), making the Ark transaction valid. In the former case, $\alpha$ is not included in a state transition. In the latter case, the existence of valid reset transactions spending each VTXO in $v_{in}$, implies that these VTXOs are no longer present in $C'$ but instead in $S'$, and the VTXOs in $v_{out}$ are in $F'$. 
            
        Finally, if we assume that $\operator$ is honest, $\alpha$ will only be included in a state transition if $\operator$ receives the corresponding Ark transaction request and it passes verification (Routine~\ref{alg:verify-ark-tx-request}). Only then will $\operator$ cooperate to provide the required witnesses, leading to the fully signed reset transactions that ensure the VTXOs in $v_{in}$ are not in $C'$ but in $S'$ and the VTXOs in $v_{out}$ are in $F'$. 
    \end{enumerate}
\end{proof}
\begin{theorem}[Operator Balance Security]
\label{thm:operator}
    The Ark protocol $\ark{\operator}$ with batch expiry time $t_e$ and VTXO unilateral delay $\tv$ guarantees that if $\operator$ is honest and $\tv>4k$, $\operator$ can w.o.p. from block height $h+4k+t_e$ retrieve at least the amount it put in the Ark through commitment transactions confirmed by block height $h$, subtracting mining fees.
\end{theorem}
\begin{proof}
    We first state and prove a number of intermediate results.
    \begin{lemma}
    \label{lem:delay}
        The Ark protocol $\ark{\operator}$ with VTXO unilateral delay $\tv$ guarantees, if $\tv>4k$ and $\operator$ is honest, that for every Ark state $\Sigma=(C,F,S)$, every $\vtxo{}\in S$ that appears onchain will be spent by a reset transaction or forfeit transaction, w.o.p.
    \end{lemma}
    \begin{proof}
        Assume without loss of generality that $\vtxo{}$ can be spent by any of its collaborative script paths immediately once posted onchain. Since $\vtxo{}\in S$, an honest operator holds either a valid reset transaction or a valid forfeit transaction (spending an anchor output confirmed onchain) that spends $\vtxo{}$. If the virtual transactions in $\patht(\vtxo{})$ are submitted to $\btc$ at block height $t$, we know by Corollary that by height $t+2k$, they are in $\C{-k}{\operator}$. Routine~\ref{alg:operator} Lines~\ref{rtn:operator-vtxo-spend-start}-\ref{rtn:operator-vtxo-spend-end} instructs the honest operator to post the reset transaction spending $\vtxo{}$ or the forfeit transaction, as soon as $\vtxo{}$ enters $\operator$ local view of the chain. By height $t+2k$, $\vtxo{}$ is definitely in $\C{}{\operator}$, so $\operator$ will for sure have broadcast the required reset/forfeit transaction(s) by height $t+2k$. By Corollary~\ref{cor:2k}, these will be confirmed by height $t+4k$. By Definition~\ref{def:vtxo}, any unilateral spend of $\vtxo{}$ is delayed by at least $\tv$ (this is enforced by an honest operator who only includes correctly formed VTXOs in its commitment transactions). Since $\tv>4k$, there are no possible race conditions and an honest operator is guaranteed to spend $\vtxo{}$ with either the corresponding forfeit or reset transaction. 
    \end{proof}
    \begin{lemma}
    \label{lem:sweep}
        Suppose a batch output expires, and the current Ark state is $\Sigma=(C,F,S)$. If $\tv>4k$, w.o.p., an honest operator can claim all funds locked in the VTXOs contained in the batch that are in $C\cup F\cup S$, i.e., all funds contained in the batch, except those locked in VTXOs that are both unilaterally exited and for which there is no forfeit transaction.
    \end{lemma}
    \begin{proof}
        By Routine~\ref{alg:operator} Line~\ref{rtn:operator-expired}, an honest operator will perform a sweep operation on the expired batch as defined in Routine~\ref{alg:sweep}. If we assume the expired batch output is specified via a VTXT $(V,A)$, every non-leaf transaction in $V$ present onchain will be spent by the operator as the locking script only requires the timelock to be expired and an operator signature. Every leaf transaction has a VTXO locking script, and will therefore not be spent by the operator through the sweep operation. These are exactly the VTXOs that have been unilaterally exited. However, by Routine~\ref{alg:operator} Line~\ref{rtn:operator-spent} and by Lemma~\ref{lem:delay}, since $\tv>4k$, any VTXO that has been spent prior to being put onchain through a unilateral exit, has been spent already w.o.p. by the honest operator via the corresponding reset and/or forfeit transaction, which an honest $\operator$ has by Theorem~\ref{thm:atomicity}. In case of a reset transaction, any of the VTXOs in the subsequent Ark transaction (if valid) that have been spent can also be claimed via a forfeit transaction an honest operator has. In aggregate, this means that the operator claims all funds locked in VTXOs that either have not been spent and not unilaterally exited (which are the VTXOs of the expired batch that are in $C\cup F$), or have been spent (which are the VTXOs in $S$, regardless of whether they have been unilaterally exited). That is, an honest operator can claim all funds held in the expired batch, up to those funds that have been unilaterally exited and that have not been spent, or spent in an Ark transaction included onchain.
    \end{proof}
    \begin{lemma}
    \label{lem:operator}
        Suppose that from block height $h$ onward, an honest operator $\operator$ does no longer accept incoming batch swap and boarding request, only accepting exit requests, in order to shut down its Ark. Then $\ark{\operator}$ with batch expiry time $t_e$ and VTXO unilateral delay $\tv$ guarantees w.o.p. that, if $\tv>4k$, by block height $h+4k+t_e$, $\operator$ has retrieved exactly the value it put into the Ark by funding commitment transactions, subtracting mining fees, and adding at least fees it collected for processing requests.
    \end{lemma}
    
    \begin{proof}
        This proof amounts to keeping track of all the flows of funds, in particular, which funds are and will come under control of the operator, assuming this operator exactly follows the protocol. For simplicity, we assume the connector outputs hold a negligible amount of funds. All groups of funds are greater than or equal to zero.
        
        Fix the block height $h$. We denote the commitment transactions from the operator $\operator$ that are confirmed onchain by $\tx{j}{}$, where $j$ runs through $0,\ldots,k_h,\ldots,\allowbreak k_e,\ldots$. The index $k_h$ refers to the latest commitment transaction confirmed before or at block height $h$, and $k_e$ to the latest commitment transaction confirmed before or at block height $h+2k+t_e$. Furthermore, we denote by $e(j)$ the latest commitment transaction before the expiry of the batches in $\tx{j}{}$. For example, we have $k_e=e(k_h)$.
        
        By construction of any honestly created commitment transaction $\tx{j}{}$ as specified in Routine~\ref{alg:commitment-tx}, we can categorise the following incoming and outgoing groups of funds, based on the transaction inputs and outputs:
        \begin{itemize}[leftmargin=*]
            \item Incoming funds:
            \begin{itemize}
                \item $L_j$: the liquidity coming directly from $\operator$ to finance $\tx{j}{}$.
                \item $B_j$: the funds coming from boarding transactions. 
            \end{itemize}
            \item Outgoing funds:
            \begin{itemize}
                \item $V_j$: the funds locked in batch outputs holding VTXOs.
                \item $U_j$: the funds locked in UTXOs as a consequence of exit requests.
                \item $M_j$: the mining fees.
            \end{itemize}
        \end{itemize}
        We can further categorise the outgoing funds $V_j$ and $U_j$ by tracking which VTXOs or boarding transaction outputs were used as part of the boarding, batch swap, or exit requests that led to the formation of the VTXOs and UTXOs in $\tx{j}{}$. That is, we can write $V_j=\sum_{i\leq j}V_{i,j}$, where $V_{i,j}$ is the aggregate value of all VTXOs or boarding transaction outputs that were spent in order to (partially) fund the VTXOs in $\tx{j}{}$. Notice that $V_{j,j}=B_j$, as all the boarding transaction outputs spent in $\tx{j}{}$ are used to create new VTXOs in $\tx{}{j}$ by Routine~\ref{alg:commitment-tx}, and $\operator$ ensures to not have more value locked in the VTXOs created from a boarding request than the value locked in the corresponding boarding output (Routine~\ref{alg:verify-boarding} Line~\ref{rtn:verify-boarding-funds}). Similarly, we write $U_j=\sum_{i<j}U_{i,j}$ (notice the strict inequality now). 
        
        Since any commitment transaction must be a valid Bitcoin transaction, we must have the following conservation of funds, for every $j\geq0$:    
        \begin{equation*}
            L_j+B_j=V_j+U_j+M_j.
        \end{equation*}
        Recalling that $B_j=V_{j,j}$, we can rewrite the above to:
        \begin{equation}
        \label{eq:conservation-tx}
            L_j=\sum_{i<j}V_{i,j}+\sum_{i<j}U_{i,j}+M_j.
        \end{equation}
        We can also establish a second conservation identity, realising that with an honest operator, a VTXO is either spent as part of a batch swap or exit request, leading to a VTXO or UTXO in a later commitment transaction, or through a unilateral exit, or not spent at all. Note that a VTXO can also be spent by an Ark (reset) transaction. However, the funds in that VTXO will just be relocated in another VTXO, that will either be batch swapped, exited, or not spent at all. Since an honest operator does not allow for spending a VTXO that has already been spent in an Ark transaction (via a reset transaction) (by Routines~\ref{alg:verify-swap} Lines~\ref{rtn:verify-swap-1}-\ref{rtn:verify-swap-4}, \ref{alg:verify-exit} Lines~\ref{rtn:verify-exit-1}-\ref{rtn:verify-exit-4}, \ref{alg:verify-ark-tx-request}) Lines~\ref{rtn:verify-ark-1}-\ref{rtn:verify-ark-4}), we can just distribute these funds over the other categories. Hence, for every $j\geq0$, and because an honest operator will only let a VTXO be spent exactly once (again, by Routines~\ref{alg:verify-swap} Lines~\ref{rtn:verify-swap-1}-\ref{rtn:verify-swap-4}, \ref{alg:verify-exit} Lines~\ref{rtn:verify-exit-1}-\ref{rtn:verify-exit-4}, \ref{alg:verify-ark-tx-request} Lines~\ref{rtn:verify-ark-1}-\ref{rtn:verify-ark-4}), and not allow more funds in the outputs than in the inputs of each batch swap or exit (by Routines~\ref{alg:verify-swap} Line~\ref{rtn:verify-swap-7} and \ref{alg:verify-exit} Line~\ref{rtn:verify-exit-5}), we have:
        \begin{equation}
        \label{eq:conservation-vtxo}
            \sum_{i\leq j}V_{i,j}=\sum_{k>j}V_{j,k}+\sum_{k>j}U_{j,k}+E_j^U+X_j+F_j.
        \end{equation}
        We introduced the term $E_j^U$ to indicate unilateral exits of \emph{unspent} VTXOs or VTXOs spent by Ark transactions (via reset transactions). We explicitly do not account for spent VTXOs that one attempts to unilaterally exit with as well, as this would result in double counting. Also, we introduced $X_j$ to account for all VTXOs that have not been spent. Finally, $F_j$ represents the fees that the operator may charge for processing requests. Indeed, when creating the new outputs in return for batch swap/exits, their aggregate value might be less than $V_j$, where the difference has been claimed by the $\operator$ once the batch expires. 
        
        By Lemma~\ref{lem:sweep}, since $\tv>4k$, an honest operator will be able to claim w.o.p. all the funds held in expired batches, except for the unilaterally exited VTXOs for which $\operator$ does not hold a forfeit transaction. In other words, the amount $S_j^\operator$ swept by the operator is given by:
        \begin{equation}
        \label{eq:sweep}
            S_j^\operator=\sum_{i\leq j}V_{i,j}-E_j^U-E_j^M-S_j^M,
        \end{equation}
        where the last two terms $E_j^M$ are the funds the operator needs to give in mining fees for forfeit transactions, Ark (reset) transactions spending unilaterally exited, spent VTXOs (Routine~\ref{alg:operator} Lines~\ref{rtn:operator-vtxo-spend-start}-\ref{rtn:operator-vtxo-spend-end}), 
        and $S_j^M$ are the funds the operator needs to give in mining fees for sweep transactions, respectively. 
        
        Now, realise that since $\operator$ does not accept boarding or batch swap requests after block height $h$, $S_j^\operator=0$ for $j>k_h$. Indeed, $(\tx{j}{})_{j>k_h}$ will have $B_j=V_j=0$. Only $U_j$ may be non-zero, containing UTXOs from users exiting collaboratively from batches with index less than or equal to $k_h$. More specifically, $U_{i,j}$ may be non-zero for $i\leq k_h$, but $U_{i,j}=0$ for $i>k_h$. Also, realise that $U_{i,j}=0$ for $j>k_e$, as by then all batches confirmed before or at height $h$ have now expired and an honest operator would not perform a collaborative exit for a VTXO coming from an expired batch.
        
        We can thus focus on the commitment transactions confirmed before or at block height $h$. By block height $h+2k+t_e$, every batch contained in one of the commitment transaction $(\tx{j}{})_{j=0}^{k_h}$ will have expired due to Routine~\ref{alg:commitment-tx}, and by Routine~\ref{alg:operator} an honest operator will have initiated sweeps for all expired, unspent funds. The latest sweep transactions will thus be submitted at height $h+2k+t_e$, and will be confirmed by height $h+4k+t_e$ w.o.p. by Corollary~\ref{cor:2k}. Note that as long as a sweep is not confirmed onchain, a user might still try to unilaterally exit funds. In case the user's exit transactions get confirmed instead of $\operator$'s sweep transactions, we count the corresponding funds as a part of $E_j^U$. The total amount gained through commitment transactions before and at block height $h$ is thus:
        \begin{align}
        \label{eq:sum-sweeps}
            \sum_{j\leq k_h}S_j^\operator&\stackrel{\eqref{eq:sweep}}{=}\sum_{j\leq k_h}\Big(\sum_{i\leq j}V_{i,j}-E_j^U-E_j^M-S_j^M\Big)\nonumber\\&\stackrel{\eqref{eq:conservation-vtxo}}{=}\sum_{j\leq k_h}\sum_{k>j}(V_{j,k}+U_{j,k})\nonumber\\&\quad\quad+\sum_{j\leq k_h}(X_j+F_j-E_j^M-S_j^M),
        \end{align}   
        Since $V_k=0$ for $k>k_h$, we rewrite the first term of \eqref{eq:sum-sweeps} as
        \begin{align}
        \label{eq:final-gains}
            \sum_{j\leq k_h}&\sum_{k>j}(V_{j,k}+U_{j,k})\nonumber\\&=\sum_{j\leq k_h}\sum_{j<k\leq k_h}(V_{j,k}+U_{j,k})+\sum_{j\leq k_h}\sum_{k>k_h}U_{j,k}\nonumber\\&=\sum_{k\leq k_h}\sum_{j<k}(V_{j,k}+U_{j,k})+\sum_{j\leq k_h}\sum_{k>k_h}U_{j,k}\nonumber\\&=\sum_{k\leq k_h}\sum_{j<k}(V_{j,k}+U_{j,k})+\sum_{k_h<k\leq k_e}\sum_{j<k}U_{j,k},
        \end{align}
        where the second equality is just a rearrangement of the sums, and the third equality first uses that $U_{j,k}=0$ for $j>k_h$ and for $k>k_e$, and then rearranges the sums.
        
        On the other hand, the funds put in by the operator through commitment transactions before and at block height $h$ add up to 
        \begin{align}
        \label{eq:final-costs}
            \sum_{j\leq k_e}L_j&\stackrel{\eqref{eq:conservation-tx}}{=}\sum_{j\leq k_e}\sum_{i<j}(V_{i,j}+U_{i,j})+\sum_{j\leq k_e}M_j\nonumber\\&=\sum_{j\leq k_h}\sum_{i<j}(V_{i,j}+U_{i,j})+\sum_{k_h<j\leq k_e}\sum_{i<j}U_{i,j}\nonumber\\&\quad\quad+\sum_{j\leq k_e}M_j,
        \end{align}
        where in the second equality, we used that $V_j=0$ for $j>k_h$. Substituting \eqref{eq:final-gains} into \eqref{eq:sum-sweeps}, and then combining the latter with \eqref{eq:final-costs}, we obtain the net balance of the operator at block height $h+4k+t_e$:
        \begin{equation*}
            \sum_{j\leq k_e}S_j^\operator-\sum_{j\leq k_e}L_j=\sum_{j\leq k_h}(X_j+F_j-E_j^M-S_j^M)-\sum_{j\leq k_e}M_j,
        \end{equation*}
        which means that $\operator$ retrieves all funds that it put into the Ark, subtracting the mining fees and adding the fees the operator charges for processing requests. Additionally, the operator may gain any funds from VTXOs that have not been spent in any way before their batch expired, and have therefore just been swept by the operator. This identity then immediately gives us a necessary condition for running an Ark profitably; the operator fees should cover the mining fees (minus any potential unspent VTXOs that have been swept).
    \end{proof}
    Lemma~\ref{lem:operator} now finishes the proof of Theorem~\ref{thm:operator}. Indeed, if an operator wants to retrieve its funds it put in up to and including block height $h$, it can always decide to stop processing boarding and batch swap requests received after $h$, and be guaranteed w.o.p., if $\tv>4k$, to have retrieved from block height $h+4k+t_e$ onward at least the funds it put in via commitment transactions confirmed by block height $h$, subtracting mining fees.
\end{proof}

Henceforth, we assume that there is a set $\Nff$ of users who opted into Ark's fast finality mechanism, and that assumptions \ref{asm:1}-\ref{asm:4} hold. We denote by $\arkff{\operator}$ the fast finality Ark protocol, i.e., the Ark protocol $\ark{\operator}$ extended by Protocol~\ref{pol:ff}, and the requirement that $\operator$ has funded a collateral UTXO of value $c>v$ (given by \ref{asm:4}) which cannot be spent by $\operator$ alone before all VTXOs held by users in $\Nff$ expired, and ensures that each of the VTXOs, as well as the relevant batch outputs and VTXT transaction outputs, allow for private key extraction in case of double-signing, as described in §\ref{sec:extensions}.

Henceforth, we assume that there is a set $\Nff$ of users who opted into Ark's fast finality mechanism, and that assumptions \ref{asm:1}-\ref{asm:4} hold. We denote by $\arkff{\operator}$ the fast finality Ark protocol, i.e., the Ark protocol $\ark{\operator}$ extended by Protocol~\ref{pol:ff}, and the requirement that $\operator$ and the users in $\Nff$ set up a BitVM instance as described in §\ref{subsec:bitvm}.

\begin{definition}[Fast Finality Ark Balance]
\label{def:ff}
    For a given party $P$ in $\Nff$ and for a given Ark state $\arkstate=(C,F,S)$, we define the \emph{Fast Finality Ark balance $b_P^{\arkstate,\ff}$ of $P$ at $\arkstate$} as the sum $b_P^{\arkstate,\ff}:=\sum_{\vtxo{}\in V_P^{\arkstate,\ff}}\vtxo{}.\val$, where $V_P^{\arkstate,\ff}$ is the set of all $\vtxo{}\in F$ which are more than $2k$ blocks from expiry, that $P$ accepted as a payment output, for which only $P$ can provide a witness for a unilateral script path of $\vtxo{}$ and possesses all fully signed transactions in $\patht(\vtxo{})$.
\end{definition}

\begin{theorem}[Fast Finality Balance Security]
\label{thm:ff}
    The fast finality Ark protocol $\arkff{\operator}$ with batch expiry time $t_e$ guarantees that for each honest party $P$ in $\Nff$ and each Ark state $\arkstate=(C,F,S)$, $P$ can claim its fast finality Ark balance $b_P^{\arkstate,\ff}$ onchain, subtracting mining fees, w.o.p.
\end{theorem}
\begin{proof}
    Assume $P$ is honest. Let $\vtxo{}\in V_P^{\arkstate,\ff}$ be arbitrary. Then $\vtxo{}.\val\leq v$ by \ref{asm:4}. Denote by $p=(\tx{}{ark}, \tx{}{re}, \patht(\vtxo{}))$ the payment in which $P$ received $\vtxo{}$. Without loss of generality, assume that $\patht(\vtxo{})$ is a simple sequence of virtual transactions as in Definition~\ref{def:vtxt}, with a number of fast finality Ark transactions appended to it. In particular, $\patht(\vtxo{})$ contains exactly one virtual transaction spending a batch output confirmed onchain. Furthermore, denote by $\tx{}{*}$ a transaction that conflicts with $\vtxo{}$, i.e., there exists a transaction $\tx{}{}\in\patht(\vtxo{})$ that spends the same output as $\tx{}{*}$. In case $\tx{}{*}$ has a VTXO as output, we denote it as $\vtxo{}^*$.

    By definition, $P$ received $\vtxo{}\in V_P^{\arkstate,\ff}$ at least a time period of $2\Delta$ ago. We claim that since $P$ behaves honestly, Protocol~\ref{pol:ff} ensures that no other honest user $P'$ would have a conflicting $\vtxo{}^*$ in its set $V_{P'}^{\arkstate,\ff}$. In other words, no two honest users could be defrauded by a double-sign.

    To see this, recall from Protocol~\ref{pol:ff} that every honest receiver of a payment $\vtxo{}$ finalises the transaction only if no conflicting transaction has been received from another user $2\Delta$ after broadcasting the payment $p$ to $\Nff$. Indeed, if $P$ receives $p$ at time $t$ and broadcasts it immediately, it can be sure that every honest user has received it by $t+\Delta$. We now distinguish two cases.
    \begin{enumerate}
        \item If a conflicting VTXO $\vtxo{}^*$ is received by $P'$ as part of a payment $p'$, after $t+\Delta$, the double-signing will be detected and $P'$ will not accept $p'$, so $\vtxo{}^*\notin V_{P'}^{\arkstate,\ff}$. Moreover, $P'$ will broadcast $p'$ to $\Nff$, and burn the operator collateral. Observe that $P$ may receive $p'$ after $t+2\Delta$, and would have thus accepted $p$ by then.  
        \item If $\vtxo{}^*$ is received by $P'$ before $t+\Delta$, $P'$ also broadcast $p'$ to $\Nff$, meaning $P$ receives it by $t+2\Delta$. At that point, $P$ will not accept the payment (neither will $P'$, in fact). Once again, the operator collateral would have been burned.
    \end{enumerate}
    It is important to realise that at most one honest user can be defrauded. Indeed, we see from the first case that $P$ could have accepted $p$ even though there was a double-sign. This would mean that $P$ does not have a guarantee to unilaterally exit $\vtxo{}$. However, in both cases, $P'$ does not accept $p'$. The gain from double-signing is thus zero, whereas the burned collateral is $c>v\geq\vtxo{}.\val>0$. A rational operator would therefore not double-sign in the first place, and $P$ can be sure that he can unilaterally exit $\vtxo{}$.  

    Now that we have argued how no two honest users could be defrauded by a double-sign, we should consider the case that not an honest, but a malicious user $P^*$ in $\Nff$ holds a transaction $\tx{}{*}$ conflicting with $\tx{}{}\in\patht(\vtxo{})$. There are two possibilities:
    \begin{enumerate}
        \item $\tx{}{*}$ is a virtual transaction in the VTXT. $P^*$ would have to collude with batch cosigners and the operator in order to achieve this, or
        \item $\tx{}{*}$ is an Ark transaction which does not have $\vtxo{}^*$ as an output.
    \end{enumerate}

    First observe that in both cases, since $P$ is monitoring the chain, if $P^*$ tries to exit onchain by publishing $\tx{}{*}$, $P$ will be able to extract $\operator$'s private key from $\tx{}{*}$ and the conflicting $\tx{}{}\in\patht(\vtxo{})$, and burn the collateral. In other words, $P$ is either able to exit unilaterally and claim its fast finality Ark balance (by Corollary~\ref{cor:2k}) onchain, or ends up in a race condition with a conflicting $\tx{}{*}$, which will lead to the collateral being burned.

    Assume that $\tx{}{*}\in\patht(\vtxo{}^*)$\footnote{If this is not the case, we have that $\tx{}{*}$ is a virtual transaction in a malformed VTXT which does not lead to a VTXO output, but some other output that would simply be put onchain, as if it were a VTXO unilateral exit.}, then $P^*$ could have done a (i) batch swap, (ii) cooperative exit, (iii) unilateral exit or (iv) Ark transaction with $\vtxo{}^*$. For (i) and (ii), $\operator$ would have taken over ownership of the funds in $\vtxo{}^*$. For (iii), $P^*$ would still own the funds in $\vtxo{}^*$. Finally, for (iv), an honest user in $\Nff$ would not accept it in case of a fast finality Ark transaction as we argued before, or would only consider it finalised after a batch swap or cooperative exit with $\operator$, which puts us back in case (i) or (ii). In other words, the gain from double-signing is at most $\vtxo{}.\val$. Once again, since the burned collateral is $c>v\geq\vtxo{}.\val$, $\operator$ will not double-sign, and $P$ is able to claim $\vtxo{}$ onchain.

    The only remaining risk for $P$ is that a previous owner along $\patht(\vtxo{})$ exits unilaterally, trying to claim funds already spent offchain through Ark transactions. However, every unilateral path is gated by a relative timelock, allowing $P$ to immediately post the collaborative Ark transactions that spend the contested VTXO (bringing other inputs onchain as needed, which one should note may be costly) and thus win the race. Since $P$ monitors the chain, a previous owner will not be able to claim $P$'s funds. 

    We can repeat this argument for any $\vtxo{}\in V_P^{\arkstate,\ff}$, for any $P\in\Nff$. Since $\sum_{\vtxo{}\in\bigcup_{P\in\Nff}V_P^{\arkstate,\ff}}\vtxo{}.\val=v<c$, we are done. 
\end{proof}

\end{document}